\documentclass[12pt,a4paper]{amsart}
\usepackage{comment}
\usepackage{amsmath,amsfonts,amssymb,extarrows}
\usepackage{mathtools}
\usepackage{color}
\usepackage[hmargin=3cm,vmargin=3cm]{geometry}
\usepackage{hyperref}
\usepackage{nameref,zref-xr}                    
\usepackage[all]{xy}           
\usepackage{cancel}
\usepackage{longtable}
\usepackage{dsfont}
\usepackage{multicol}
\numberwithin{equation}{section}
\newtheorem{theorem}{Theorem}[section]   
\newtheorem{definition}[theorem]{Definition}
\newtheorem{proposition}[theorem]{Proposition}
\newtheorem{lemma}[theorem]{Lemma}
\newtheorem{corollary}[theorem]{Corollary}

\newtheorem{example}[theorem]{Example}
\newtheorem{example-notation}[theorem]{Example-Notation}
\newtheorem{remark}[theorem]{Remark}

\def\g{\mathfrak{g}}

\def\d{\partial}

\def\f{\frac}
\def\dna{d_{\nabla}}

\def\na{\nabla}

\newcommand{\eqa}{\begin{eqnarray}}
\newcommand{\eeqa}{\end{eqnarray}}
\newcommand{\beq}{\begin{equation}}
\newcommand{\eeq}{\end{equation}}

\begin{document}
\title{Integrable hierarchies and F-manifolds with compatible connection}
\author{Paolo Lorenzoni}
\address{P.~Lorenzoni:\newline Dipartimento di Matematica e Applicazioni, Universit\`a di Milano-Bicocca, \newline
Via Roberto Cozzi 55, I-20125 Milano, Italy and INFN sezione di Milano-Bicocca}
\email{paolo.lorenzoni@unimib.it}
\author{Sara Perletti}
\address{S.~Perletti:\newline Dipartimento di Matematica e Applicazioni, Universit\`a di Milano-Bicocca, \newline
Via Roberto Cozzi 55, I-20125 Milano, Italy and INFN sezione di Milano-Bicocca}
\email{sara.perletti1@unimib.it}
\author{Karoline van Gemst}
\address{K.~van Gemst:\newline Dipartimento di Matematica e Applicazioni, Universit\`a di Milano-Bicocca, \newline
Via Roberto Cozzi 55, I-20125 Milano, Italy and INFN sezione di Milano-Bicocca}
\email{karoline.vangemst@unimib.it}

\begin{abstract} 
Building on the interplay between geometry and integrability, we show that F-manifolds with compatible connection $(\nabla,\circ,e)$ are the geometric counterpart of integrable systems of quasilinear first order evolutionary PDEs. We consider F-manifolds equipped  with an Euler vector field and assume that the operator $L=E\circ$ is regular.
 This generalises previous results in the  semisimple context. As an example we study regular  F-manifolds with compatible connection $(\nabla,\circ,e,E)$
  associated with integrable  hierarchies obtained from the solutions of the equation $d\cdot d_L \,a_0=0$ by applying the construction of  \cite{LM}. We show that $n$-dimensional F-manifolds associated to operators $L$ with $r\le n$ Jordan blocks $L_\alpha$ of  size $m_\alpha$ are classified by $n$ arbitrary functions of a single variable, where each block $L_\alpha$ contributes with $m_\alpha$ functions of the  variable  appearing in the  diagonal of the block. In the case of a single Jordan block of arbitrary size we show that flat connections $\nabla$ correspond to linear solutions  $a_0$. This generalises part of the construction of \cite{LP23} where special linear solutions were considered. We illustrate the construction in dimensions $2,3,$ and $4$ for any choice of Jordan canonical form and any choice of the corresponding solution $a_0$. In these dimensions we have that linear solutions define bi-flat F-manifolds, and that  the special linear solutions studied in \cite{LP23} are related to Riemannian F-manifolds with Killing unit vector field. We conjecture that this is true in general.
\end{abstract}    
 
\maketitle

\tableofcontents

\section*{Introduction}
The last few years have witnessed a growing interest in integrable systems of quasilinear first order evolutionary PDEs  (integrable systems of hydrodynamic type),
\newline
\beq\label{SHT-intro}
w^i_{t}=V^i_k(w)w^k_x,
\eeq
\newline
that do not (necessarily) admit Riemann invariants \cite{LM,L2006,KK,KO,XF,FP,VF,BKM3}, and the (non-semisimple) geometric structures connected to them such as:
\begin{itemize}
\item Hamiltonian and bihamiltonian structures of hydrodynamic type and their deformations  \cite{FLS,DVLS};
\item F-manifolds, (bi-)flat F-manifolds  and Dubrovin-Frobenius manifolds \cite{LPR,LP,DH,ALmulti,ABLR,LP22,LP23,AK}.
\newline
\end{itemize}
In this paper, following the approach of \cite{LPR}, we consider systems of  the form
\newline
\beq\label{sht-intro}
w^i_{t}=(X\circ w_x)^i=c^i_{jk}X^jw^k_x,
\eeq
\newline
where $\circ$ is a commutative associative product satisfying the Hertling-Manin condition \cite{HM}.  In this setting, the existence of Riemann invariants is related to the semisimplicity  of the product. In the semisimple case, the integrability of the system can be written  as a  condition on a torsionless connection $\nabla$ uniquely determined by the system \eqref{sht-intro}. This condition relates the Riemann tensor of the connection $\nabla$ with the structure functions of the product.  It takes the two equivalent forms:
\newline
\beq\label{shc-intro}
R^s_{lmi}c^j_{ks}+R^s_{lik}c^j_{ms}+R^s_{lkm}c^j_{is}=0,\qquad R^j_{skl}c^s_{mi}+R^j_{smk}c^s_{li}+R^j_{slm}c^s_{ki}=0\ .
\eeq
\newline
This observation leads the authors  of \cite{LPR}
 to introduce the notion of \emph{F-manifold with compatible connection} $(M,\nabla,\circ,e)$. Here $e$ is the unit of the product and, in order to fix some inessential  freedom, it is required to be covariantly constant \cite{LP}.  If $\nabla$ is flat, the condition \eqref{shc-intro} is automatically satisfied. This selects a special class of F-manifolds called \emph{flat F-manifolds} or, according to  the original definition of Manin \cite{manin}, \emph{F-manifolds with compatible flat connection}.  
\newline
\newline
In this paper we present a construction of non semisimple F-manifolds with compatible connection starting from integrable systems of hydrodynamic type that do not admit Riemann invariants.
\newline
\newline
In the first part of the paper we show that, under some (not too restrictive) technical assumptions, given $n=\text{dim}(M)$ compatible systems of PDEs of the form \eqref{sht-intro} defined by $n$ linearly independent vector fields $X_{(0)},\dots,X_{(n-1)}$, it is possible to construct an F-manifold with compatible connection (with flat unit). The main assumption, called  \emph{regularity}, has been introduced by David and Hertling in \cite{DH} for F-manifolds with Euler vector field. Regularity   ensures the existence of  a special  set of coordinates $(u^1,\dots,u^n)$,  generalising Dubrovin's canonical coordinates in the semisimple setting. Here, the operator of multiplication by the Euler vector field $L=E\circ$ has a block diagonal form, as opposed to diagonal. In the case of $r$ Jordan blocks, we shall denote the blocks by  $L_1,...,L_r$ and their respective sizes by  $m_1,...,m_r$. We will refer also to such generalised coordinates as \emph{canonical coordinates}. By setting $m_1=m_2=\cdots=m_r=1$ we fall in the usual semisimple case and with this choice, the special coordinates do indeed  coincide with Dubrovin's canonical coordinates \cite{D96}. It will be  convenient to re-label the coordinates $(u^1,\dots,u^n)$ according to the following rule. For each $\alpha\in\{2,\dots,r\}$ and for each $j\in\{1,\dots,m_\alpha\}$ we write
\begin{equation}
	\label{relabellingcoordinates}
	j(\alpha)=m_1+\dots+m_{\alpha-1}+j
\end{equation}
(for $\alpha=1$ we set $j(\alpha)=j$), so that $u^{j(\alpha)}$ denotes the $j$-th coordinate associated to the Jordan block with label $\alpha$. Using this notation
 the blocks $\{L_\alpha\}_{\alpha\in\{1,\dots,r\}}$ have the  following lower triangular Toeplitz form 
\begin{equation}\label{Lblock}
L_\alpha=\begin{bmatrix}
	u^{1(\alpha)} & 0 & \dots & 0\cr
	u^{2(\alpha)} & u^{1(\alpha)} & \dots & 0\cr
	\vdots & \ddots & \ddots & \vdots\cr
	u^{m_\alpha(\alpha)} & \dots & u^{2(\alpha)} & u^{1(\alpha)}
\end{bmatrix},\qquad\alpha\in\{1,\dots,r\}.
 \end{equation}  
\newline
\newline
In the second part of the paper, we study F-manifolds with compatible  connection  associated with integrable hierarchies of hydrodynamic type 
\[u_{t_k}=X_{(k)}\circ u_x\]
which are defined by vector fields of the form
\begin{equation}\label{chain-intro}
X_{(k)}=E^k-a_0E^{k-1}-a_1E^{k-2}-\cdots-a_{k-1}E^0,
\end{equation}
where $E^0=e$,  $E^1=E$, $E^2=E\circ E$, $E^3=E\circ E\circ E$, and so on. We assume that the data $(\circ,e,E)$ define a regular F-manifold with Euler vector field, and as such  the  operator $L=E\circ$ has vanishing Nijenhuis torsion and we can apply the results of \cite{LM}. The coefficients $a_1,a_2,\dots, a_n$ are determined recursively, as in \cite{LM}, starting from a solution, $a_0$, of the equation
\beq\label{mainEQ}
d\cdot d_La_0=0,
\eeq
where $d_L$ is the differential associated with $L$. In local coordinates the equation \eqref{mainEQ} reads
\begin{equation}\label{eq:ddainL}
L^s_i\d_s\d_ja_0-L^s_j\d_s\d_ia_0+\d_sa_0\d_jL^s_i-\d_sa_0\d_iL^s_j=0.
\end{equation}
Classifying such hierarchies in the regular case amounts to solving the  equation \eqref{eq:ddainL} for each David-Hertling canonical form of $L$. 
\newline
\newline
In the semisimple case, $L=\text{diag}(u^1,...,u^n)$, the equation \eqref{eq:ddainL} is equivalent to the system
\[\d_i\d_j a_0=0,\qquad i\ne j,\]
for which the general solution is
\[a_0=\sum_{i=1}^nF_i(u^i),\]
where $F_i(u^i)$  are  arbitrary functions of  a single variable. The corresponding integrable systems of hydrodynamic type appeared in \cite{Pav} as  reductions of  the  
 infinite momentum chain
 \[\d_tc_k=\d_xc_{k+1}-c_1\d_xc_k,\qquad k\in\{0,\pm 1,\pm 2,\dots\} \, .\]
In canonical  coordinates, the associated F-manifolds with compatible flat connection and flat unit $(\nabla,\circ,e)$
 are defined by the data
\[\Gamma^i_{ij}=\f{\d_ja_0}{u^i-u^j}=\f{\d_jF_j}{u^i-u^j}=-\Gamma^i_{jj},\quad\Gamma^i_{ii}=-\sum_{s\ne i}\Gamma^i_{is},\qquad c^i_{jk}=\delta^i_j\delta^i_k,\qquad e^i=1,\]
where $\Gamma^i_{ij},\Gamma^i_{jj}$ (with $i\ne j$) and $\Gamma^i_{ii}$ are the
 non-vanishing Christoffel symbols of  $\nabla$,  $c^i_{jk}$ are the structure functions of $\circ$ and $e^i$ are the  components of  the unit vector field. It turns out that the connection $\nabla$ is flat if and only if the function $a_0$ is linear,
\[a_0=\sum_{i=1}^n\varepsilon_iu^i,\]
with 
$\varepsilon_1,...,\varepsilon_n$ being arbitrary constants. The corresponding integrable systems of  hydrodynamic type, especially those corresponding to the choice 
$a_0=\varepsilon\,\text{Tr}(L)$, have been studied by several authors. 
In particular, we refer to \cite{Pav87,ts89} for the case 
$\varepsilon=-1$, to \cite{FP91} for the case $\varepsilon=-\frac{1}{2}$, 
to  \cite{F91} for the case $\varepsilon=1$, and to \cite{Pav,LM,L2006} for the general case. Linear solutions also appeared in the works of Ferapontov on   
  non-local Hamiltonian formalism for diagonal  systems of hydrodynamic type (see for instance \cite{F92}) and in  the  study of a  special class  of bi-flat F-manifolds called Lauricella bi-flat F-manifolds (see  \cite{ALmulti} and references therein).
\newline
\newline
In the non-semisimple case the equation \eqref{mainEQ} looks much more involved.    
  In the case of a single Jordan block of  size  $n$ the general solution is given  by the polynomial in $(u^2,...,u^n)$ 
\newline  
\beq\label{sol-intro}
a_0^{(n)}=F_n(u^1)+\sum_{s>0}\frac{1}{s!}\sum_{k_1=2}^n\cdots\sum_{k_s=2}^n u^{k_1}\cdots u^{k_s} \left(\frac{d}{d u_1}\right)^{(s-1)}F_{n+s-(k_1+k_2+...+k_s)}(u_1),
\eeq
\newline
where $F_{n+s-(k_1+k_2+...+k_s)}=0$ when $n+s-(k_1+k_2+...+k_s)<1$ and 
 $F_1(u^1),\dots,F_n(u^1)$  are arbitrary  functions of a  single variable. Meanwhile, the general solution for $r$ Jordan blocks is the  sum  of the solutions associated to each block. 
\newline
\newline 
It turns out that the first $n$ vector fields of the sequence \eqref{chain-intro} are linearly independent. 
Applying the results of the first part of the paper, starting from the solutions  of  the equation \eqref{mainEQ}, one can construct families of F-manifolds with compatible connection. In particular, we show that linear solutions are in one-to-one correspondence with flat F-manifolds. This is consistent with the results of \cite{LP23}, in which  it was proved that for special linear solutions containing only the ``main''\footnote{The ``main'' variables is the one appearing on the diagonal in each block, i.e. the first variable of each block.} variables of each block
\begin{equation}\label{sls-intro}
		a_0=\overset{r}{\underset{\alpha=1}{\sum}}\,m_\alpha\varepsilon_{\alpha}u^{1(\alpha)}=\overset{r}{\underset{\alpha=1}{\sum}}\,m_\alpha\varepsilon_{\alpha}u^{m_0+m_1+\dots+m_{\alpha-1}+1},
	\end{equation}
 where $r$ denotes the number of Jordan blocks, the corresponding F-manifold  is in fact a bi-flat F-manifold.
\newline
\newline
For illustration, we show the construction in detail in dimensions $2,3,4$ for any possible Jordan block structure. In all cases (including multi-block cases) linear solutions of  \eqref{mainEQ} define bi-flat structures as in the semisimple setting, and as proved in general in the case of a single Jordan block. We conjecture
 that this fact remains true in any dimension. Moreover, we show that special linear solutions  \eqref{sls-intro} correspond to a special class of F-manifolds (Riemannian F-manifolds with Killing unit vector field) introduced in \cite{ABLR}.
\newline
\newline
The paper is organized as follows. Sections 1 and  2 contain  basic  material on integrable  systems of hydrodynamic type and F-manifolds with compatible  connection.
 In Section 3 we show how to construct F-manifolds with compatible connection (and flat unit) starting from an integrable hierarchy of hydrodynamic type. In particular, under some technical assumption, we show that the condition \eqref{shc-intro} is a consequence of integrability (Theorem \ref{mainTh}).  Section 4 is  devoted to the study of the general solution of the equation \eqref{mainEQ} and of the corresponding F-manifolds with compatible connection (whose  existence
  follows from Theorem \ref{mainTh} and Proposition \ref{LinIn}). The main result here is Theorem\eqref{lin=flat} where equivalence between linearity of $a_0$ and flatness of $\nabla$ is proved. In the final part we illustrate the construction in dimensions 2, 3, 4. This leads us (in Section 5) to conjecture that the flat structures associated with linear $a_0$ are indeed bi-flat and the bi-flat structures associated with the special linear solutions \eqref{sls-intro} come from Riemannian F-manifolds with Killing unit vector field.
\newline
\newline
\noindent{\bf Acknowledgements}. The authors are supported by funds of  INFN   (Istituto Nazionale di Fisica Nucleare) by IS-CSN4 Mathematical Methods of Nonlinear Physics. Authors are also thankful to GNFM (Gruppo Nazionale di Fisica Matematica) for supporting activities that contributed to the research reported in this paper. Data sharing not applicable to the present article as no datasets were generated or analyzed during the current study.

\section{Integrable systems of hydrodynamic type admitting Riemann invariants}
Diagonal integrable systems of hydrodynamic type,
\begin{equation}
\label{hts}
u^i_t=v^i(u)u^i_x\qquad i\in\{1,\dots,n\},
\end{equation}
have been studied by Tsarev in \cite{ts91}.  In the strictly hyperbolic case, $v^i\ne v^j$, he proved that integrability of such systems is equivalent to
\begin{equation}
\label{semih}
\partial_j\Gamma^i_{ik}=\partial_k\Gamma^i_{ij},\qquad\forall i\ne j\ne k\ne i,
\end{equation}
where
\begin{equation}\label{CHsymb}
\Gamma^i_{ij}=\f{\d_j v^i}{v^j-v^i},\qquad i\ne j.
\end{equation}

Systems satisfying the integrability condition \eqref{semih} are called semi-Hamiltonian systems or rich systems (see \cite{Serre}). In  particular, Tsarev  proved that a system satisfying \eqref{semih} 
\begin{itemize}
\item possesses infinitely many symmetries (depending on $n$ functions of a single variable)
\begin{equation}
\label{hts2}
u^i_{\tau}=\tilde{v}^i(u)u^i_x,\qquad i\in\{1,\dots,n\},
\end{equation}
obtained by solving the linear system
\beq\label{sym}
\d_j\tilde{v}^i=\Gamma^i_{ij}(\tilde{v}^j-\tilde{v}^i), \qquad i\neq j;
\eeq
\item possesses infinitely many densities of conservation laws 
 obtained by solving the linear system
\begin{equation}
\label{cl}
\d_i\d_j h-\Gamma^i_{ij}\d_i h-\Gamma^j_{ji}\d_j h=0, \qquad i\neq j.
\end{equation}
\end{itemize}
Due  to \eqref{semih}, both systems \eqref{sym} and \eqref{cl} are compatible and in both cases the general solution depends on $n$ arbitrary functions  of a single variable. Moreover, the knowledge of the symmetries allows one to write (locally) any solution, 
\[u(x,t)=(u^1(x,t),\dots,u^n(x,t)),\] 
of \eqref{hts} in implicit form as
\beq\label{ghm}
x+v^i(u^1,\dots,u^n)t=\tilde{v}^i(u^1,\dots,u^n),\qquad i\in\{1,\dots,n\}.
\eeq
Let us now consider a general system of hydrodynamic type
\beq\label{sht}
w^i_t=V^i_j(w)w^j_x,\qquad i\in\{1,\dots,n\}.
\eeq
The existence of  local coordinates (Riemann invariants)  reducing the system \eqref{sht} to  diagonal form is a non-trivial request even assuming  pairwise distinct eigenvalues of  the $(1,1)$-tensor field $V$ at each point. In this case the diagonalisability of the system is equivalent to the vanishing of the Haantjes tensor of $V$ \cite{H}. If this condition is fulfilled, the integrability condition \eqref{shder} is equivalent to the vanishing of a tensor field, called the \emph{semi-Hamiltonian tensor}  \cite{PSS}.
\newline
\newline
The main equations of Tsarev's theory can also be formulated in terms of a family of torsionless connections (see \cite{LPR,ALimrn,LP23}). Indeed, the $n(n-1)$ functions $\Gamma^i_{ij}$ defined by \eqref{CHsymb} can be identified with a subset of the Christoffel symbols of a torsionless connection $\nabla$. The condition \eqref{CHsymb} can be written as 
\beq\label{tsarevconnection}
d_{\nabla}V=0,
\eeq
where $d_{\nabla}$ is usually referred to as exterior covariant derivative of vector-valued differential forms and is defined by
$$(\dna \omega)(X_0, \dots, X_k)=\sum_{i=0}^k (-1)^i \na_{X_i}(\omega(X_0, \dots, \hat{X}_i, \dots, X_k))$$
$$+\sum_{0\leq i<j\leq k}(-1)^{i+j}\omega([X_i, X_j], X_0, \dots, \hat{X}_i, \dots, \hat{X}_j, \dots X_k).$$
Indeed, in Riemann invariants $(u^1,...,u^n)$, where $V={\rm diag}(v^1,...,v^n)$, the above condition is equivalent to $\Gamma^i_{jk}=0$, for pairwise distinct indices, together with \eqref{CHsymb}. Moreover, $\Gamma^i_{ji}=\Gamma^i_{ij}$ due to the vanishing of the torsion of $\nabla$, while all the remaining Christoffel symbols $\Gamma^i_{jj}$ and $\Gamma^i_{ii}$ are free. The integrability of a system of hydrodynamic type  \eqref{sht} defined by a $(1,1)$ tensor field $V$ with distinct eigenvalues and vanishing Haantjes tensor, is equivalent to the request that 
 \beq\label{shc-new}
 d_{\nabla}^2\tilde{V}=0,
 \eeq
 for all connections $\nabla$ satisfying \eqref{tsarevconnection} and for any $(1,1)$-tensor field $\tilde{V}$ commuting with $V$ (see Theorem 2.2 in \cite{LP23} and also \cite{ALimrn}).
  Condition \eqref{shc-new} can be written in terms of the Riemann tensor of $\nabla$ as
\beq\label{shc-new-LC}
[\dna^2\tilde{V}]_{jik}^l =R^l_{sij}\tilde{V}^s_k+R^l_{sjk}\tilde{V}^s_i+R^l_{ski}\tilde{V}^s_j,
\eeq
 where $R$ is the Riemann tensor of $\nabla$:
\[R^k_{lij}=\d_j\Gamma^k_{il}- \d_i\Gamma^k_{jl}+\Gamma^k_{js}\Gamma^s_{il}-\Gamma^k_{is}\Gamma^s_{jl}.\]
In Riemann invariants, the  matrices $\tilde{V}$ are diagonal and the
condition \eqref{shc-new-LC} reads
\beq\label{shc-new-RI}
 [\dna^2\tilde{V}]_{jik}^l =R^l_{kij}\tilde{v}^k+R^l_{ijk}\tilde{v}^i+R^l_{jki}\tilde{v}^j=0,
 \eeq
 for any $(\tilde{v}^1,...,\tilde{v}^n)$. Taking into account the arbitrariness of $\tilde{V}$ and the identity $R^l_{kji}=-R^l_{kij}$,  \eqref{shc-new-RI} reduces to the following pair of conditions
 \[R^i_{ijk}=0,\quad R^k_{jki}=0,\qquad i\ne j\ne k\ne i.\]
The latter condition, known as \emph{Darboux-Tsarev system}, reads
\begin{equation}\label{sh}
\d_i\Gamma^k_{kj}+\Gamma^k_{ki}\Gamma^k_{kj}-\Gamma^k_{kj}\Gamma^j_{ji}
-\Gamma^k_{ik}\Gamma^i_{ij}=0,\qquad\forall i\ne j\ne k\ne i,
\end{equation}
while the former reads
\begin{equation}
\label{shder}
\partial_j\Gamma^i_{ik}=\partial_k\Gamma^i_{ij},\qquad\forall i\ne j\ne k\ne i.
\end{equation}
Clearly, \eqref{shder} follows from \eqref{sh}. Remarkably, if the functions $\Gamma^i_{ij}$ are given by \eqref{CHsymb}, conditions \eqref{semih} and  \eqref{sh} are equivalent due to the following identity proved in \cite{ts91}
\begin{equation}\label{tsarevid}
\d_i\Gamma^k_{kj}+\Gamma^k_{ki}\Gamma^k_{kj}-\Gamma^k_{kj}\Gamma^j_{ji}
-\Gamma^k_{ki}\Gamma^i_{ij}=\f{v^i-v^k}{v^j-v^i}(\d_j\Gamma^k_{ki} -\d_i\Gamma^k_{kj}).
\end{equation} 

The symmetries 
\beq\label{symsht}
w^i_\tau=\tilde{V}^i_j(w)w^j_x,\qquad i\in\{1,\dots,n\},
\eeq
of the system \eqref{sht},  are defined by $(1,1)$-tensor fields $\tilde{V}(w)$ commuting with $V$ and satisfying the condition
\beq\label{symmetries}
d_{\nabla}\tilde{V}=0.
\eeq
It is easy to check that, in Riemann invariants, this condition reduces to the linear system \eqref{sym}. 
\newline
\newline
Any diagonalisable system can be written in the form
\beq
w^i_{t}=(X\circ w_x)^i=c^i_{jk}X^jw^k_x,
\eeq
where $\circ$ is a semisimple commutative associative product and $X$ is a vector field. 
 For such systems, the condition of integrability \eqref{shc-new} reads
\beq\label{shc-new-new}
[\dna^2(X\circ)]_{jik}^l =\left(R^l_{sij}c^s_{kt}+R^l_{sjk}c^s_{it}+R^l_{ski}c^s_{jt}\right)X^t.
\eeq
As a consequence, due to the arbitrariness of $X$, the above condition is equivalent to  condition \eqref{shc-intro}.

\section{F-manifolds with compatible connection} 
\subsection{F-manifolds}
F-manifolds have been introduced by Hertling and Manin in \cite{HM}.
\begin{definition}\label{defFmani}
An \emph{F-manifold} is a manifold $M$ equipped with
\begin{itemize}
\item[(i)] a commutative associative bilinear product  $\circ$  on the module of (local) vector fields, satisfying the following identity:
\begin{align}
&[X\circ Y,W\circ Z]-[X\circ Y, Z]\circ W-[X\circ Y, W]\circ Z\label{HMeq1free}\\
&-X\circ [Y, Z \circ W]+X\circ [Y, Z]\circ W +X\circ [Y, W]\circ Z\notag\\
&-Y\circ [X,Z\circ W]+Y\circ [X,Z]\circ W+Y\circ [X, W]\circ Z=0,\notag
\end{align}
for all local vector fields $X,Y,W, Z$, where $[X,Y]$ is the Lie bracket,
\newline
\item[(ii)] a distinguished vector field $e$ on $M$ such that 
\[e\circ X=X\] 
for all local vector fields $X$. 
\end{itemize}
\end{definition}
Condition \eqref{HMeq1free} is known as the \emph{Hertling-Manin condition}.

\subsection{F-manifolds with compatible connection}
F-manifolds are usually equipped with additional structures. For instance, in \cite{DS}, motivated by Dubrovin's duality, the authors introduced the notion of eventual identity.
\begin{definition}\label{defFwithEvId}
An \emph{F-manifold with an eventual identity} is a manifold $M$ equipped with a vector field $E$ satisfying
\beq
\mathcal{L}_E(\circ)(X,Y)=[e,E]\circ X\circ Y,
\eeq
for arbitrary local vector fields $X$, $Y$.
\end{definition}
The presence of an eventual identity $E$ allows one  to introduce a new product, with unit $E$, defined by
\[X*Y:=(E\circ)^{-1}X\circ Y,\] 
where $X$ and $Y$ are arbitrary local vector fields.

In \cite{ALimrn}, it was proved that if $E$ is an eventual identity, then the Nijenhuis torsion of the operator $V=E\,\circ$ vanishes. That is,
$$N_{V}(X,Y)=[E\circ X,E\circ Y]+E\circ E\circ [X,Y]-E\circ [X,E\circ Y]-E\circ [E\circ X,Y]$$
for all local vector fields $X$, $Y$.

Euler vector fields constitute an example of an eventual identities.  
\begin{definition}\label{defFwithE}
An \emph{F-manifold with Euler vector field} is a manifold $M$ equipped with a vector field $E$ satisfying
$$[e,E]=e,\qquad  \mathcal{L}_E \circ=\circ.$$
\end{definition} 

Following \cite{LPR}, we now introduce the notion of F-manifold with compatible connection (and flat unit).
\begin{definition}\label{Fmnfwcompatconn_flatunit}
An F-manifold with compatible connection and flat unit  is 
 a manifold $M$ equipped with a product 
\[\circ : TM \times TM \rightarrow TM\] 
with structure functions $c^i_{jk}$, a connection $\nabla$ with Christoffel symbols 
$\Gamma^i_{jk}$ and a distinguished vector field $e$ such that
\begin{itemize}
\item[(i)] the one parameter family of  connections $\nabla_{\lambda}$ with Christoffel symbols
$$\Gamma^i_{jk}-\lambda c^i_{jk},$$
is torsionless for any $\lambda$, and the Riemann  tensor of $\nabla_{\lambda}$ coincides
 with the Riemann tensor of $\nabla$
\beq\label{curvlambda}
R_{\lambda}(X,Y)(Z)=R(X,Y)(Z)
\eeq
and satisfies the condition
\beq\label{rc-intri}
Z\circ R(W,Y)(X)+W\circ R(Y,Z)(X)+Y\circ R(Z,W)(X)=0,
\eeq
for all local vector fields $X$, $Y$, $Z$, $W$;
\item[(ii)] $e$ is the unit of the product;
\item[(iii)] $e$ is flat: $\nabla e=0$.
\end{itemize}
\end{definition}
Alternatively, an $F$-manifold with compatible connection can be defined by dropping the requirement for flatness of the unit in Definition \ref{Fmnfwcompatconn_flatunit}.
\newline
\newline
Let's discuss some consequences of condition \eqref{curvlambda}.  
For a given $\lambda$, the torsion and curvature are respectively
\begin{eqnarray*}
T^{(\lambda)k}_{ij}&=&\Gamma^k_{ij}-\Gamma^k_{ji}+\lambda(c^k_{ij}-c^k_{ji}),\\
R^{(\lambda)k}_{ijl}&=&R^k_{ijl}+\lambda(\nabla_i c^k_{jl}-\nabla_j c^k_{il})+\lambda^2(c^k_{im}c^m_{jl}-c^k_{jm}c^m_{il}),
\end{eqnarray*}
where $R^k_{ijl}$ is the Riemann tensor of $\nabla$. Thus, condition (i) is equivalent to
\begin{enumerate}
\item the vanishing of the torsion of $\nabla$;
\item the commutativity of the product  $\circ$;
\item the symmetry  in the lower indices of the tensor field $\nabla_l c^k_{ij}$;
\item the associativity of the product $\circ$.
\end{enumerate}

\begin{remark}	
Condition \eqref{rc-intri} can be written in the equivalent  form  
\begin{equation}\label{rc-intri-2}
R(Y,Z)(X\circ W)+R(X,Y)(Z\circ W)+R(Z,X)(Y\circ W)=0,
\end{equation}
for all local vector fields $X$, $Y$, $Z$, $W$ (see \cite{LPR} for  details).
\end{remark}
In local coordinates conditions \eqref{rc-intri} and \eqref{rc-intri-2} coincide with conditions \eqref{shc-intro}.

\subsection{Flat F-manifolds}
If $\nabla$ is flat then condition \eqref{rc-intri} is automatically satisfied and  condition \eqref{curvlambda} reduces to
\beq\label{zerocurvlambda}
R_{\lambda}(X,Y)(Z)=0.
\eeq
This is the case of flat F-manifolds introduced by Manin in \cite{manin}.
\begin{remark}\label{RemarkFlatness}
	It is useful to observe that the condition of flatness for $\nabla$ can be split into two conditions:
	\begin{enumerate}
		\item Condition \eqref{rc-intri} or condition \eqref{rc-intri-2}, i.e. conditions \eqref{shc-intro} in local coordinates;
		\item 
		\beq\label{flatness}
		R(e,X)(Y)=0,
		\eeq
		for all local vector fields $X$, $Y$.
	\end{enumerate}
	This follows immediately from the identity \eqref{rc-intri-2} evaluated at $Y=e$  (see Lemma 1.11 in \cite{ABLR}).
\end{remark}
\subsection{Riemannian F-manifolds}\label{subsectionRiemannian}
Following \cite{ABLR}, we introduce Riemannian F-manifolds with Killing unit
vector field and recall some useful results in this context. Let $M$ be an $n$-dimensional manifold equipped with a commutative, associative product $\circ$ on the tangent bundle.

\begin{definition}
	A (pseudo-)Riemannian metric $g$ on $M$ is \emph{invariant}, or \emph{compatible with the product $\circ$}, if
	\begin{equation}\label{gcompatc}
		g(X\circ Y,Z)=g(X\circ Z,Y),
	\end{equation}
	for all local vector fields $X,Y,Z$ on $M$.
\end{definition}
\begin{definition}
	A \emph{(pseudo-)Riemannian F-manifold} is the data of an F-manifold $(M,\circ,e)$ with an invariant (pseudo-)Riemannian metric $g$ whose Riemann tensor satisfies \eqref{rc-intri}. Moreover, if $\mathcal{L}_e g=0$ then such a (pseudo-)Riemannian F-manifold is said to possess a \emph{Killing unit vector field}.
\end{definition}
\begin{theorem}(\cite{ABLR})
	Let $(M,\circ,g,e)$ be a Riemannian F-manifold with Killing unit vector field. Then there exists a unique torsionless connection $\nabla$ on $M$ satisfying the condition
	\begin{equation}
		\big(\nabla_X g\big)(Y,Z)=\frac{1}{2}\,d\theta(X\circ Y,Z)+\frac{1}{2}\,d\theta(X\circ Z,Y)
		\label{Riemannian_flatFmnf_bridge}
	\end{equation}
	for all local vector fields $X,Y,Z$, where the $1$-form $\theta(\cdot):=g(e,\cdot)$ is called the \emph{counit}. Moreover, $(M,\circ,e,\nabla)$ is a flat F-manifold.
\end{theorem}
In canonical coordinates, condition \eqref{Riemannian_flatFmnf_bridge} reads
\begin{align}
	\partial_kg_{ij}-\Gamma^s_{ik}g_{sj}-\Gamma^s_{jk}g_{si}=&\frac{1}{2}e^mc^s_{ik}(\partial_sg_{mj}-\partial_jg_{ms})+\frac{1}{2}e^mc^s_{jk}(\partial_sg_{mi}-\partial_ig_{ms})
	\label{Riemannian_flatFmnf_bridge_coords},
\end{align}
for $i,j,k\in\{1,\dots,n\}$.
\begin{theorem}(\cite{ABLR})
	Let $(M,\circ,e,\nabla)$ be a flat F-manifold and let $g$ be an invariant metric satisfying \eqref{Riemannian_flatFmnf_bridge}. Then $(M,\circ,g,e)$ is a Riemannian F-manifold with Killing unit vector field.
\end{theorem}
We point out that in the non-semisimple case,  given a flat F-manifold,  the existence of an invariant metric $g$ satisfying \eqref{Riemannian_flatFmnf_bridge} is not guaranteed (see \cite{ABLR}).

\begin{remark}
The notion of Riemannian F-manifold appears in literature with a slightly different meaning. In \cite{LPR} the  metric $g$ is required to be compatible with $\nabla$ and in \cite{DS} the counit is required to  be closed.
\end{remark} 

\subsection{Bi-flat F-manifolds and Dubrovin-Frobenius manifolds}
Bi-flat F-manifolds are manifolds equipped with two ``compatible'' flat structures. They are defined as follows. 
\begin{definition}[\cite{ALjgp}]
A \emph{bi-flat}  F-manifold is a manifold $M$ equipped with two distinct 
 flat structures $(\nabla,\circ,e)$, $(\nabla^{*},*,E)$ such that
\begin{itemize}
\item[(i)] $E$ is an Euler vector field for the first structure;
\item[(ii)] $X*Y := (E\circ)^{-1}X\circ Y$;
\item[(iii)] $(d_{\nabla}-d_{\nabla^{*}})(X\,\circ)=0,$
\end{itemize}
where $X$ and $Y$ are arbitrary local vector fields and at a generic point the operator $E\circ$ is assumed to be invertible. 
\label{def:biflatF}
\end{definition}
The second flat structure  $(\nabla^{*},*,E)$ is called the \emph{dual structure} and can be thought of as of a generalisation of the notion of (almost-)duality for Dubrovin-Frobenius manifolds  (see \cite{Dad}). The compatibility of the connections $\nabla$, $\nabla^*$ has a natural interpretation in terms of a differential bicomplex associated to the bi-flat structure (see \cite{AL24}).

\begin{remark}
Notice that not all the axioms in Definition \ref{def:biflatF} are independent. For instance, the compatibility between the dual connection and the dual product follows from the other axioms \cite{ALmulti}. Furthermore, the dual connection is only defined at the points where the operator $E\circ$ is invertible. At these points, it immediately follows from the axioms that
\beq\label{dualfromnatural}
\Gamma^{*k}_{ij} =\Gamma^k_{ij}- c^{*l}_{ji}\nabla_l E^k.
\eeq
Moreover, the flatness of the dual connection follows from the linearity of the Euler vector field (see Theorem 4.4 in \cite{ALcomplex} for the semisimple  case and Lemmas 4.2 and 4.3 in \cite{KMS} for the general case).
\end{remark}

\begin{remark}
The manifolds in the above definitions are real or complex manifolds. For the former, all the geometric data are to be smooth.
 For the latter, $TM$ is to be taken as the holomorphic tangent bundle and
all the geometric data are to be holomorphic. 
\end{remark}

Following a non-historical path, Dubrovin-Frobenius manifolds (see \cite{D96}) can be defined as bi-flat  F-manifolds equipped with 
 a metric $\eta$ compatible with the product $\circ$ and the connection $\nabla$:
\[\eta_{il}c^l_{jk}=\eta_{jl}c^l_{ik},\qquad\qquad\nabla\eta=0.\]

\subsection{The regular case}
\begin{definition}[\cite{DH}]\label{DavidHertlingdef}
 An F-manifold with Euler vector field $(M,\circ, e,E)$ is called \emph{regular} if for each $p\in M$ the matrix representing the endomorphism
 $$L_p := E_p\circ : T_pM \to T_pM$$
 has exactly one Jordan block for each distinct eigenvalue.
 \end{definition} 
 In this paper we will use the following important result of \cite{DH}
 regarding the existence of non-semisimple \textit{canonical} coordinates for regular F-manifolds with Euler vector field.
\begin{theorem}[\cite{DH}]\label{DavidHertlingth}
Let $(M, \circ, e, E)$ be a regular F-manifold of dimension $n \geq 2$ with an Euler vector field $E$. Furthermore, assume that locally around a point $p\in M$, the Jordan canonical form of the operator $L$ has $r$ Jordan blocks of sizes $m_1,...,m_r$ with distinct eigenvalues. Then
there exists locally around $p$ a distinguished system of coordinates $\{u^1, \dots, u^{m_1+\dots+m_r}\}$ such that, recalling the notation \eqref{relabellingcoordinates}, \begin{align}
	e^{i(\alpha)}&=\delta^i_1,\qquad
	E^{i(\alpha)}=u^{i(\alpha)},\qquad
	c^{i(\alpha)}_{j(\beta)k(\gamma)}=\delta^\alpha_\beta\delta^\alpha_\gamma\delta^i_{j+k-1},\notag
\end{align}
for all $\alpha,\beta,\gamma\in\{1,\dots,r\}$ and  $i\in\{1,\dots,m_\alpha\}$, $j\in\{1,\dots,m_\beta\}$, ${k\in\{1,\dots,m_\gamma\}}$.
\end{theorem}
This special system of coordinates are called David-Hertling coordinates, and in such coordinates the operator  of multiplication by the Euler vector field has the lower  triangular Toeplitz form \eqref{Lblock} mentioned in the introduction.

\section{F-manifolds with compatible connection and integrable hierarchies of hydrodynamic type}

Let us consider a system of hydrodynamic type
\beq\label{SHS}
w^i_{t}=(X\circ w_x)^i=c^i_{jk}X^jw^k_x,
\eeq
associated with an F-manifold structure. Let $n$ be the dimension of such an F-manifold.

\subsection{The semisimple case}
Assuming that the product is semisimple, it is immediate to check (see \cite{LP}) that there is a unique torsionless  connection satisfying 
\[d_{\nabla}(X\circ)=0,\qquad \nabla e=0.\]
In canonical coordinates it is defined by
\[\Gamma^i_{ij}=\f{\d_jX^i}{X^j-X^i}=-\Gamma^i_{jj},\quad\Gamma^i_{ii}=-\sum_{s\ne i}\Gamma^i_{is},\]
where $\Gamma^i_{ij},\Gamma^i_{jj}$ (with $i\ne j$) and $\Gamma^i_{ii}$ are the
 non-vanishing Christoffel symbols of  $\nabla$. If the system in  \eqref{SHS} is integrable, this connection, together with the product
  $\circ$ with structure functions $c^i_{jk}=\delta^i_j\delta^i_k$ and the vector field $e$ of components $e^i=1$, defines (locally) an F-manifold with compatible  connection and flat unit. Indeed, we  know that  condition \eqref{rc-intri} coincides
 with Tsarev's integrability condition. 
  
\subsection{An example} Let us consider for instance the system of hydrodynamic type $u_{t}=X\circ u_x$ 
defined by 
\[X:=E-a_0e.\] 
In canonical coordinates, $(u^1,...,u^n)$, we have  
\[e^i=1,\qquad E^i=u^i,\qquad c^i_{jk}=\delta^i_j\delta^i_k\]
for all $i \in \{1, \dots, n\}$, and
\[L:=E\circ=\text{diag}(u^1,...,u^n).\]
The general solution of Equation \eqref{mainEQ} is of the form
\begin{equation}
    a_0=\sum_{i=1}^nF_i(u^i).
    \label{eq:a0exss}
\end{equation}
For any choice of the functions $\{F_i\}_{i\in\{1,\dots,n\}}$, the system  satisfies Tsarev's integrability conditions.
 In canonical  coordinates, the associated F-manifolds with compatible flat connection and flat unit $(\nabla,\circ,e)$
 are defined by the data
\[\Gamma^i_{ij}=\f{\d_jF_j}{u^i-u^j}=-\Gamma^i_{jj},\quad\Gamma^i_{ii}=-\sum_{s\ne i}\Gamma^i_{is},\qquad c^i_{jk}=\delta^i_j\delta^i_k,\qquad e^i=1,\]
where $\Gamma^i_{ij},\Gamma^i_{jj}$ (with $i\ne j$) and $\Gamma^i_{ii}$ are the
 non-vanishing Christoffel symbols of  $\nabla$,  $c^i_{jk}$ are the structure functions of $\circ$ and $e^i$ are the  components of  the unit vector field. It is easy to  check  that imposing the additional flatness condition \eqref{flatness} amounts to asking that 
\[ e(\Gamma^i_{jk})=\sum_{s=1}^n\d_s \Gamma^i_{jk}=0,\] 
which, given that $X=E-a_0\,e$ with $a_0$ as in \eqref{eq:a0exss}, can be rewritten as
\[
\partial_j^2 F_j=0,\qquad j\in\{1,\dots,n\}.
\]
This means that $\nabla$ is flat if and only if $a_0$ is a linear function of the canonical coordinates
\[a_0=\sum_{i=1}^n\varepsilon_iu^i.\]
It turns out that the Euler vector field is linear in flat coordinates, ensuring the existence of a second flat structure $(\nabla^*,*,E)$ compatible with $(\nabla,\circ,e)$. 

The function $a_0$ is a flat coordinate of $\nabla$. Besides $a_0$ there are $n-1$ functionally  independent flat coordinates satisfying the following conditions:
\begin{enumerate}
\item $e(f)=0$;
\item $f$ is homogeneous of degree $1-\sum_{i=1}^n \varepsilon_i$;
\item $f$ satisfies the Euler-Poisson-Darboux system 
\beq\label{ddf}
d\cdot d_L f=da_0\wedge df.
\eeq
\end{enumerate}
Functions satisfying the above conditions are called Lauricella functions (see \cite{Lauricella,Looijenga}) and  for this reason the semisimple bi-flat F-manifolds obtained in this way are called Lauricella bi-flat F-manifolds (see \cite{ALmulti}). This construction has been generalised in \cite{LP23} to arbitrary David-Hertling canonical forms.

\begin{theorem}[\cite{LP23}]\label{GT}
	For any choice of $\varepsilon_1,\dots,\varepsilon_r$, there exists a unique regular bi-flat structure $(\nabla,\nabla^*,\circ,*,e,E)$ with canonical coordinates $\{u^1,\dots,u^n\}$ such that $d_{\nabla}\left[(E-a_0\,e)\,\circ\right]=0$, where $r$ is the number of the Jordan blocks (of sizes $m_1,\dots,m_r$) of $E\circ$ and
	\begin{equation}
		a_0=\overset{r}{\underset{\alpha=1}{\sum}}\,m_\alpha\varepsilon_{\alpha}u^{1(\alpha)}
		.\notag
	\end{equation}
\end{theorem}

\subsection{A remark}
The existence of $n$ commuting flows of the form  \eqref{SHS} defined by a frame of vector fields $X_{(0)},\dots,X_{(n-1)}$ implies  that the quantities (indepedent on $\alpha$)  
\[\Gamma^i_{ij}=\f{\d_jX^i_{(\alpha)}}{X^j_{(\alpha)}-X^i_{(\alpha)}},\qquad i\ne j,\]
satisfy Tsarev's conditions of integrability. Indeed,  by straightforward computation one gets (for $i\ne j\ne k\ne i$)
\begin{eqnarray*}
0=\d_k\d_jX^i_{(\alpha)}-\d_j\d_kX^i_{(\alpha)}&=&(\d_j\Gamma^i_{ik}-\d_k\Gamma^i_{ij})X^i_{(\alpha)}\\
&&+(\d_k\Gamma^i_{ij}-\Gamma^i_{ij}\Gamma^j_{jk}-\Gamma^i_{ik}\Gamma^k_{kj}+\Gamma^i_{ij}\Gamma^i_{ik})X^j_{(\alpha)}\\
&&+(-\d_j\Gamma^i_{ik}+\Gamma^i_{ik}\Gamma^k_{kj}+\Gamma^i_{ij}\Gamma^j_{jk}-\Gamma^i_{ij}\Gamma^i_{ik})X^k_{(\alpha)}.
\end{eqnarray*}
Due to the linear independence  of the vector fields $X_{(0)},\dots,X_{(n-1)}$ it follows that the quantities 
 \begin{eqnarray*}
&&(\d_j\Gamma^i_{ik}-\d_k\Gamma^i_{ij})X^i+(\d_k\Gamma^i_{ij}-\Gamma^i_{ij}\Gamma^j_{jk}-\Gamma^i_{ik}\Gamma^k_{kj}+\Gamma^i_{ij}\Gamma^i_{ik})X^j\\
&&+(-\d_j\Gamma^i_{ik}+\Gamma^i_{ik}\Gamma^k_{kj}+\Gamma^i_{ij}\Gamma^j_{jk}-\Gamma^i_{ij}\Gamma^i_{ik})X^k
\end{eqnarray*}
vanish for a generic vector field $X$ and thus the coefficients of $X^i,X^j,X^k$ also vanish.
 We get  in this way conditions  \eqref{sh} and \eqref{shder}. This means that, given $n$ commuting flows of the form  \eqref{SHS} defined by a frame of vector fields $X_{(0)},\dots,X_{(n-1)}$, imposing the conditions
\[d_{\nabla}(X_{(\alpha)}\circ)=0,\qquad \nabla e=0,\qquad\alpha\in\{0,1,\dots,n-1\},\]
one obtains an F-manifold with compatible connection. We will use a similar idea in the general case.

\subsection{The general case}
Let us consider a set of $n$ linearly independent local vector fields $\{X_{(0)},\dots,X_{(n-1)}\}$ on the $n$-dimensional regular F-manifold $(M,\circ,e)$\footnote{By linear independence, we mean linear independence at a generic point.}. Let us assume that the corresponding flows
\begin{align}
	u_{t_i}=X_{(i)}\circ u_x,\qquad i\in\{0,\dots,n-1\},
	\notag
\end{align}
pairwise commute.
\begin{lemma}
	Two $(1,1)$-tensor fields $A$, $B$ define commuting flows
	\begin{equation*}
	    u_{t}= A\,u_x, \qquad \qquad		u_{\tau}= B\,u_x,
	\end{equation*}
	if and only if $[A,B]=0$ and
	\begin{align}
		&\big(\partial_sA^i_j\big)B^s_m+A^i_s\big(\partial_jB^s_m\big)=\big(\partial_sB^i_j\big)A^s_m+B^i_s\big(\partial_jA^s_m\big),\qquad i,j,m\in\{1,\dots,n\}.
		\label{commflows2}
	\end{align}
\end{lemma}
\begin{proof}
	For each $i\in\{1,\dots,n\}$, a straightforward computation yields
	\begin{align}
		u^i_{t\tau}=&\big(\partial_sA^i_j\big)B^s_m\,u^m_x\,u^j_x+A^i_s\big(\partial_jB^s_m\big)u^j_x\,u^m_x+A^i_j\,B^j_m\,u^m_{xx}
		\notag
	\end{align}
	and, analogously,
	\begin{align}
		u^i_{\tau t}=&\big(\partial_sB^i_j\big)A^s_m\,u^m_x\,u^j_x+B^i_s\big(\partial_jA^s_m\big)u^j_x\,u^m_x+B^i_j\,A^j_m\,u^m_{xx}.
		\notag
	\end{align}
	Then, commutativity of the flows $u_{t\tau}=u_{\tau t}$ amounts to
	\begin{equation}
		\begin{cases}
			\big(A^i_jB^j_m-B^i_jA^j_m\big)u^m_{xx}=0,\notag\\
			\big[\big(\partial_sA^i_j\big)B^s_m+A^i_s\big(\partial_jB^s_m\big)-\big(\partial_sB^i_j\big)A^s_m-B^i_s\big(\partial_jA^s_m\big)\big]u^m_xu^j_x=0,\notag
		\end{cases}
	\end{equation}
where $i\in\{1,\dots,n\}$,	concluding the proof.	
\end{proof}
\begin{remark}
	If $[A,B]=0$ then \eqref{commflows2} amounts to
	\begin{align}
		&\big(\partial_sA^i_j-\partial_jA^i_s\big)B^s_m=\big(\partial_sB^i_j-\partial_jB^i_s\big)A^s_m,\qquad i,j,m\in\{1,\dots,n\}.
		\label{commflows2_bis}
	\end{align}
	Moreover, if $d_\nabla A=0$ then
	\begin{align}
		\partial_sA^i_j-\partial_jA^i_s=\Gamma^i_{jt}A^t_s-\Gamma^i_{st}A^t_j,\qquad i,j,s\in\{1,\dots,n\},
		\notag
	\end{align}
	which, together with \eqref{commflows2_bis}, gives
	\begin{align}
		\big(\Gamma^i_{jt}A^t_s-\Gamma^i_{st}A^t_j\big)B^s_m=&\big(\partial_sB^i_j-\partial_jB^i_s\big)A^s_m\notag\\
		=&\big(d_\nabla B\big)^i_{sj}A^s_m+\big(\Gamma^i_{jt}B^t_s-\Gamma^i_{st}B^t_j\big)A^s_m,\qquad i,j,m\in\{1,\dots,n\}.
		\notag
	\end{align}
	Thus, a straightforward computation yields
	\begin{align}
		\big(d_\nabla B\big)^i_{js}A^s_m+\big(d_\nabla B\big)^i_{ms}A^s_j=0,\qquad i,j,m\in\{1,\dots,n\}.
		\label{dnablaBhp}
	\end{align}
	
\end{remark}
\begin{lemma}\label{LemmadnablaBfromtriangular}
	If $A=X\circ$, $B=Y\circ$ satisfy \eqref{dnablaBhp}, with $X^{1(\alpha)}\neq X^{1(\beta)}$ for $\alpha\neq\beta$, $\alpha,\beta\in\{1,\dots,r\}$, and $X^{2(\alpha)}\neq0$ for each ${\alpha\in\{1,\dots,r\}}$, then
	$d_{\nabla} B=0$.
\end{lemma}
\begin{proof}
	For each $\alpha,\beta,\gamma\in\{1,\dots,r\}$ and $i\in\{1,\dots,m_\alpha\}$, $j\in\{1,\dots,m_\beta\}$, $k\in\{1,\dots,m_\gamma\}$ we have
	\begin{align}
		(d_\nabla B)^{i(\alpha)}_{j(\beta)k(\gamma)}= \, & \, \nabla_{j(\beta)}B^{i(\alpha)}_{k(\gamma)}-\nabla_{k(\gamma)}B^{i(\alpha)}_{j(\beta)}\notag\\
		= \, & \, (\nabla_{j(\beta)}Y^{l(\lambda)})\,c^{i(\alpha)}_{l(\lambda)k(\gamma)}+Y^{l(\lambda)}\,\nabla_{j(\beta)}c^{i(\alpha)}_{l(\lambda)k(\gamma)}\notag\\ \, & \, -(\nabla_{k(\gamma)}Y^{l(\lambda)})\,c^{i(\alpha)}_{l(\lambda)j(\beta)}-Y^{l(\lambda)}\,\nabla_{k(\gamma)}c^{i(\alpha)}_{l(\lambda)j(\beta)}\notag\\
		= \, & \, \delta^{\alpha}_{\gamma}\,(\nabla_{j(\beta)}Y^{(i-k+1)(\alpha)})+\delta^{\alpha}_{\gamma}\,Y^{l(\alpha)}\,\nabla_{j(\beta)}c^{i(\alpha)}_{l(\alpha)k(\alpha)}\notag\\ \, & \, -\delta^{\alpha}_{\beta}\,(\nabla_{k(\gamma)}Y^{(i-j+1)(\alpha)})-\delta^{\alpha}_{\beta}\,Y^{l(\alpha)}\,\nabla_{k(\gamma)}c^{i(\alpha)}_{l(\alpha)j(\alpha)},
		\notag
	\end{align}
	which trivially vanishes when $\alpha\notin\{\beta,\gamma\}$. Hence, the only cases to be considered are $\alpha=\beta=\gamma$ and  $\alpha=\gamma\neq\beta$ (or, equivalently, $\alpha=\beta\neq\gamma$).
	
	Let us start by considering the case where $\alpha=\beta=\gamma$. For every ${i,j,k\in\{1,\dots,m_\alpha\}}$, condition \eqref{dnablaBhp} reads
	\begin{equation}
		(d_{\nabla} B)^{i(\alpha)}_{j(\alpha)s(\sigma)}A^{s(\sigma)}_{k(\alpha)}+(d_{\nabla} B)^{i(\alpha)}_{k(\alpha)s(\sigma)}A^{s(\sigma)}_{j(\alpha)}=0.
		\notag
	\end{equation}
	That is,
	\begin{equation}
		(d_{\nabla} B)^{i(\alpha)}_{j(\alpha)s(\alpha)}\,X^{(s-k+1)(\alpha)}+(d_{\nabla} B)^{i(\alpha)}_{k(\alpha)s(\alpha)}\,X^{(s-j+1)(\alpha)}=0,
		\label{dnablaB_proof_oneblock}
	\end{equation}
	as $A^{s(\sigma)}_{h(\alpha)}=\delta^\sigma_\alpha\,X^{(s-h+1)(\alpha)}$ for all $\sigma\in\{1,\dots,r\}$ and $h\in\{1,\dots,m_\alpha\}$, $s\in\{1,\dots,m_\sigma\}$. For $k=m_\alpha$ in \eqref{dnablaB_proof_oneblock} we get that
	\begin{equation}
		(d_{\nabla} B)^{i(\alpha)}_{j(\alpha)m_\alpha(\alpha)}\,X^{1(\alpha)}+(d_{\nabla} B)^{i(\alpha)}_{m_\alpha(\alpha)s(\alpha)}\,X^{(s-j+1)(\alpha)}=0.
		\notag
	\end{equation}
	In particular, the second term is a sum over $s\geq j$, where the term corresponding to $s=j$ cancels with $(d_{\nabla} B)^{i(\alpha)}_{j(\alpha)m_\alpha(\alpha)}\,X^{1(\alpha)}$, yielding
	\begin{equation}
		\overset{m_\alpha}{\underset{s=j+1}{\sum}}\,(d_{\nabla} B)^{i(\alpha)}_{m_\alpha(\alpha)s(\alpha)}\,X^{(s-j+1)(\alpha)}=0.
		\label{dnablaB_proof_oneblock_kmax}
	\end{equation}
	Setting $j=m_\alpha-1$ in \eqref{dnablaB_proof_oneblock_kmax} gives
	\begin{equation}
		(d_{\nabla} B)^{i(\alpha)}_{m_\alpha(\alpha)m_\alpha(\alpha)}\,X^{2(\alpha)}=0,
		\notag
	\end{equation}
	which implies $(d_{\nabla} B)^{i(\alpha)}_{m_\alpha(\alpha)m_\alpha(\alpha)}=0$ for every $i\in\{1,\dots,m_\alpha\}$. Let's now fix an integer $h\in\{2,\dots,m_\alpha-1\}$ and inductively assume that $(d_{\nabla} B)^{i(\alpha)}_{m_\alpha(\alpha)l(\alpha)}=0$ for every $i\in\{1,\dots,m_\alpha\}$ for each $l\in\{h+1,\dots,m_\alpha\}$. Setting $j=h-1$ in \eqref{dnablaB_proof_oneblock_kmax} we get
	\begin{equation}
		\overset{m_\alpha}{\underset{s=h}{\sum}}\,(d_{\nabla} B)^{i(\alpha)}_{m_\alpha(\alpha)s(\alpha)}\,X^{(s-h+2)(\alpha)}=0,
		\notag
	\end{equation}
	where, by the inductive assumption, only the term corresponding to $s=h$ survives, yielding $(d_{\nabla} B)^{i(\alpha)}_{m_\alpha(\alpha)h(\alpha)}=0$. This proves that $(d_{\nabla} B)^{i(\alpha)}_{m_\alpha(\alpha)j(\alpha)}=0$ for every ${i\in\{1,\dots,m_\alpha\}}$, for each $j\in\{2,\dots,m_\alpha\}$.	Let us fix $m\in\{2,...,m_\alpha-1\}$ and inductively assume that $(d_{\nabla} B)^{i(\alpha)}_{k(\alpha)j(\alpha)}=0$ for every $i\in\{1,\dots,m_\alpha\}$, $j\in\{2,\dots,m_\alpha\}$, for each $k\in\{m+1,\dots,m_\alpha\}$. We need to show that $(d_{\nabla} B)^{i(\alpha)}_{m(\alpha)j(\alpha)}=0$ for every $i\in\{1,\dots,m_\alpha\}$, $j\in\{2,\dots,m_\alpha\}$. For $k=m$ in \eqref{dnablaB_proof_oneblock} we get
	\begin{equation}
		(d_{\nabla} B)^{i(\alpha)}_{j(\alpha)s(\alpha)}\,X^{(s-m+1)(\alpha)}+(d_{\nabla} B)^{i(\alpha)}_{m(\alpha)s(\alpha)}\,X^{(s-j+1)(\alpha)}=0,
		\notag
	\end{equation}
	where the first terms in the sums cancel out with each other, yielding
	\begin{equation}
		\overset{m_\alpha}{\underset{s=m+1}{\sum}}\,(d_{\nabla} B)^{i(\alpha)}_{j(\alpha)s(\alpha)}\,X^{(s-m+1)(\alpha)}+\overset{m_\alpha}{\underset{s=j+1}{\sum}}\,(d_{\nabla} B)^{i(\alpha)}_{m(\alpha)s(\alpha)}\,X^{(s-j+1)(\alpha)}=0.
		\notag
	\end{equation}
	By the inductive assumption, only the second sum survives, giving
	\begin{equation}
		\overset{m_\alpha}{\underset{s=j+1}{\sum}}\,(d_{\nabla} B)^{i(\alpha)}_{m(\alpha)s(\alpha)}\,X^{(s-j+1)(\alpha)}=0,
		\label{dnablaB_proof_oneblock_km}
	\end{equation}
	which, again by the inductive assumption, trivially holds whenever $j\geq m$. For $j=m-1$ in \eqref{dnablaB_proof_oneblock_km} we get
	\begin{equation}
		\overset{m_\alpha}{\underset{s=m}{\sum}}\,(d_{\nabla} B)^{i(\alpha)}_{m(\alpha)s(\alpha)}\,X^{(s-m+2)(\alpha)}=0,
		\notag
	\end{equation}
	where, by the inductive assumption, only the term corresponding to $s=m$ survives, yielding $(d_{\nabla} B)^{i(\alpha)}_{m(\alpha)m(\alpha)}=0$ for every $i\in\{1,\dots,m_\alpha\}$. Let us fix $h\in\{2,\dots,m_\alpha-1\}$ and inductively assume that $(d_{\nabla} B)^{i(\alpha)}_{m(\alpha)l(\alpha)}=0$ for every $i\in\{1,\dots,m_\alpha\}$, for each $l\in\{h+1,\dots,m_\alpha\}$. For $j=h-1$ in \eqref{dnablaB_proof_oneblock_km}
	\begin{equation}
		\overset{m_\alpha}{\underset{s=h}{\sum}}\,(d_{\nabla} B)^{i(\alpha)}_{m(\alpha)s(\alpha)}\,X^{(s-h+2)(\alpha)}=0,
		\notag
	\end{equation}
	which, by the inductive assumption, gives $(d_{\nabla} B)^{i(\alpha)}_{m(\alpha)h(\alpha)}=0$ for each $i\in\{1,\dots,m_\alpha\}$. This proves that $(d_{\nabla} B)^{i(\alpha)}_{m(\alpha)j(\alpha)}=0$ for every $i\in\{1,\dots,m_\alpha\}$, for each $j\in\{2,\dots,m_\alpha\}$.	Thus, we have proved that $(d_{\nabla} B)^{i(\alpha)}_{k(\alpha)j(\alpha)}=0$ for all $i\in\{1,\dots,m_\alpha\}$ and for all ${j,k\in\{2,\dots,m_\alpha\}}$, and we are left with showing that $(d_{\nabla} B)^{i(\alpha)}_{j(\alpha)k(\alpha)}=0$ when $k=1$ or $j=1$. For this purpose, let us now consider $j=1$ in \eqref{dnablaB_proof_oneblock}. This gives
	\begin{equation}
		\overset{m_\alpha}{\underset{s=k+1}{\sum}}\,(d_{\nabla} B)^{i(\alpha)}_{1(\alpha)s(\alpha)}\,X^{(s-k+1)(\alpha)}+\overset{m_\alpha}{\underset{s=2}{\sum}}\,(d_{\nabla} B)^{i(\alpha)}_{k(\alpha)s(\alpha)}\,X^{s(\alpha)}=0,
		\label{dnablaB_proof_oneblock_j1}
	\end{equation}
	which for $k=m_\alpha-1$ becomes
	\begin{equation}
		(d_{\nabla} B)^{i(\alpha)}_{1(\alpha)m_\alpha(\alpha)}\,X^{(s-m_\alpha+2)(\alpha)}+\overset{m_\alpha}{\underset{s=2}{\sum}}\,(d_{\nabla} B)^{i(\alpha)}_{(m_\alpha-1)(\alpha)s(\alpha)}\,X^{s(\alpha)}=0,
		\notag
	\end{equation}
	implying $(d_{\nabla} B)^{i(\alpha)}_{1(\alpha)m_\alpha(\alpha)}=0$ by what we showed above. Let us fix ${h\in\{3,\dots,m_\alpha-1\}}$ and inductively assume that $(d_{\nabla} B)^{i(\alpha)}_{1(\alpha)k(\alpha)}=0$ for each $k\in\{h+1,\dots,m_\alpha\}$. For $k=h-1$ in \eqref{dnablaB_proof_oneblock_j1} we get
	\begin{equation}
		\overset{m_\alpha}{\underset{s=h}{\sum}}\,(d_{\nabla} B)^{i(\alpha)}_{1(\alpha)s(\alpha)}\,X^{(s-h+2)(\alpha)}+\overset{m_\alpha}{\underset{s=2}{\sum}}\,(d_{\nabla} B)^{i(\alpha)}_{(h-1)(\alpha)s(\alpha)}\,X^{s(\alpha)}=0,
		\notag
	\end{equation}
	implying $(d_{\nabla} B)^{i(\alpha)}_{1(\alpha)h(\alpha)}=0$. This proves that
	\begin{equation}
		(d_{\nabla} B)^{i(\alpha)}_{1(\alpha)k(\alpha)}=0,\qquad k\in\{3,\dots,m_\alpha\}.
		\label{dnablaB_proof_oneblock_jequalk_laststepbutone}
	\end{equation}
	Finally, consider \eqref{dnablaB_proof_oneblock} with $k=j$. This yields
	\begin{equation}
		(d_{\nabla} B)^{i(\alpha)}_{j(\alpha)s(\alpha)}\,X^{(s-j+1)(\alpha)}=0,
	\end{equation}
	which for $j=1$ becomes
	\begin{equation}
		\overset{m_\alpha}{\underset{s=2}{\sum}}\,(d_{\nabla} B)^{i(\alpha)}_{1(\alpha)s(\alpha)}\,X^{s(\alpha)}=0,
		\notag
	\end{equation}
	giving $(d_{\nabla} B)^{i(\alpha)}_{1(\alpha)2(\alpha)}\,X^{2(\alpha)}=-\overset{m_\alpha}{\underset{s=3}{\sum}}\,(d_{\nabla} B)^{i(\alpha)}_{1(\alpha)s(\alpha)}\,X^{s(\alpha)}=0$ by \eqref{dnablaB_proof_oneblock_jequalk_laststepbutone}. This completes the proof in the case of $\alpha=\beta=\gamma$.
	
	Let us now turn to the case of $\alpha=\gamma\neq\beta$. For every ${i,k\in\{1,\dots,m_\alpha\}}$ and every ${j\in\{1,\dots,m_\beta\}}$,  \eqref{dnablaBhp} reads
	\begin{equation}
		(d_{\nabla} B)^{i(\alpha)}_{j(\beta)s(\sigma)}A^{s(\sigma)}_{k(\alpha)}+(d_{\nabla} B)^{i(\alpha)}_{k(\alpha)s(\sigma)}A^{s(\sigma)}_{j(\beta)}=0,
		\notag
	\end{equation}
	giving
	\begin{equation}
		(d_{\nabla} B)^{i(\alpha)}_{j(\beta)s(\alpha)}X^{(s-k+1)(\alpha)}+(d_{\nabla} B)^{i(\alpha)}_{k(\alpha)s(\beta)}X^{(s-j+1)(\beta)}=0,
		\label{dnablaB_proof_manyblocks}
	\end{equation}
	which for $j=m_\beta$ in \eqref{dnablaB_proof_manyblocks} gives
	\begin{equation}
		(d_{\nabla} B)^{i(\alpha)}_{m_\beta(\beta)k(\alpha)}(X^{1(\alpha)}-X^{1(\beta)})+\overset{m_\alpha}{\underset{s=k+1}{\sum}}\,(d_{\nabla} B)^{i(\alpha)}_{m_\beta(\beta)s(\alpha)}X^{(s-k+1)(\alpha)}=0.
		\label{dnablaB_proof_manyblocks_jtop}
	\end{equation}
	Setting $k=m_\alpha$ in\eqref{dnablaB_proof_manyblocks_jtop} we get
	\begin{equation}
		(d_{\nabla} B)^{i(\alpha)}_{m_\beta(\beta)m_\alpha(\alpha)}(X^{1(\alpha)}-X^{1(\beta)})=0,
		\notag
	\end{equation}
	yielding $(d_{\nabla} B)^{i(\alpha)}_{m_\beta(\beta)m_\alpha(\alpha)}=0$. Let us fix $H\in\{1,\dots,m_\alpha-1\}$ and inductively assume that $(d_{\nabla} B)^{i(\alpha)}_{m_\beta(\beta)k(\alpha)}=0$ for each $k\in\{H+1,\dots,m_\alpha\}$. For $k=H$ in \eqref{dnablaB_proof_manyblocks_jtop} we get
	\begin{equation}
		(d_{\nabla} B)^{i(\alpha)}_{m_\beta(\beta)H(\alpha)}(X^{1(\alpha)}-X^{1(\beta)})+\overset{m_\alpha}{\underset{s=H+1}{\sum}}\,(d_{\nabla} B)^{i(\alpha)}_{m_\beta(\beta)s(\alpha)}X^{(s-H+1)(\alpha)}=0,
		\notag
	\end{equation}
	which implies $(d_{\nabla} B)^{i(\alpha)}_{m_\beta(\beta)H(\alpha)}=0$ by the inductive assumption. This proves that $(d_{\nabla} B)^{i(\alpha)}_{m_\beta(\beta)k(\alpha)}=0$ for every $i,k\in\{1,\dots,m_\alpha\}$.	Let us fix $h\in\{1,\dots,m_\beta-1\}$ and inductively assume that $(d_{\nabla} B)^{i(\alpha)}_{j(\beta)k(\alpha)}=0$ for every $i,k\in\{1,\dots,m_\alpha\}$, for each $j\in\{h+1,\dots,m_\beta\}$. For $j=h$ in \eqref{dnablaB_proof_manyblocks} we get
	\begin{equation}
		(d_{\nabla} B)^{i(\alpha)}_{h(\beta)s(\alpha)}X^{(s-k+1)(\alpha)}+(d_{\nabla} B)^{i(\alpha)}_{k(\alpha)s(\beta)}X^{(s-h+1)(\beta)}=0,
		\notag
	\end{equation}
	which, by the inductive assumption, amounts to
	\begin{equation}
		(d_{\nabla} B)^{i(\alpha)}_{h(\beta)k(\alpha)}(X^{1(\alpha)}-X^{1(\beta)})+
		\overset{m_\alpha}{\underset{s=k+1}{\sum}}\,(d_{\nabla} B)^{i(\alpha)}_{h(\beta)s(\alpha)}X^{(s-k+1)(\alpha)}=0.
		\label{dnablaB_proof_manyblocks_jh}
	\end{equation}
	For $k=m_\alpha$, \eqref{dnablaB_proof_manyblocks_jh} becomes
	\begin{equation}
		(d_{\nabla} B)^{i(\alpha)}_{h(\beta)m_\alpha(\alpha)}(X^{1(\alpha)}-X^{1(\beta)})=0,
		\notag
	\end{equation}
	yielding $(d_{\nabla} B)^{i(\alpha)}_{h(\beta)m_\alpha(\alpha)}=0$. Let us fix $H\in\{1,\dots,m_\alpha-1\}$ and inductively assume that $(d_{\nabla} B)^{i(\alpha)}_{h(\beta)k(\alpha)}=0$ for each $k\in\{H+1,\dots,m_\alpha\}$. For $k=H$ in \eqref{dnablaB_proof_manyblocks_jh} we get
	\begin{equation}
		(d_{\nabla} B)^{i(\alpha)}_{h(\beta)H(\alpha)}(X^{1(\alpha)}-X^{1(\beta)})+
		\overset{m_\alpha}{\underset{s=H+1}{\sum}}\,(d_{\nabla} B)^{i(\alpha)}_{h(\beta)s(\alpha)}X^{(s-H+1)(\alpha)}=0,
		\notag
	\end{equation}
	which implies $(d_{\nabla} B)^{i(\alpha)}_{h(\beta)H(\alpha)}=0$ by means of the inductive assumption. Hence, ${(d_{\nabla} B)^{i(\alpha)}_{h(\beta)H(\alpha)}=0}$ for every $i,k\in\{1,\dots,m_\alpha\}$ and every $j\in\{1,\dots,m_\beta\}$.
\end{proof}
$\,$\newline
In particular, if we consider $X\in\{X_{(0)},\dots,X_{(n-1)}\}$ such that $X^{1(\alpha)}\neq X^{1(\beta)}$ for ${\alpha\neq\beta}$, $\alpha,\beta\in\{1,\dots,r\}$, and $X^{2(\alpha)}\neq0$ for each ${\alpha\in\{1,\dots,r\}}$, then imposing $d_\nabla(X\circ)=0$ is sufficient to guarantee that also $d_\nabla(Y\circ)=0$ for any other ${Y\in\{X_{(0)},\dots,X_{(n-1)}\}}$.Thus, we shall assume that the family $\{X_{(0)},\dots,X_{(n-1)}\}$ contains at least one local vector field $X$ such that $X^{1(\alpha)}\neq X^{1(\beta)}$ for ${\alpha\neq\beta}$, $\alpha,\beta\in\{1,\dots,r\}$, and $X^{2(\alpha)}\neq0$ for each ${\alpha\in\{1,\dots,r\}}$, satisfying $d_\nabla(X\circ)=0$.
\newline
\newline
The Hertling-Manin condition \eqref{HMeq1free} in coordinates takes the form
         \begin{equation}
             (\partial_s c^k_{im})c^s_{jl} - (\partial_s c^k_{jl})c^s_{im} + (\partial_i c^s_{jl})c^k_{sm} + (\partial_m c^s_{jl})c^k_{si} - (\partial_l c^s_{im})c^k_{js} - (\partial_j c^s_{im})c^k_{ls} = 0.
             \label{eq:HMcoords}
         \end{equation}
It turns out that this condition can be equivalently phrased in terms of covariant derivatives, as the following lemma shows.
 \begin{lemma}
 	For all suitable indices, the left-hand side of \eqref{eq:HMcoords} is equal to
         \begin{equation}\notag
              (\nabla_s c^k_{im})c^s_{jl} - (\nabla_s c^k_{jl})c^s_{im} + (\nabla_i c^s_{jl})c^k_{sm} + (\nabla_m c^s_{jl})c^k_{si} - (\nabla_l c^s_{im})c^k_{js} - (\nabla_j c^s_{im})c^k_{ls}.
            \end{equation}
            \label{lemma:HMinnabla}
         \end{lemma}
         \begin{proof}
             As $\nabla_s c^k_{im} = \partial_s c^k_{im} + \Gamma^k_{ds}c^d_{im} - \Gamma^d_{is}c^k_{dm} - \Gamma^d_{ms}c^k_{id}$, we have that Lemma \ref{lemma:HMinnabla} is equivalent to the vanishing of
             \begin{equation}
             \begin{aligned}
             \, & \,   \left( \Gamma^k_{ds}c^d_{im} -  \Gamma^d_{is}c^k_{dm}- \Gamma^d_{ms}c^k_{id}\right)c^s_{jl} - \left( \Gamma^k_{ds}c^d_{jl} -  \Gamma^d_{js}c^k_{dl}- \Gamma^d_{ls}c^k_{jd}\right)c^s_{im}\\
             + \, & \, \left( \Gamma^s_{di}c^d_{jl} -  \Gamma^d_{ji}c^s_{dl}- \Gamma^d_{li}c^s_{jd}\right)c^k_{sm} + \left( \Gamma^s_{dm}c^d_{jl} -  \Gamma^d_{jm}c^s_{dl}- \Gamma^d_{lm}c^s_{jd}\right)c^k_{si} \\
             - \, & \, \left( \Gamma^s_{dl}c^d_{im} -  \Gamma^d_{il}c^s_{dm}- \Gamma^d_{ml}c^s_{id}\right)c^k_{js} - \left( \Gamma^s_{dj}c^d_{im} -  \Gamma^d_{ij}c^s_{dm}- \Gamma^d_{mj}c^s_{id}\right)c^k_{ls}.
               \end{aligned}
               \label{HMnablaextra}
             \end{equation}
             Let the $j^\text{th}$ term in the $i^\text{th}$ bracket in \eqref{HMnablaextra} be labelled $T_{ij}$, i.e. $T_{12} = -\Gamma^d_{is} c^k_{dm}c^s_{jl}$. We then obtain the following cancellations:
\begin{multicols}{2}
 \begin{itemize}
                 \item $T_{11} + T_{21} = 0$,
                  \item $T_{12} + T_{31} = 0$,
                   \item $T_{13} + T_{41} = 0$,
                    \item $T_{22} + T_{61} = 0$,
                     \item $T_{23} + T_{51} = 0$,
             \end{itemize} \begin{itemize}
                 \item $T_{32} + T_{62} = 0$,
                   \item $T_{33} + T_{52} = 0$,
                     \item $T_{42} + T_{63} = 0$,
                       \item $T_{43} + T_{53} = 0$.
             \end{itemize}
\end{multicols}
             Here, the relations in left column are obtained after relabelling the summed variables appropriately, while the right column follows from the associativity of the product
             \begin{equation*}
                 c^m_{ij} c^d_{mk} = c^m_{jk} c^d_{im},
             \end{equation*}
             commutativity ($c^k_{ij} = c^k_{ji}$), and the fact that $\nabla$ is torsionless ($\Gamma^i_{jk} = \Gamma^i_{kj}$).
             
            Consequently,  $\eqref{HMnablaextra}=0$, as required. 
         \end{proof}

 \begin{lemma}
  Let $X$ be a local vector field, then
         \begin{equation}
             (\nabla_e c^i_{js})X^s = 0.
         \end{equation}
         \label{lemma:nablaeczero}
     \end{lemma}
     \begin{proof}        
Using Lemma \ref{lemma:HMinnabla}, we have that
\begin{equation}
     (\nabla_s c^k_{im})c^s_{jl} - (\nabla_s c^k_{jl})c^s_{im} + (\nabla_i c^s_{jl})c^k_{sm} + (\nabla_m c^s_{jl})c^k_{si} - (\nabla_l c^s_{im})c^k_{js} - (\nabla_j c^s_{im})c^k_{ls} = 0.
\end{equation}
Now, let's contract with the unit vector field $e = e^l \partial_l$. That is, 
         \begin{equation}
             \begin{aligned}
               0 = \, & \,    e^l\bigg((\nabla_s c^k_{im})c^s_{jl} - (\nabla_s c^k_{jl})c^s_{im} + (\nabla_i c^s_{jl})c^k_{sm} + (\nabla_m c^s_{jl})c^k_{si} - (\nabla_l c^s_{im})c^k_{js} - (\nabla_j c^s_{im})c^k_{ls} \bigg) \\
                 = \, & \,  \bcancel{\nabla_j c^k_{im}} - \overset{= 0}{\overbrace{(\nabla_s \delta^k_j)}} \, c^s_{im} + \overset{= 0}{\overbrace{(\nabla_i \delta^s_j)}} \, c^k_{sm} +  \overset{= 0}{\overbrace{(\nabla_m \delta^s_j)}} \, c^k_{si} - (\nabla_e c^s_{im})c^k_{js} - \bcancel{\nabla_j c^k_{im}}   \\
                  \implies \, & \, (\nabla_e c^s_{im})c^k_{js} = 0,
             \end{aligned}
         \end{equation}
         where we have used the flatness of the unit vector field, and the fact that $e^lc^s_{jl} = \delta^s_j$. Contracting again with $e = e^j\partial_j$ gives
         \begin{equation}
             0 =( \nabla_e c^s_{im})\delta^k_s = \nabla_e c^i_{im}.
         \end{equation}
     \end{proof}
         
\begin{theorem}
    Let $X$ be a local vector field satisfying  
    \begin{equation}
\text{d}_\nabla(X \circ) = 0.
        \label{eq:dnablaXcirc}
    \end{equation} 
    Then, 
    \begin{equation}
        (R^k_{lmi}c^n_{pk} + R^k_{lip}c^n_{mk} + R^k_{lpm}c^n_{ik})X^l = 0.
        \label{eq: 3RCX}
    \end{equation}
    \label{thm:dnablezeroimplies3RC}
\end{theorem}
\begin{proof}
In coordinates \eqref{eq:dnablaXcirc} takes the form
    \begin{equation}
        (\nabla_j c^i_{ks})X^s + c^i_{ks} \nabla_j X^s = (\nabla_k c^i_{js})X^s + c^i_{js} \nabla_k X^s.
        \label{eq:dnablaXcirccoords}
    \end{equation}
    Contracting with the unit vector field $e = e^k \partial_k$ gives
    \begin{equation}
    \begin{aligned}\notag
          e^k (\nabla_j c^i_{ks})X^s + e^k c^i_{ks} \nabla_j X^s \, & \, = e^k (\nabla_k c^i_{js})X^s + e^k c^i_{js} \nabla_k X^s,
    \end{aligned}
    \end{equation}
	implying
	\begin{equation}
		\begin{aligned}
			\nabla_j X^i \, & \, = (\nabla_e c^i_{js})X^s + c^i_{js}\nabla_e X^s,
		\end{aligned}
		\label{eq:dnablaXcirccoordscontracted}
	\end{equation}
    as \begin{equation*}
        e^k(\nabla_j c^i_{ks}) = \nabla_j (e^k c^i_{ks}) = \nabla_j \delta^i_s = 0,
    \end{equation*}
    \begin{equation*}
        e^k(c^i_{ks}\nabla_jX^s) = \delta^i_s \nabla_j X^s = \nabla_j X^i,
    \end{equation*}
    and 
    \begin{equation*}
        e^k \nabla_k = \nabla_e,
    \end{equation*}
    where we have used the flatness of the unit, and the fact that $e^k c^i_{ks} = \delta^i_s$.  
    
     Thus, by Lemma \ref{lemma:nablaeczero}, \eqref{eq:dnablaXcirccoordscontracted} becomes
     \begin{equation}
         \nabla_jX^i = c^i_{js} \nabla_e X^s. 
     \end{equation}
     By contracting with the $c$--tensor we obtain:
     \begin{equation*}
         c^i_{ks}\nabla_j X^s = c^i_{ks}c^s_{jt}\nabla_e X^t,
     \end{equation*}
 	 implying
 	 \begin{equation}
 	 	\begin{aligned}
 	 		c^i_{js}\nabla_kX^s-c^i_{ks}\nabla_jX^s&=\big(c^i_{js}c^s_{kt}-c^i_{ks}c^s_{jt}\big)\nabla_e X^t=0,
 	 	\end{aligned}
 	 	\label{eq:3RCcriterion}
 	 \end{equation}
  	 where we have used the associativity of the product. Finally, as shown in (\cite{LPR}, Proposition 16), \eqref{eq:3RCcriterion} is precisely the condition to ensure that \eqref{eq: 3RCX} holds.
\end{proof}
\begin{corollary}
  \eqref{eq: 3RCX} holds for $X$ being an arbitrary local vector field.
  \label{cor:dnablezeroimplies3RC}
\end{corollary}
\begin{proof}
   By Lemma \ref{LemmadnablaBfromtriangular}, $V^i  \equiv X_{(i)}\circ$ satisfies \eqref{eq:dnablaXcirc} for any $i\in\{0,\dots,n-1\}$. Hence, we must have that  
    \begin{equation}
        (R^k_{lmj}c^n_{pk} + R^k_{ljp}c^n_{mk} + R^k_{lpm}c^n_{jk})X_{(i)}^l = 0
    \end{equation}
    holds for any  $i\in\{0,\dots,n-1\}$. Finally, under the assumption that $\{X_{(i)}\}_{i\in\{0,\dots,n-1\}}$ provides a basis, we must have that \eqref{eq: 3RCX} holds for arbitrary local vector field $X$. 
\end{proof}
This implies that
\begin{equation}
	R^j_{skl}c^s_{mi} + R^j_{smk}c^s_{li}+R^j_{slm}c^s_{ki} = 0,
	\label{eq:3RC}
\end{equation}
which phrases \eqref{rc-intri}, or equivalently \eqref{rc-intri-2}, in local coordinates.

\begin{corollary}
   $\nabla \circ$ is symmetric.
\end{corollary}
\begin{proof}
	Let $X$ belong to the set $\{X_{(0)},\dots,X_{(n-1)}\}$. Since such $X$ satisfies both \eqref{eq:dnablaXcirc}, or equivalently \eqref{eq:dnablaXcirccoords}, and \eqref{eq:3RCcriterion}, we get that 
	\begin{equation}
		\big(\nabla_jc^i_{ks}-\nabla_kc^i_{js}\big)X^s=0.
		\label{almostsymmetryofnablac}
	\end{equation}
	As $\{X_{(0)},\dots,X_{(n-1)}\}$ provides a basis, \eqref{almostsymmetryofnablac} holds for an arbitrary local vector field $X$, yielding
	\begin{equation}
		\nabla_jc^i_{ks}-\nabla_kc^i_{js}=0.
		\label{symmetryofnablac}
	\end{equation}
\end{proof}
We have proved that if among a set of linearly independent local vector fields $\{X_{(0)},\dots,X_{(n-1)}\}$ on an $n$-dimensional regular F-manifold $(M,\circ,e)$ defining pairwise commuting flows
\begin{align}
	u_{t_i}=X_{(i)}\circ u_x,\qquad i\in\{0,\dots,n-1\},
	\notag
\end{align}
there exists a local vector field $X\in\{X_{(0)},\dots,X_{(n-1)}\}$ such that ${X^{1(\alpha)}\neq X^{1(\beta)}}$ for ${\alpha\neq\beta}$, $\alpha,\beta\in\{1,\dots,r\}$, and $X^{2(\alpha)}\neq0$ for each ${\alpha\in\{1,\dots,r\}}$ then, for $\nabla$ being a torsionless connection such that $\nabla e=0$ and $d_\nabla(X\circ)=0$,  $(M,\circ,e,\nabla)$ defines an F-manifold with compatible connection and flat unit. Below, we discuss in more detail the existence and uniqueness of such a connection. In particular, in the proof of the following, explicit recursive expressions for the Christoffel symbols are recovered.
\begin{proposition}
	For $(M,\circ,e)$ being an $n$-dimensional regular F-manifold and $X$ being a local vector field realising
	\begin{itemize}
		\item $X^{1(\alpha)}\neq X^{1(\beta)}$ for $\alpha\neq\beta$, $\alpha,\beta\in\{1,\dots,r\}$,
		\item $X^{2(\alpha)}\neq0$, $\alpha\in\{1,\dots,r\}$,
	\end{itemize}
	there exists a unique torsionless connection $\nabla$, with Christoffel symbols $\big\{\Gamma^i_{jk}\big\}_{i,j,k\in\{1,\dots,n\}}$, such that $d_\nabla(X\circ)=0$ and $\nabla e=0$.
\end{proposition}
\begin{proof}
In coordinates, the condition $d_{\nabla}V = 0$ reads
\begin{equation}
    (d_{\nabla}V)^{k(\gamma)}_{i(\alpha) j(\beta)} = \partial_{i(\alpha)} V^{k(\gamma)}_{j(\beta)} - \partial_{j(\beta)} V^{k(\gamma)}_{i(\alpha)} + \Gamma^{k(\gamma)}_{i(\alpha) l(\sigma)} V^{l(\sigma)}_{j(\beta)} -\Gamma^{k(\gamma)}_{j(\beta) l(\sigma)} V^{l(\sigma)}_{i(\alpha)}  = 0,
    \label{eq:dnablaVcoords}
\end{equation}
 for every $\alpha,\beta,\gamma\in\{1,\dots,r\}$ with $k \in \{1, \dots, m_{\gamma}\}$, $i \in \{1, \dots, m_{\alpha}\}$, $j \in \{1, \dots, m_{\beta}\}$.

	We need to consider the following cases:
	\begin{itemize}
		\item[a.] $\alpha\neq\beta\neq\gamma\neq\alpha$,
		\item[b.] $\alpha=\gamma\neq\beta$ (which covers $\alpha=\beta\neq\gamma$ as well),
		\item[c.] $\alpha\neq\beta=\gamma$,
		\item[d.] $\alpha=\beta=\gamma$.
	\end{itemize}

\textbf{Case a: $\alpha\neq\beta\neq\gamma\neq\alpha$.} 
Here we prove that whenever $\alpha, \beta, \gamma$ are pairwise distinct, the Christoffel symbols vanish, i.e.
        \begin{equation}
            \Gamma^{k(\gamma)}_{i(\alpha) j(\beta)} = 0.
            \label{eq:chrdistinctzero}
        \end{equation}

If $\alpha, \beta, \gamma$ are pairwise distinct the first two terms in \eqref{eq:dnablaVcoords} vanish by the block diagonal form of $V$. Furthermore, for the same reason, the only nonzero contributions are given by $\sigma = \beta$ in the third term, and $\sigma = \alpha$ in the fourth. Consequently, we can write
\begin{equation}
    \Gamma^{k(\gamma)}_{i(\alpha) j(\beta)}(X^{1(\alpha)}- X^{1(\beta)}) = \sum_{l > j} \Gamma^{k(\gamma)}_{i(\alpha)l(\beta)} X^{(l-j+1)(\beta)} - \sum_{l > i} \Gamma^{k(\gamma)}_{j(\beta)l(\alpha)} X^{(l-i+1)(\alpha)}.
    \label{eq:gammadistinctequa}
\end{equation}

    We shall now proceed by induction over the sum of the lower indices, $i + j$. The base step is given by ${i,j}$ such $i+j = m_{\alpha} + m_{\beta}$ which implies $i = m_{\alpha}$ and $j = m_{\beta}$. In this case  \eqref{eq:gammadistinctequa} becomes
    \begin{equation}
        \Gamma^{k(\gamma)}_{m_{\alpha}(\alpha) m_{\beta}(\beta)} (X^{1(\alpha)}- X^{1(\beta)}) = 0,
    \end{equation}
    and as $X^{1(\alpha)} \neq X^{1(\beta)}$ by assumption for any $\alpha \neq \beta$, we must have $ \Gamma^{k(\gamma)}_{m_{\alpha}(\alpha) m_{\beta}(\beta)} = 0$. 
    Now, let's assume that $\Gamma^{k(\gamma)}_{i(\alpha) j(\beta)} = 0$ for any $i, j$ such that $i + j > M$. Then for $p, q$ with $p + q = M$ we have
    \begin{equation}
         \Gamma^{k(\gamma)}_{p(\alpha) q(\beta)}(X^{1(\alpha)}- X^{1(\beta)}) = \sum_{l > q} \underset{ = \, 0 \text{ by hypothesis}}{\underbrace{\Gamma^{k(\gamma)}_{p(\alpha)l(\beta)}}} X^{(l-q+1)(\beta)} - \sum_{l > p} \underset{ = \, 0 \text{ by hypothesis}}{\underbrace{\Gamma^{k(\gamma)}_{q(\beta)l(\alpha)}}} X^{(l-p+1)(\alpha)},
    \end{equation}
    giving $\Gamma^{k(\gamma)}_{p(\alpha) q(\beta)} = 0$, as required.

\vspace{1em}
 
	Letting $V = X \circ$ in \eqref{eq:dnablaVcoords} we have
	\begin{align}
 0=&\big(d_\nabla(X\circ)\big)^{i(\alpha)}_{j(\beta)k(\gamma)}=\nabla_{j(\beta)}(X\circ)^{i(\alpha)}_{k(\gamma)}-\nabla_{k(\gamma)}(X\circ)^{i(\alpha)}_{j(\beta)}\notag\\
		=&\partial_{j(\beta)}(X\circ)^{i(\alpha)}_{k(\gamma)}+\Gamma^{i(\alpha)}_{j(\beta)s(\sigma)}(X\circ)^{s(\sigma)}_{k(\gamma)}-\partial_{k(\gamma)}(X\circ)^{i(\alpha)}_{j(\beta)}-\Gamma^{i(\alpha)}_{k(\gamma)s(\sigma)}(X\circ)^{s(\sigma)}_{j(\beta)}\notag\\
		=&\delta^\alpha_\gamma\partial_{j(\beta)}X^{(i-k+1)(\alpha)}+\Gamma^{i(\alpha)}_{j(\beta)s(\gamma)}X^{(s-k+1)(\gamma)}-\delta^\alpha_\beta\partial_{k(\gamma)}X^{(i-j+1)(\alpha)}-\Gamma^{i(\alpha)}_{k(\gamma)s(\beta)}X^{(s-j+1)(\beta)}.
		\label{dnablaXcirc0}
	\end{align}
 The index $i(\alpha)$ will be fixed throughout the proof of the cases b, c, d.
	\\\textbf{Case b: $\alpha=\gamma\neq\beta$.} 
 In this case, condition \eqref{dnablaXcirc0} becomes
	\begin{align}
		0=&\partial_{j(\beta)}X^{(i-k+1)(\alpha)}+\Gamma^{i(\alpha)}_{j(\beta)s(\alpha)}X^{(s-k+1)(\alpha)}-\Gamma^{i(\alpha)}_{k(\alpha)s(\beta)}X^{(s-j+1)(\beta)}\notag\\
		=&\partial_{j(\beta)}X^{(i-k+1)(\alpha)}+\Gamma^{i(\alpha)}_{j(\beta)k(\alpha)}\big(X^{1(\alpha)}-X^{1(\beta)}\big)\notag\\&+\overset{m_\alpha}{\underset{s=k+1}{\sum}}\Gamma^{i(\alpha)}_{j(\beta)s(\alpha)}X^{(s-k+1)(\alpha)}-\overset{m_\beta}{\underset{s=j+1}{\sum}}\Gamma^{i(\alpha)}_{k(\alpha)s(\beta)}X^{(s-j+1)(\beta)}.
		\label{dnablaXcirc0_caseb}
	\end{align}
	By choosing $k=m_\alpha$ in \eqref{dnablaXcirc0_caseb}, we get
	\begin{align}
		0=&\partial_{j(\beta)}X^{(i-m_\alpha+1)(\alpha)}+\Gamma^{i(\alpha)}_{j(\beta)m_\alpha(\alpha)}\big(X^{1(\alpha)}-X^{1(\beta)}\big)-\overset{m_\beta}{\underset{s=j+1}{\sum}}\Gamma^{i(\alpha)}_{m_\alpha(\alpha)s(\beta)}X^{(s-j+1)(\beta)}.
		\label{dnablaXcirc0_caseb_kmalpha}
	\end{align}
	By choosing $j=m_\beta$ in \eqref{dnablaXcirc0_caseb_kmalpha}, we get
	\begin{align}
		0=&\partial_{m_\beta(\beta)}X^{(i-m_\alpha+1)(\alpha)}+\Gamma^{i(\alpha)}_{m_\beta(\beta)m_\alpha(\alpha)}\big(X^{1(\alpha)}-X^{1(\beta)}\big)
		\notag
	\end{align}
	that is
	\begin{align}
		\Gamma^{i(\alpha)}_{m_\beta(\beta)m_\alpha(\alpha)}=&-\frac{\partial_{m_\beta(\beta)}X^{(i-m_\alpha+1)(\alpha)}}{X^{1(\alpha)}-X^{1(\beta)}}.
		\notag
	\end{align}
	For any $j\in\{1,\dots,m_\beta-1\}$, the condition \eqref{dnablaXcirc0_caseb_kmalpha} allows one to recover $\Gamma^{i(\alpha)}_{j(\beta)m_\alpha(\alpha)}$ in terms of $\big\{\Gamma^{i(\alpha)}_{s(\beta)m_\alpha(\alpha)}\big\}_{s\geq j+1}$ as
	\begin{align}
		\Gamma^{i(\alpha)}_{j(\beta)m_\alpha(\alpha)}=&-\frac{1}{X^{1(\alpha)}-X^{1(\beta)}}\Bigg(\partial_{j(\beta)}X^{(i-m_\alpha+1)(\alpha)}-\overset{m_\beta}{\underset{s=j+1}{\sum}}\Gamma^{i(\alpha)}_{m_\alpha(\alpha)s(\beta)}X^{(s-j+1)(\beta)}\Bigg).
		\notag
	\end{align}
	By choosing $j=m_\beta$ in \eqref{dnablaXcirc0_caseb}, we get
	\begin{align}
		0=&\partial_{m_\beta(\beta)}X^{(i-k+1)(\alpha)}+\Gamma^{i(\alpha)}_{m_\beta(\beta)k(\alpha)}\big(X^{1(\alpha)}-X^{1(\beta)}\big)+\overset{m_\alpha}{\underset{s=k+1}{\sum}}\Gamma^{i(\alpha)}_{m_\beta(\beta)s(\alpha)}X^{(s-k+1)(\alpha)}
		\notag
	\end{align}
	which, for each $k\in\{1,\dots,m_\alpha-1\}$, allows one to recover $\Gamma^{i(\alpha)}_{m_\beta(\beta)k(\alpha)}$ in terms of $\big\{\Gamma^{i(\alpha)}_{m_\beta(\beta)h(\alpha)}\big\}_{h\geq k+1}$ as
	\begin{align}
		\Gamma^{i(\alpha)}_{m_\beta(\beta)k(\alpha)}=&-\frac{1}{X^{1(\alpha)}-X^{1(\beta)}}\Bigg(\partial_{m_\beta(\beta)}X^{(i-k+1)(\alpha)}+\overset{m_\alpha}{\underset{s=k+1}{\sum}}\Gamma^{i(\alpha)}_{m_\beta(\beta)s(\alpha)}X^{(s-k+1)(\alpha)}\Bigg).
		\notag
	\end{align}
	For any $k\in\{1,\dots,m_\alpha-1\}$ and $j\in\{1,\dots,m_\beta-1\}$, the condition \eqref{dnablaXcirc0_caseb} allows one to recover $\Gamma^{i(\alpha)}_{j(\beta)k(\alpha)}$ in terms of $\big\{\Gamma^{i(\alpha)}_{s(\beta)h(\alpha)}\big\}_{h\geq k+1, s\geq 1}$ and $\big\{\Gamma^{i(\alpha)}_{s(\beta)k(\alpha)}\big\}_{s\geq j+1}$ as
	\begin{align}
		\Gamma^{i(\alpha)}_{j(\beta)k(\alpha)}=&-\frac{1}{X^{1(\alpha)}-X^{1(\beta)}}\Bigg(\partial_{j(\beta)}X^{(i-k+1)(\alpha)}\notag\\&+\overset{m_\alpha}{\underset{s=k+1}{\sum}}\Gamma^{i(\alpha)}_{j(\beta)s(\alpha)}X^{(s-k+1)(\alpha)}-\overset{m_\beta}{\underset{s=j+1}{\sum}}\Gamma^{i(\alpha)}_{k(\alpha)s(\beta)}X^{(s-j+1)(\beta)}\Bigg).
		\notag
	\end{align}

	From the formulas obtained by considering the cases a and b, we are able, by means of the requirement concerning the flatness of the unit, to determine
	\[
	\big\{\Gamma^{i(\alpha)}_{j(\beta)k(\beta)}\,|\,j=1\,\vee\,k=1\big\}_{\alpha,\beta\in\{1,\dots,r\}}.
	\]
	We recall that such a condition, $\nabla e=0$, reads
	\[
	\overset{r}{\underset{\sigma=1}{\sum}}\,\Gamma^{i(\alpha)}_{1(\sigma)j(\beta)}=0
	\]
	for every $\alpha,\beta\in\{1,\dots,r\}$ and every $i\in\{1,\dots,m_\alpha\}$, $j\in\{1,\dots,m_\beta\}$. For $\alpha\neq\beta$, we have
	\begin{align}
		\Gamma^{i(\alpha)}_{j(\beta)1(\beta)}=&-\overset{}{\underset{\sigma\neq\beta}{\sum}}\,\Gamma^{i(\alpha)}_{j(\beta)1(\sigma)}\overset{\eqref{eq:chrdistinctzero}}{=}-\Gamma^{i(\alpha)}_{j(\beta)1(\alpha)}.
		\notag
	\end{align}
	For $\alpha=\beta$, we have
	\begin{align}
		\Gamma^{i(\alpha)}_{j(\alpha)1(\alpha)}=&-\overset{}{\underset{\sigma\neq\alpha}{\sum}}\,\Gamma^{i(\alpha)}_{j(\alpha)1(\sigma)}.
		\notag
	\end{align}
	With this additional piece of information, we discuss the remaining cases c and d. The formulas that we recover below are then to be considered in the remaining cases where both the lower indices are of the form $i(\sigma)$ for some $\sigma\in\{1,\dots,r\}$ and some $i\in\{2,\dots,m_\sigma\}$. In particular, we consider the case where both the lower indices refer to blocks of size greater or equal than $2$.
	\\\textbf{Case d: $\alpha=\beta=\gamma$.} Condition \eqref{dnablaXcirc0} becomes
	\begin{align}
		0=&\partial_{j(\alpha)}X^{(i-k+1)(\alpha)}+\Gamma^{i(\alpha)}_{j(\alpha)s(\alpha)}X^{(s-k+1)(\alpha)}-\partial_{k(\alpha)}X^{(i-j+1)(\alpha)}-\Gamma^{i(\alpha)}_{k(\alpha)s(\alpha)}X^{(s-j+1)(\alpha)}\notag\\
		=&\partial_{j(\alpha)}X^{(i-k+1)(\alpha)}+\overset{m_\alpha}{\underset{s=k+1}{\sum}}\Gamma^{i(\alpha)}_{j(\alpha)s(\alpha)}X^{(s-k+1)(\alpha)}-\partial_{k(\alpha)}X^{(i-j+1)(\alpha)}\notag\\&-\overset{m_\alpha}{\underset{s=j+1}{\sum}}\Gamma^{i(\alpha)}_{k(\alpha)s(\alpha)}X^{(s-j+1)(\alpha)}.
		\label{dnablaXcirc0_cased}
	\end{align}
	By choosing $k=m_\alpha$ in \eqref{dnablaXcirc0_cased}, we get
	\begin{align}
		\partial_{j(\alpha)}X^{(i-m_\alpha+1)(\alpha)}-\partial_{m_\alpha(\alpha)}X^{(i-j+1)(\alpha)}-\overset{m_\alpha}{\underset{s=j+1}{\sum}}\Gamma^{i(\alpha)}_{m_\alpha(\alpha)s(\alpha)}X^{(s-j+1)(\alpha)}=0.
		\label{dnablaXcirc0_cased_kmalpha}
	\end{align}
	By choosing $j=m_\alpha-1$ in \eqref{dnablaXcirc0_cased_kmalpha}, we get
	\begin{align}
		\partial_{(m_\alpha-1)(\alpha)}X^{(i-m_\alpha+1)(\alpha)}-\partial_{m_\alpha(\alpha)}X^{(i-m_\alpha+2)(\alpha)}-\Gamma^{i(\alpha)}_{m_\alpha(\alpha)m_\alpha(\alpha)}X^{2(\alpha)}=0
		\notag
	\end{align}
	that is
	\begin{align}
		\Gamma^{i(\alpha)}_{m_\alpha(\alpha)m_\alpha(\alpha)}=\frac{1}{X^{2(\alpha)}}\bigg(\partial_{(m_\alpha-1)(\alpha)}X^{(i-m_\alpha+1)(\alpha)}-\partial_{m_\alpha(\alpha)}X^{(i-m_\alpha+2)(\alpha)}\bigg).
		\notag
	\end{align}
	For any $j\in\{1,\dots,m_\alpha-2\}$, the condition \eqref{dnablaXcirc0_cased_kmalpha} allows one to recover $\Gamma^{i(\alpha)}_{m_\alpha(\alpha)(j+1)(\alpha)}$ in terms of $\big\{\Gamma^{i(\alpha)}_{m_\alpha(\alpha)s(\alpha)}\big\}_{s\geq j+2}$ as
	\begin{align}
		\Gamma^{i(\alpha)}_{m_\alpha(\alpha)(j+1)(\alpha)}=\frac{1}{X^{2(\alpha)}}\bigg(&\partial_{j(\alpha)}X^{(i-m_\alpha+1)(\alpha)}-\partial_{m_\alpha(\alpha)}X^{(i-j+1)(\alpha)}\notag\\&-\overset{m_\alpha}{\underset{s=j+2}{\sum}}\Gamma^{i(\alpha)}_{m_\alpha(\alpha)s(\alpha)}X^{(s-j+1)(\alpha)}\bigg).
		\notag
	\end{align}
	By simply shifting the index $j\to j-1$, for any $j\in\{2,\dots,m_\alpha-1\}$ we have recovered $\Gamma^{i(\alpha)}_{m_\alpha(\alpha)j(\alpha)}$ in terms of $\big\{\Gamma^{i(\alpha)}_{m_\alpha(\alpha)s(\alpha)}\big\}_{s\geq j+1}$ as
	\begin{align}
		\Gamma^{i(\alpha)}_{m_\alpha(\alpha)j(\alpha)}=\frac{1}{X^{2(\alpha)}}\bigg(&\partial_{(j-1)(\alpha)}X^{(i-m_\alpha+1)(\alpha)}-\partial_{m_\alpha(\alpha)}X^{(i-j+2)(\alpha)}\notag\\&-\overset{m_\alpha}{\underset{s=j+1}{\sum}}\Gamma^{i(\alpha)}_{m_\alpha(\alpha)s(\alpha)}X^{(s-j+2)(\alpha)}\bigg).
		\notag
	\end{align}
	By choosing $j=k-1$ in \eqref{dnablaXcirc0_cased}, we get
	\begin{align}
		&\partial_{(k-1)(\alpha)}X^{(i-k+1)(\alpha)}+\overset{m_\alpha}{\underset{s=k+1}{\sum}}\Gamma^{i(\alpha)}_{(k-1)(\alpha)s(\alpha)}X^{(s-k+1)(\alpha)}\notag\\&-\partial_{k(\alpha)}X^{(i-k+2)(\alpha)}-\overset{m_\alpha}{\underset{s=k}{\sum}}\Gamma^{i(\alpha)}_{k(\alpha)s(\alpha)}X^{(s-k+2)(\alpha)}=0
		\notag
	\end{align}
	which, for each $k\in\{2,\dots,m_\alpha-1\}$, allows one to recover $\Gamma^{i(\alpha)}_{k(\alpha)k(\alpha)}$ in terms of $\big\{\Gamma^{i(\alpha)}_{h(\alpha)j(\alpha)}\big\}_{h\geq k+1,\,j\geq1}$ as
	\begin{align}
		\Gamma^{i(\alpha)}_{k(\alpha)k(\alpha)}=\frac{1}{X^{2(\alpha)}}\bigg(&\partial_{(k-1)(\alpha)}X^{(i-k+1)(\alpha)}+\overset{m_\alpha}{\underset{s=k+1}{\sum}}\Gamma^{i(\alpha)}_{(k-1)(\alpha)s(\alpha)}X^{(s-k+1)(\alpha)}\notag\\&-\partial_{k(\alpha)}X^{(i-k+2)(\alpha)}-\overset{m_\alpha}{\underset{s=k+1}{\sum}}\Gamma^{i(\alpha)}_{k(\alpha)s(\alpha)}X^{(s-k+2)(\alpha)}\bigg).
		\notag
	\end{align}
	For any $k\in\{2,\dots,m_\alpha-1\}$ and $j\in\{1,\dots,k-2\}$, the condition \eqref{dnablaXcirc0_cased} allows one to recover $\Gamma^{i(\alpha)}_{k(\alpha)(j+1)(\alpha)}$ in terms of $\big\{\Gamma^{i(\alpha)}_{h(\alpha)t(\alpha)}\big\}_{h\geq k+1,\,t\geq1}$ and $\big\{\Gamma^{i(\alpha)}_{k(\alpha)s(\alpha)}\big\}_{s\geq j+2}$ as
	\begin{align}
		\Gamma^{i(\alpha)}_{k(\alpha)(j+1)(\alpha)}=\frac{1}{X^{2(\alpha)}}\bigg(&\partial_{j(\alpha)}X^{(i-k+1)(\alpha)}+\overset{m_\alpha}{\underset{s=k+1}{\sum}}\Gamma^{i(\alpha)}_{j(\alpha)s(\alpha)}X^{(s-k+1)(\alpha)}\notag\\&-\partial_{k(\alpha)}X^{(i-j+1)(\alpha)}-\overset{m_\alpha}{\underset{s=j+2}{\sum}}\Gamma^{i(\alpha)}_{k(\alpha)s(\alpha)}X^{(s-j+1)(\alpha)}\bigg).
		\notag
	\end{align}
	By simply shifting the index $j\to j-1$, for any $k\in\{2,\dots,m_\alpha-1\}$, ${j\in\{2,\dots,k-1\}}$ we have recovered $\Gamma^{i(\alpha)}_{k(\alpha)j(\alpha)}$ in terms of $\big\{\Gamma^{i(\alpha)}_{h(\alpha)t(\alpha)}\big\}_{h\geq k+1,\,t\geq1}$ and $\big\{\Gamma^{i(\alpha)}_{k(\alpha)s(\alpha)}\big\}_{s\geq j+1}$ as
	\begin{align}
		\Gamma^{i(\alpha)}_{k(\alpha)j(\alpha)}=\frac{1}{X^{2(\alpha)}}\bigg(&\partial_{(j-1)(\alpha)}X^{(i-k+1)(\alpha)}+\overset{m_\alpha}{\underset{s=k+1}{\sum}}\Gamma^{i(\alpha)}_{(j-1)(\alpha)s(\alpha)}X^{(s-k+1)(\alpha)}\notag\\&-\partial_{k(\alpha)}X^{(i-j+2)(\alpha)}-\overset{m_\alpha}{\underset{s=j+1}{\sum}}\Gamma^{i(\alpha)}_{k(\alpha)s(\alpha)}X^{(s-j+2)(\alpha)}\bigg).
		\notag
	\end{align}
	\textbf{Case c: $\alpha\neq\beta=\gamma$.} Condition \eqref{dnablaXcirc0} becomes
	\begin{align}
		0=&\Gamma^{i(\alpha)}_{j(\beta)s(\beta)}X^{(s-k+1)(\beta)}-\Gamma^{i(\alpha)}_{k(\beta)s(\beta)}X^{(s-j+1)(\beta)}\notag\\
		=&\overset{m_\beta}{\underset{s=k+1}{\sum}}\Gamma^{i(\alpha)}_{j(\beta)s(\beta)}X^{(s-k+1)(\beta)}-\overset{m_\beta}{\underset{s=j+1}{\sum}}\Gamma^{i(\alpha)}_{k(\beta)s(\beta)}X^{(s-j+1)(\beta)}.
		\label{dnablaXcirc0_casec}
	\end{align}
	By choosing $k=m_\beta$ in \eqref{dnablaXcirc0_casec}, we get
	\begin{align}
		0=&-\overset{m_\beta}{\underset{s=j+1}{\sum}}\Gamma^{i(\alpha)}_{m_\beta(\beta)s(\beta)}X^{(s-j+1)(\beta)}.
		\label{dnablaXcirc0_casec_kmbeta}
	\end{align}
	By choosing $j=m_\beta-1$ in \eqref{dnablaXcirc0_casec_kmbeta}, we get
	\begin{align}
		\Gamma^{i(\alpha)}_{m_\beta(\beta)m_\beta(\beta)}=0.
		\notag
	\end{align}
	For any $j\in\{1,\dots,m_\beta-2\}$, the condition \eqref{dnablaXcirc0_casec_kmbeta} allows one to recover $\Gamma^{i(\alpha)}_{m_\beta(\beta)(j+1)(\beta)}$ in terms of $\big\{\Gamma^{i(\alpha)}_{m_\beta(\beta)s(\beta)}\big\}_{s\geq j+2}$ as
	\begin{align}
		\Gamma^{i(\alpha)}_{m_\beta(\beta)(j+1)(\beta)}=&-\frac{1}{X^{2(\beta)}}\overset{m_\beta}{\underset{s=j+2}{\sum}}\Gamma^{i(\alpha)}_{m_\beta(\beta)s(\beta)}X^{(s-j+1)(\beta)}.
	\end{align}
	By simply shifting the index $j\to j-1$, for any $j\in\{2,\dots,m_\beta-1\}$ we have recovered $\Gamma^{i(\alpha)}_{m_\beta(\beta)j(\beta)}$ in terms of $\big\{\Gamma^{i(\alpha)}_{m_\beta(\beta)s(\beta)}\big\}_{s\geq j+1}$ as
	\begin{align}
		\Gamma^{i(\alpha)}_{m_\beta(\beta)j(\beta)}=&-\frac{1}{X^{2(\beta)}}\overset{m_\beta}{\underset{s=j+1}{\sum}}\Gamma^{i(\alpha)}_{m_\beta(\beta)s(\beta)}X^{(s-j+2)(\beta)}.
	\end{align}
	By choosing $j=k-1$ in \eqref{dnablaXcirc0_casec}, we get
	\begin{align}
		&\overset{m_\beta}{\underset{s=k+1}{\sum}}\Gamma^{i(\alpha)}_{(k-1)(\beta)s(\beta)}X^{(s-k+1)(\beta)}-\overset{m_\beta}{\underset{s=k}{\sum}}\Gamma^{i(\alpha)}_{k(\beta)s(\beta)}X^{(s-k+2)(\beta)}=0
		\notag
	\end{align}
	which, for each $k\in\{2,\dots,m_\beta-1\}$, allows one to recover $\Gamma^{i(\alpha)}_{k(\beta)k(\beta)}$ in terms of $\big\{\Gamma^{i(\alpha)}_{h(\beta)j(\beta)}\big\}_{h\geq k+1,\,j\geq1}$ as
	\begin{align}
		\Gamma^{i(\alpha)}_{k(\beta)k(\beta)}=\frac{1}{X^{2(\beta)}}\Bigg(\overset{m_\beta}{\underset{s=k+1}{\sum}}\Gamma^{i(\alpha)}_{(k-1)(\beta)s(\beta)}X^{(s-k+1)(\beta)}-\overset{m_\beta}{\underset{s=k+1}{\sum}}\Gamma^{i(\alpha)}_{k(\beta)s(\beta)}X^{(s-k+2)(\beta)}\Bigg).
		\notag
	\end{align}
	For any $k\in\{2,\dots,m_\beta-1\}$ and $j\in\{1,\dots,k-2\}$, the condition \eqref{dnablaXcirc0_casec} allows one to recover $\Gamma^{i(\alpha)}_{k(\beta)(j+1)(\beta)}$ in terms of $\big\{\Gamma^{i(\alpha)}_{h(\beta)t(\beta)}\big\}_{h\geq k+1,\,t\geq1}$ and $\big\{\Gamma^{i(\alpha)}_{k(\beta)s(\beta)}\big\}_{s\geq j+2}$ as
	\begin{align}
		\Gamma^{i(\alpha)}_{k(\beta)(j+1)(\beta)}=&\frac{1}{X^{2(\beta)}}\Bigg(\overset{m_\beta}{\underset{s=k+1}{\sum}}\Gamma^{i(\alpha)}_{j(\beta)s(\beta)}X^{(s-k+1)(\beta)}-\overset{m_\beta}{\underset{s=j+2}{\sum}}\Gamma^{i(\alpha)}_{k(\beta)s(\beta)}X^{(s-j+1)(\beta)}\Bigg).
		\notag
	\end{align}
	By simply shifting the index $j\to j-1$, for any $k\in\{2,\dots,m_\beta-1\}$, ${j\in\{2,\dots,k-1\}}$ we have recovered $\Gamma^{i(\alpha)}_{k(\beta)j(\beta)}$ in terms of $\big\{\Gamma^{i(\alpha)}_{h(\beta)t(\beta)}\big\}_{h\geq k+1,\,t\geq1}$ and $\big\{\Gamma^{i(\alpha)}_{k(\beta)s(\beta)}\big\}_{s\geq j+1}$ as
	\begin{align}
		\Gamma^{i(\alpha)}_{k(\beta)j(\beta)}=&\frac{1}{X^{2(\beta)}}\Bigg(\overset{m_\beta}{\underset{s=k+1}{\sum}}\Gamma^{i(\alpha)}_{(j-1)(\beta)s(\beta)}X^{(s-k+1)(\beta)}-\overset{m_\beta}{\underset{s=j+1}{\sum}}\Gamma^{i(\alpha)}_{k(\beta)s(\beta)}X^{(s-j+2)(\beta)}\Bigg).
		\notag
	\end{align}
\end{proof}
We have proved the following.
\begin{theorem}\label{mainTh}
	Let $\{X_{(0)},\dots,X_{(n-1)}\}$ be a set of linearly independent local vector fields on an $n$-dimensional regular F-manifold $(M,\circ,e)$. Let's assume that the corresponding flows
	\begin{align}
		u_{t_i}=X_{(i)}\circ u_x,\qquad i\in\{0,\dots,n-1\},
		\notag
	\end{align}
	pairwise commute, and that there exists a local vector field $X\in\{X_{(0)},\dots,X_{(n-1)}\}$ such that ${X^{1(\alpha)}\neq X^{1(\beta)}}$ for ${\alpha\neq\beta}$, $\alpha,\beta\in\{1,\dots,r\}$, and $X^{2(\alpha)}\neq0$ for each ${\alpha\in\{1,\dots,r\}}$. Then an F-manifold $(M,\circ,e,\nabla)$ with compatible connection and flat unit is defined by the unique torsionless connection $\nabla$ realising $\nabla e=0$ and $d_\nabla(X\circ)=0$.
\end{theorem}

\section{F-manifolds with compatible connection from Fr\"{o}licher-Nijenhuis bicomplexes}
\subsection{Fr\"olicher-Nijenhuis bicomplex and integrable systems}
Let $L$ be a  tensor field of type $(1,1)$ with vanishing Nijenhuis torsion. This means that for any pair of vector fields, $X$, $Y$, we have
\beq\label{nijenhis}
[LX,LY]-L\,[X,LY]-L\,[LX,Y]+L^2\,[X,Y]=0.
\eeq
Following \cite{LM}, we now recall a construction of integrable hierarchies starting from the  Fr\"{o}licher-Nijenhuis  bicomplex $(d,d_L,\Omega(M))$ (see \cite{FN}). 
  The differential $d$ is the usual de Rham differential, while the differential $d_L$ is defined as
 $$(d_L \omega)(X_0, \dots, X_k):=\sum_{i=0}^k (-1)^i (LX_i)(\omega(X_0, \dots, \hat{X}_i, \dots, X_k))$$
$$+\sum_{0\leq i<j\leq k}(-1)^{i+j}\,\omega([X_i, X_j]_L, X_0, \dots, \hat{X}_i, \dots, \hat{X}_j, \dots X_k),$$
where  $\omega\in\Omega^k(M)$ and  
$$[X_i,X_j]_L:=[LX_i, X_j]+[X_i, LX_j]-L[X_i,X_j].$$
For $L=I$ the vector field $[X_i,X_j]_L$ reduces to the commutator $[X_i,X_j]$ and the differential $d_L$ to $d$.

The fact that $d^2_L=0$, is equivalent to the vanishing of the Nijenhuis torsion.
The differentials $d$ and $d_L$ anticommute and thus the pair $(d,d_L)$ defines a bidifferential complex. 
In \cite{LM} (see also \cite{L2006}), it was proved that given any solution of 
 the equation
\beq
d\cdot d_La_0=0,
\label{eq:main}
\eeq
and the corresponding sequence of functions $a_1,a_2,a_3,...$  defined recursively by 
\[da_{k+1}=d_L a_k -a_k da_0,\qquad k\geq0,\]
the tensor fields of type $(1,1)$  defined  by 
\begin{eqnarray*}
V_{k+1} \coloneqq V_k L-a_k I,\qquad k\geq0,
\end{eqnarray*}
starting from $V_0$ being the identity $I$:
\begin{eqnarray*}
&&V_0=I\\
&&V_1=L-a_0 I\\
&&V_{2}=L^2-a_0 L-a_1 I\\
&&\qquad \vdots
\end{eqnarray*}
 define an integrable hierarchy of hydrodynamic type.
\newline
\newline
We apply the construction from the previous section to the case where ${X_{(0)},\dots,X_{(n-1)}}$ are specified by the vector fields
\begin{equation}
	V_i=X_{(i)}\circ,\qquad i\in\{0,\dots,n-1\}.
\end{equation}
By considering $X:=X_{(1)}=E-a_0\,e$, in canonical coordinates we have that ${X^{i(\alpha)}=u^{i(\alpha)}-a_0\,\delta^i_1}$ for each $\alpha\in\{1,\dots,r\}$, $i\in\{1,\dots,m_\alpha\}$, implying ${X^{1(\alpha)}\neq X^{1(\beta)}}$ for ${\alpha\neq\beta}$, $\alpha,\beta\in\{1,\dots,r\}$, and $X^{2(\alpha)}\neq0$ for each ${\alpha\in\{1,\dots,r\}}$. Moreover, we impose $d_\nabla(L-a_0\,I)=0$. Since, as shown in the following proposition, $X_{(0)},\dots,X_{(n-1)}$ are linearly independent, Lemma \ref{LemmadnablaBfromtriangular} guarantees that $d_\nabla(Y\circ)=0$ for any other ${Y\in\{X_{(0)},\dots,X_{(n-1)}\}}$.
\begin{proposition}\label{LinIn}
	The vector fields $X_{0},\dots,X_{(n-1)}$, defined by $X_{(k)}\circ=V_k$ for ${k\in\{0,\dots,n-1\}}$, are linearly independent.
\end{proposition}
\begin{proof}
	It is immediate to see that the vector fields $\{X_{0},\dots,X_{(n-1)}\}$ take the form
	\begin{align}
		&X_{(0)}=e,\notag\\
		&X_{(k)}=L\,X_{(k-1)}-a_{k-1}\,e=E^k-\overset{k-1}{\underset{l=0}{\sum}}\,a_l\,E^{k-l-1},\qquad k\in\{1,\dots,n-1\}.\notag
	\end{align}
	As the linear independence of $\{e,E,\dots,E^{n-1}\}$ is guaranteed by the regularity assumption (see Lemma 6 of \cite{DH}, see also \cite{AK} and \cite{BKM3}), the proof is reduced to showing that the linear independence of $\{X_{(0)},\dots,X_{(n-1)}\}$ amounts to the linear independence of $\{e,E,\dots,E^{n-1}\}$.
	
	Let us consider the matrix $\underline{X}$ whose $j$-th column is given by the vector field $X_{(j-1)}$ and, analogously, the matrix $\underline{E}$ whose $j$-th column is given by the vector field $E^{j-1}$. The determinant of $\underline{X}$ is
		\begin{align}
			\det\underline{X}=&\det[X_{(0)}|X_{(1)}|\dots|X_{(n-1)}]=\det[E^{0}|E^{1}-a_0\,E^0|X_{(2)}|\dots|X_{(n-1)}]\notag\\
			=&\det[E^{0}|E^{1}|X_{(2)}|\dots|X_{(n-1)}]-a_0\,\det[E^{0}|E^0|X_{(2)}|\dots|X_{(n-1)}],
			\notag
		\end{align}
		by multilinearity of the determinant. As the second term vanishes due to the presence of a repeated column, we get
		\begin{align}
			\det\underline{X}=\det[E^{0}|E^{1}|X_{(2)}|\dots|X_{(n-1)}].
			\notag
		\end{align}
		Let us fix $k\in\{2,\dots,n-1\}$ and inductively assume that
		\begin{align}
			\det\underline{X}=\det[E^{0}|\dots|E^{k-1}|X_{(k)}|\dots|X_{(n-1)}].
			\notag
		\end{align}
		Since the $k$-th vector $X_{(k)}$ can be rewritten as
		\[
		X_{(k)}=E^k-\overset{k-1}{\underset{l=0}{\sum}}\,a_l\,E^{k-l-1}
		\]
		for $k\in\{1,\dots,n-1\}$, we have
		\begin{align}
			\det\underline{X}= \, & \, \det[E^{0}|\dots|E^{k-1}|X_{(k)}|\dots|X_{(n-1)}]\notag\\
			= \, & \, \det\Bigg[E^{0}\Bigg|\dots\Bigg|E^{k-1}\Bigg|E^k-\overset{k-1}{\underset{l=0}{\sum}}\,a_l\,E^{k-l-1}\Bigg|X_{(k+1)}\Bigg|\dots\Bigg|X_{(n-1)}\Bigg]\notag\\
			= \, & \, \det[E^{0}|\dots|E^{k-1}|E^k|X_{(k+1)}|\dots|X_{(n-1)}]\notag\\
			\, & \, -\overset{k-1}{\underset{l=0}{\sum}}\,a_l\,\det[E^{0}|\dots|E^{k-1}|E^{k-l-1}|X_{(k+1)}|\dots|X_{(n-1)}].
			\notag
		\end{align}
		As each term in the second sum vanishes due to the presence of repeated columns in the corresponding matrix, we are left with 
		\begin{align}
			\det\underline{X}=\det[E^{0}|\dots|E^{k-1}|E^k|X_{(k+1)}|\dots|X_{(n-1)}].
			\notag
		\end{align}
		Thus,
		\begin{align}
			\det\underline{X}=\det[E^{0}|\dots|E^k|X_{(k+1)}|\dots|X_{(n-1)}]
			\notag
		\end{align}
		for each $k\in\{1,\dots,n-1\}$. In particular, for $k=n-1$ we have
		\begin{align}
			\det\underline{X}=\det[E^{0}|\dots|E^{n-1}]=\det\underline{E}.
			\notag
		\end{align}
\end{proof}

\subsection{Technical Lemmas}
Let $L=E\circ$ be regular, having $r$ Jordan blocks of sizes $m_1,\dots,m_r$, with $\overset{r}{\underset{\alpha=1}{\sum}}\,m_\alpha=n$. Furthermore,  let
\[
\{u^{i(\alpha)}\,|\,i\in\{1,\dots,m_\alpha\}\}_{\alpha\in\{1,\dots,r\}}=\{u^{j}\}_{j\in\{1,\dots,n\}}
\]
denote canonical coordinates. From \cite{LP23}, we know that
\begin{equation}
	L_{i(\alpha)}^{j(\gamma)} = \begin{dcases}
		u^{(j-i+1)(\alpha)}, &  \text{ if } \gamma = \alpha \text{ and } j \geq i,\\
		0, & \text{ otherwise},\\
	\end{dcases}
	\label{eq:diffL}
\end{equation}
and 
\begin{equation}
	\partial_{j(\gamma)} L_{i(\alpha)}^{s(\beta)} = \begin{dcases}
		1, & \text{ if } \gamma = \alpha = \beta \, \text{ and } s = i+j-1, \\
		0, & \text{ otherwise}.
	\end{dcases}\label{eq:diffL2}
\end{equation}
We want to study condition \eqref{eq:ddainL} for a function $a_0$($u^1,\dots,u^n$).
\begin{remark}
	Firstly, note that, given the regularity assumption, \eqref{eq:ddainL} is equivalent to 
	\begin{equation}
		L^s_i \partial_s \partial_j a_0 - L^s_j \partial_s \partial_i a_0 = 0, \quad i,j \in \{1, \dots, n\}.
		\label{eq:ddasimpl}
	\end{equation}
	Let us consider $i$, $j$ as associated to the $\alpha$-th, $\beta$-th block respectively, for some ${\alpha,\beta\in\{1,\dots,r\}}$. When $\alpha\neq\beta$, such equivalence immediately follows from \eqref{eq:diffL2}. Thus, let $\alpha = \beta$. We drop the greek indices to ease the notation, as only coordinates associated to one block are involved. Then,
	\begin{equation}
		\begin{aligned}
			\,  \,   \sum_{s \geq i}^{}\partial_s a_0 \partial_j u^{s-i+1} -  \sum_{s \geq j}^{}\partial_s a_0 \partial_i u^{s-j+1}  
			&= \,   \,  \sum_{s \geq i}^{}\partial_s a_0 \delta_{j, s-i+1} -  \sum_{s \geq j}^{}\partial_s a_0\delta_{i, s-j+1 } \\ &= \,  \, \partial_{i+j-1}a_0 - \partial_{i+j-1}a_0 
			= \,  \,  0,
		\end{aligned}
	\end{equation}
	as $s=i+j-1 \geq i$ and  $s=i+j-1 \geq j$. According to the double-index notation,  and taking into account \eqref{eq:diffL},  \eqref{eq:ddasimpl} may be rewritten as
	\begin{equation}
		\begin{aligned}
		&\overset{m_\alpha}{\underset{s=i+1}{\sum}}\,u^{(s-i+1)(\alpha)}\partial_{s(\alpha)}\partial_{j(\beta)}a_0 - \overset{m_\beta}{\underset{s=j+1}{\sum}}\,u^{(s-j+1)(\beta)}\partial_{s(\beta)}\partial_{i(\alpha)}a_0\\
		&+\big(u^{1(\alpha)}-u^{1(\beta)}\big)\partial_{i(\alpha)}\partial_{j(\beta)}a_0=0,
		\end{aligned}
		\label{eq:ddasimpl_bis}
	\end{equation}
	for $\alpha,\beta\in\{1,\dots,r\}$ and $i\in\{1,\dots,m_\alpha\}$, $j\in\{1,\dots,m_\beta\}$.
\end{remark}

\subsection{The case of one Jordan block}
Let us start by considering the situation in which $L=E\circ$ has a single Jordan block, that is $r=1$. In this case, $m_1=n$ and $u^{i(1)}=u^i$ for each $i\in\{1,\dots,n\}$. In order to ease the notation, we will drop greek indices. Condition \eqref{eq:ddasimpl_bis} then becomes
\begin{equation}
	\overset{n}{\underset{s=i+1}{\sum}}\,u^{s-i+1}\partial_{s}\partial_{j}a_0 - \overset{n}{\underset{s=j+1}{\sum}}\,u^{s-j+1}\partial_{s}\partial_{i}a_0 = 0,\qquad i,j\in\{1,\dots,n\}.
	\label{eq:ddasimpl_bis_J1}
\end{equation}
Condition \eqref{eq:ddasimpl_bis_J1} implies the two relations described in Proposition \ref{Prop1_statement} and Proposition \ref{Prop2_statement} below.
\begin{proposition}\label{Prop1_statement}
	Condition \eqref{eq:ddasimpl_bis_J1} implies that
	\begin{equation}
		\partial_i\partial_ja_0=0,\qquad\text{when }n-i-j\leq-2.
		\label{Prop1}
	\end{equation}
\end{proposition}
\begin{proof}
	We prove this statement by induction over $j\in\{2,\dots,n\}$, starting from $j=n$. By choosing $j=n$ in \eqref{eq:ddasimpl_bis_J1} we get
	\begin{equation}
		\overset{n}{\underset{s=i+1}{\sum}}\,u^{s-i+1}\partial_{s}\partial_{n}a_0 = 0,\qquad i\in\{1,\dots,n\}.
		\label{eq:ddasimpl_bis_J1_jequalsn}
	\end{equation}
	Setting $i=n-1$ in \eqref{eq:ddasimpl_bis_J1_jequalsn} gives $u^{2}\partial_{n}\partial_{n}a_0 = 0$, implying $\partial_{n}\partial_{n}a_0 = 0$. Let us fix $q\in\{2,\dots,n-1\}$ and assume
	\begin{equation}
		\partial_i\partial_na_0=0,\qquad i\in\{q+1,\dots,n\}.
		\label{eq:ddasimpl_bis_J1_jequalsn_indhpi}
	\end{equation}
	By choosing $i=q-1$ in \eqref{eq:ddasimpl_bis_J1_jequalsn} we get
	\begin{equation}
		0=\overset{n}{\underset{s=q}{\sum}}\,u^{s-q+2}\partial_{s}\partial_{n}a_0\overset{\eqref{eq:ddasimpl_bis_J1_jequalsn_indhpi}}{=}u^{2}\partial_{q}\partial_{n}a_0,
	\end{equation}
	implying $\partial_{q}\partial_{n}a_0=0$ and hence 
	\begin{equation}
		\partial_i\partial_na_0=0,\qquad i\in\{2,\dots,n\}.
		\label{eq:ddasimpl_bis_J1_jequalsn_alli}
	\end{equation}
	Let us now fix $p\in\{2,\dots,n-1\}$ and assume
	\begin{equation}
		\partial_i\partial_ja_0=0,\qquad i\in\{n-p+2,\dots,n\},
		\label{eq:ddasimpl_bis_J1_jequalsn_indhpj}
	\end{equation}
	for each $j\in\{p+1,\dots,n\}$. In the following, we show that $\partial_i\partial_pa_0=0$ for each ${i\in\{n-p+2,\dots,n\}}$. By choosing $j=p$ in \eqref{eq:ddasimpl_bis_J1} we get
	\begin{equation}
		0=\overset{n}{\underset{s=i+1}{\sum}}\,u^{s-i+1}\partial_{s}\partial_{p}a_0 - \overset{n}{\underset{s=p+1}{\sum}}\,u^{s-p+1}\partial_{s}\partial_{i}a_0 \overset{\eqref{eq:ddasimpl_bis_J1_jequalsn_indhpj}}{=}\overset{n}{\underset{s=i+1}{\sum}}\,u^{s-i+1}\partial_{s}\partial_{p}a_0,\quad i\in\{1,\dots,n\},
		\label{eq:ddasimpl_bis_J1_iequalsp}
	\end{equation}
	and by further setting $i=n-1$ we get $u^{2}\partial_{n}\partial_{p}a_0=0$, implying $\partial_{n}\partial_{p}a_0=0$. Let us now fix $q\in\{n-p+2,\dots,n-1\}$ and assume
	\begin{equation}
		\partial_i\partial_pa_0=0,\qquad i\in\{q+1,\dots,n\}.
		\label{eq:ddasimpl_bis_J1_jequalsn_indhpj_indhpi}
	\end{equation}
	By choosing $i=q-1$ in \eqref{eq:ddasimpl_bis_J1_iequalsp},
	\begin{equation}
		0=\overset{n}{\underset{s=q}{\sum}}\,u^{s-q+2}\partial_{s}\partial_{p}a_0\overset{\eqref{eq:ddasimpl_bis_J1_jequalsn_indhpj_indhpi}}{=}u^{2}\partial_{q}\partial_{p}a_0,
	\end{equation}
	implying $\partial_{q}\partial_{p}a_0=0$ and hence proving that $\partial_i\partial_pa_0=0$ for each ${i\in\{n-p+2,\dots,n\}}$.	
\end{proof}

 In order to prove Proposition \ref{Prop2_statement}, we need Lemma \ref{Lemma1_statement}.

\begin{lemma}\label{Lemma1_statement}
	Let $h\geq-1$. If
	\begin{align}
		\partial_{i+1}\partial_{n-i-h-1}a_0=\partial_{n-i-h}\partial_{i}a_0,\qquad i\in\{1,\dots,n\},
		\label{Lemma1_hp}
	\end{align}
	then
	\begin{equation}
		\partial_i\partial_ja_0=\partial_{\overline{i}}\partial_{\overline{j}}a_0,\qquad\text{when }i+j=\overline{i}+\overline{j},\qquad n-i-j=h.
		\label{Lemma1_th}
	\end{equation}
\end{lemma}
\begin{proof}
	Let us fix $h\geq-1$. By relabelling $a:=i$, condition \eqref{Lemma1_hp} reads
	\begin{align}
		\partial_{a+1}\partial_{n-a-h-1}a_0=\partial_{n-a-h}\partial_{a}a_0,\qquad a\in\{1,\dots,n\}.
		\label{Lemma1_hp_a}
	\end{align}
	Let us consider $i,j,\overline{i},\overline{j}\in\{1,\dots,n\}$ such that $i+j=\overline{i}+\overline{j}$. The requirement $n-i-j=h$ forces $j$ and $\overline{j}$ to be $j=n-i-h$ and $\overline{j}=n-\overline{i}-h$, respectively. Without loss of generality, let us consider $\overline{i}<i$.
	
	By choosing $a=i-1$ in \eqref{Lemma1_hp_a}, we get
	\begin{align}
		\partial_{i}\partial_{n-i-h}a_0=\partial_{n-i-h+1}\partial_{i-1}a_0,
		\notag
	\end{align}
	proving \eqref{Lemma1_th} for $\overline{i}=i-1$.
	
	Let us fix $s\geq2$ and inductively assume \eqref{Lemma1_th} for each $\overline{i}\geq i-s+1$. We need to show that \eqref{Lemma1_th} holds for $\overline{i}=i-s$, namely
	\begin{align}
		\partial_i\partial_{n-i-h}a_0=\partial_{i-s}\partial_{n-i+s-h}a_0.
		\label{Prop2_proof_endgoal}
	\end{align}
	
	Let us fix $\overline{i}=i-s$. By the inductive assumption for $\overline{i}=i-s+1$, we have
	\begin{align}
		\partial_i\partial_{n-i-h}a_0&=\partial_{i-s+1}\partial_{n-i+s-h-1}a_0.\notag
	\end{align}
	By choosing $a=i-s$ in \eqref{Lemma1_hp_a}, we have
	\begin{align}
		\partial_{i-s+1}\partial_{n-i+s-h-1}a_0&=\partial_{n-i+s-h}\partial_{i-s}a_0.
		\notag
	\end{align}
	Thus, we get \eqref{Prop2_proof_endgoal}, as required.
\end{proof}
\begin{proposition}\label{Prop2_statement}
	Condition \eqref{eq:ddasimpl_bis_J1} implies
	\begin{equation}
		\partial_i\partial_ja_0=\partial_{\overline{i}}\partial_{\overline{j}}a_0,\qquad\text{when }i+j=\overline{i}+\overline{j},\qquad n-i-j\geq-1.
		\label{Prop2}
	\end{equation}
\end{proposition}
\begin{proof}
	Let us consider $i,j\in\{1,\dots,n\}$ such that $n-i-j\geq-1$. By means of Proposition \ref{Prop1_statement},  \eqref{eq:ddasimpl_bis_J1} becomes
	\begin{equation}
		\overset{n-j+1}{\underset{s=i+1}{\sum}}\,u^{s-i+1}\partial_{s}\partial_{j}a_0 - \overset{n-i+1}{\underset{s=j+1}{\sum}}\,u^{s-j+1}\partial_{s}\partial_{i}a_0 = 0,\qquad i,j\in\{1,\dots,n\},
		\label{eq:ddasimpl_bis_J1_bis}
	\end{equation}
	that is,
	\begin{equation}
		u^2\big(\partial_{i+1}\partial_{j}a_0-\partial_{j+1}\partial_{i}a_0\big)+\overset{n-j+1}{\underset{s=i+2}{\sum}}\,u^{s-i+1}\partial_{s}\partial_{j}a_0 - \overset{n-i+1}{\underset{s=j+2}{\sum}}\,u^{s-j+1}\partial_{s}\partial_{i}a_0 = 0,\quad i,j\in\{1,\dots,n\},
		\label{eq:ddasimpl_bis_J1_tris}
	\end{equation}
	where the last two sums only survive when $n\geq i+j+1$.
	
	By choosing $j=n-i$ in \eqref{eq:ddasimpl_bis_J1_tris}, we get
	\begin{equation}
		\partial_{i+1}\partial_{n-i}a_0=\partial_{n-i+1}\partial_{i}a_0,\qquad i\in\{1,\dots,n\},
		\label{eq:ddasimpl_bis_J1_tris_j=n-i}
	\end{equation}
	which, by Lemma \ref{Lemma1_statement}, implies 
	\begin{equation}
		\partial_i\partial_ja_0=\partial_{\overline{i}}\partial_{\overline{j}}a_0,\qquad\text{when }i+j=\overline{i}+\overline{j},\qquad n-i-j=-1.
		\notag
	\end{equation}
	
	Let us now fix $h\geq-1$ and inductively assume 
	\begin{equation}
		\partial_i\partial_ja_0=\partial_{\overline{i}}\partial_{\overline{j}}a_0,\qquad\text{when }i+j=\overline{i}+\overline{j},\qquad n-i-j\leq h-1.
		\label{Prop2_indhp}
	\end{equation}
	We need to show that
	\begin{equation}
		\partial_i\partial_ja_0=\partial_{\overline{i}}\partial_{\overline{j}}a_0,\qquad\text{when }i+j=\overline{i}+\overline{j},\qquad n-i-j=h.
		\label{Prop2_indth}
	\end{equation}

	By choosing $j=n-i-h-1$ in \eqref{eq:ddasimpl_bis_J1_tris}, we get
	\begin{align}
		0&=u^2\big(\partial_{i+1}\partial_{n-i-h-1}a_0-\partial_{n-i-h}\partial_{i}a_0\big)\notag\\& \quad +\overset{i+h+2}{\underset{s=i+2}{\sum}}\,u^{s-i+1}\partial_{s}\partial_{n-i-h-1}a_0 - \overset{n-i+1}{\underset{s=n-i-h+1}{\sum}}\,u^{s-n+i+h+2}\partial_{s}\partial_{i}a_0\notag\\
		&=u^2\big(\partial_{i+1}\partial_{n-i-h-1}a_0-\partial_{n-i-h}\partial_{i}a_0\big)\notag\\& \quad +\overset{i+h+2}{\underset{s=i+2}{\sum}}\,u^{s-i+1}\partial_{s}\partial_{n-i-h-1}a_0 - \overset{i+h+2}{\underset{t=i+2}{\sum}}\,u^{t-i+1}\partial_{t+n-2i-h-1}\partial_{i}a_0\notag\\
		&=u^2\big(\partial_{i+1}\partial_{n-i-h-1}a_0-\partial_{n-i-h}\partial_{i}a_0\big)\notag\\& \quad +\overset{i+h+2}{\underset{s=i+2}{\sum}}\,u^{s-i+1}\big(\partial_{s}\partial_{n-i-h-1}a_0-\partial_{s+n-2i-h-1}\partial_{i}a_0\big),\qquad i\in\{1,\dots,n\},
		\label{eq:ddasimpl_bis_J1_tris_j=n-i-h-1}
	\end{align}
	where $\partial_{s}\partial_{n-i-h-1}a_0=\partial_{s+n-2i-h-1}\partial_{i}a_0$ for each $s\geq i+2$ by \eqref{Prop2_indhp}. This gives
	\begin{align}
		\partial_{i+1}\partial_{n-i-h-1}a_0=\partial_{n-i-h}\partial_{i}a_0,\qquad i\in\{1,\dots,n\},
	\end{align}
	which, by Lemma \ref{Lemma1_statement}, implies
	\begin{equation}
		\partial_i\partial_ja_0=\partial_{\overline{i}}\partial_{\overline{j}}a_0,\qquad\text{when }i+j=\overline{i}+\overline{j},\qquad n-i-j=h,
		\notag
	\end{equation}
	namely \eqref{Prop2_indth}.
\end{proof}

Together, the conditions from Proposition \ref{Prop1_statement} and Proposition \ref{Prop2_statement} imply \eqref{eq:ddasimpl_bis_J1}.
\begin{proposition}\label{Prop3_statement}
	If \eqref{Prop1} and \eqref{Prop2} hold, then so does \eqref{eq:ddasimpl_bis_J1}.	
\end{proposition}
\begin{proof}
	Let us assume that \eqref{Prop1} and \eqref{Prop2} hold and let us fix $i,j\in\{1,\dots,n\}$. We have
	\begin{align}
		&\overset{n}{\underset{s=i+1}{\sum}}\,u^{s-i+1}\partial_{s}\partial_{j}a_0 - \overset{n}{\underset{s=j+1}{\sum}}\,u^{s-j+1}\partial_{s}\partial_{i}a_0\notag\\
		& \overset{\eqref{Prop1}}{=}\overset{n-j+1}{\underset{s=i+1}{\sum}}\,u^{s-i+1}\partial_{s}\partial_{j}a_0 - \overset{n-i+1}{\underset{s=j+1}{\sum}}\,u^{s-j+1}\partial_{s}\partial_{i}a_0\notag\\
		& \ \ =\overset{n-i-j+1}{\underset{t=1}{\sum}}\,u^{t+1}\partial_{t+i}\partial_{j}a_0 - \overset{n-i-j+1}{\underset{t=1}{\sum}}\,u^{t+1}\partial_{t+j}\partial_{i}a_0\notag\\
		& \ \ =\overset{n-i-j+1}{\underset{t=1}{\sum}}\,u^{t+1}\big(\partial_{t+i}\partial_{j}a_0-\partial_{t+j}\partial_{i}a_0\big)\overset{\eqref{Prop2}}{=}0,
		\notag
	\end{align}
 where the second equality is obtained by respectively substituting $s = t -i$ and $s = t-j$ in the two sums.
\end{proof}
\begin{remark}\label{rmk_equivsys}
	As a consequence of Propositions \ref{Prop1_statement}, \ref{Prop2_statement} and \ref{Prop3_statement}, condition \eqref{eq:ddasimpl_bis_J1} is equivalent to the system
 \begin{subequations}\label{Props12}
	\begin{equation}\label{Props12a}
\partial_i\partial_ja_0=\partial_{\overline{i}}\partial_{\overline{j}}a_0,\qquad \text{when }i+j=\overline{i}+\overline{j},
	\end{equation}
 where \begin{equation}\label{Props12b}
     	\partial_i\partial_ja_0=0,\qquad \text{for }n-i-j\leq-2,
 \end{equation}
 \end{subequations}
	as given by Propositions \ref{Prop1_statement},  \ref{Prop2_statement}. Thus, a function $a_0$ is a solution to \eqref{eq:ddasimpl_bis_J1} if and only if it solves the system \eqref{Props12}.
\end{remark}

So far, we have studied condition \eqref{eq:ddasimpl_bis_J1} for a function $a_0$ of the canonical coordinates $u^1,\dots,u^n$. Let us now specify the size of the Jordan block to be $n$ and denote $a_0 \equiv a_0^{(n)}$. When doing so we will write $a_0^{(k)}$ and $L^{(k)}$ (instead of $a_0$ and $L$, respectively) to denote the case of a Jordan block of generic size $k$. The following proposition establishes a link between the general solutions of \eqref{mainEQ} for the case of a single Jordan block of different sizes.

\begin{proposition}
     The derivative with respect to $u^k$, $\partial_k a_0^{(n)}$,  of a general solution to \eqref{eq:ddainL} for $L$ having a single Jordan block of size $n$, $a_0^{(n)}$, is the general solution of \eqref{eq:ddainL} for $L$ having a single Jordan block of size $n-k+1$:
	\begin{equation}
		d \cdot d_{L^{(n-k+1)}}(\partial_k a_0^{(n)})=\partial_k(d \cdot d_{L^{(n)}} a_0^{(n)})=0.
		\label{prop:dimensionbridge_eq}
	\end{equation}
 \label{prop:dimensionbridge}
\end{proposition}

    In order to prove Proposition \ref{prop:dimensionbridge}, we need the following lemma.

\begin{lemma} 
	The derivative with respect to $u^k$ of a function $f(u^1,\dots,u^n)$, $\partial_k f$, satisfies the relation
	\begin{equation}
		\begin{aligned}
			\partial_k(d \cdot d_{L^{(n)}} f)_{ij}&=\big( d \cdot d_{L^{(n-k+1)}}(\partial_k f)\big)_{ij}+\partial_i\partial_{k+j-1}f-\partial_j\partial_{k+i-1}f\\
			&+\overset{n}{\underset{s=n-k+2}{\sum}}u^{s-j+1}\partial_i\partial_s(\partial_k f)-\overset{n}{\underset{s=n-k+2}{\sum}}u^{s-i+1}\partial_j\partial_s(\partial_k f)
		\end{aligned}
		\label{switch_genf}
	\end{equation}
	for every $i,j\in\{1,\dots,n\}$.
	\label{lemma:diffprop}
\end{lemma}

\begin{proof}
    	Let us consider a generic function $f$ of $u^1,\dots,u^n$ and fix ${i,j\in\{1,\dots,n\}}$. We have
	\begin{equation}
		\begin{aligned}
			\partial_k(d \cdot d_{L^{(n)}} f)_{ij}&=\partial_k\Bigg(\overset{n}{\underset{s=j+1}{\sum}}u^{s-j+1}\partial_i\partial_s f-\overset{n}{\underset{s=i+1}{\sum}}u^{s-i+1}\partial_j\partial_s f\Bigg)\\
			&=\partial_i\partial_{k+j-1} f-\partial_j\partial_{k+i-1} f\\
			& \quad +\overset{n}{\underset{s=j+1}{\sum}}u^{s-j+1}\partial_i\partial_s(\partial_k f)-\overset{n}{\underset{s=i+1}{\sum}}u^{s-i+1}\partial_j\partial_s (\partial_k f).
		\end{aligned}
		\notag
	\end{equation}
	On the other hand, 
	\begin{equation}
		\begin{aligned}
			\big( d \cdot d_{L^{(n-k+1)}}(\partial_k f)\big)_{ij}&=\overset{n-k+1}{\underset{s=j+1}{\sum}}u^{s-j+1}\partial_i\partial_s(\partial_k f)-\overset{n-k+1}{\underset{s=i+1}{\sum}}u^{s-i+1}\partial_j\partial_s(\partial_k f).
		\end{aligned}
		\notag
	\end{equation}
	This implies \eqref{switch_genf}.
\end{proof}

We are now ready to prove Proposition \ref{prop:dimensionbridge}.

 \begin{proof}
	Let us consider a solution $a_0^{(n)}$ of \eqref{eq:ddainL} for $L \equiv L^{(n)}$ having a single Jordan block of size $n$. As pointed out in Remark \ref{rmk_equivsys}, this requirement is equivalent to satisfying system \eqref{Props12}. In this particular case, we get
	\[
	\partial_i\partial_{k+j-1} a_0^{(n)}-\partial_j\partial_{k+i-1} a_0^{(n)}=0,
	\]
	by \eqref{Props12a}, and
	\[
	\overset{n}{\underset{s=n-k+2}{\sum}}u^{s-j+1}\partial_i\partial_s(\partial_k a_0^{(n)})\,\,\,-\overset{n}{\underset{s=n-k+2}{\sum}}u^{s-i+1}\partial_j\partial_s (\partial_k a_0^{(n)})=0,
	\]
	by  \eqref{Props12b} as $n-s-k \leq n-(n-k+2)-k = -2$.  Thus, using Lemma \ref{lemma:diffprop},
	\[
	d \cdot d_{L^{(n-k+1)}}(\partial_k a_0^{(n)})=\partial_k(d \cdot d_{L^{(n)}} a_0^{(n)}).
	\]
	Finally, since $d \cdot d_{L^{(n)}} a_0^{(n)}=0$, we obtain  \eqref{prop:dimensionbridge_eq}.
\end{proof}

\begin{lemma}
	A solution to \eqref{eq:ddasimpl_bis_J1} is a polynomial in $u^2,\dots,u^n$ with functions of $u^1$ as coefficients.
    \label{lemma:poly}
\end{lemma}

\begin{proof}
	We proceed by induction over $n$, starting from $n=2$ as base case. By \eqref{Props12b}, we have
	\[
	\partial_2\partial_2 a_0^{(2)}=0,
	\]
	implying that $a_0^{(2)}$ is a polynomial in $u^2$ with functions of $u^1$ as coefficients. For each $m\geq2$, let us denote by $\mathcal{P}^{2\dots m}_1$ the set of polynomials in $u^2,\dots,u^m$ with coefficients being functions of $u^1$. Let us fix $N\geq3$ and inductively assume $a_0^{(n)}\in\mathcal{P}^{2\dots n}_{1}$ for each $n\leq N-1$. We need to show that $a_0^{(N)}\in\mathcal{P}^{2\dots N}_{1}$.
	
	By Proposition \ref{prop:dimensionbridge}, we know that $\partial_k a_0^{(N)}$ satisfies $d\cdot d_{L^{(n-k+1)}}\big(\partial_k a_0^{(N)}\big)=0$ for each $k\in\{1,\dots,N\}$. In particular, for each $k\geq2$ we get $N-k+1\leq N-1$ and then, by the induction hypothesis, $\partial_k a_0^{(N)}\in\mathcal{P}^{2\dots N-k+1}_1\subset\mathcal{P}^{2\dots N}_1$. For $k=2$, we get $\partial_2 a_0^{(N)}\in\mathcal{P}^{2\dots N}_1$, implying $a_0^{(N)}\in\mathcal{P}^{2\dots N}_1$.
\end{proof}
\begin{proposition}
    If it exists, the general solution of $d \cdot d_L a_0^{(n)}$ for $L$ having a single Jordan block of size $n$ is of the form  
    \begin{equation}
        a_0^{(n)} = F_n(u^1) + \sum_{s>0} \dfrac{1}{s!} \sum_{k_1=2}^n \cdots \sum_{k_s=2}^{n} u^{k_1}\cdots u^{k_s}\left(\dfrac{d}{d u^1}\right)^{(s-1)}F_{n+s-\sum_{j=1}^s k_j}(u^1),
        \label{eq:gensol}
    \end{equation}
    where $F_1, \dots, F_n$ are arbitrary functions of $u^1$, and $F_{n+s-\sum_{j=1}^{s}k_j} = 0$ whenever ${n+s-\sum_{j=1}^{s}k_j \leq 0}$.
    \label{prop:gensol}
\end{proposition}
\begin{proof}
By Lemma \ref{lemma:poly}, we may consider $a_0^{(n)}$ as a polynomial in $u^2, \dots, u^n$ while treating $u^1$ as a parameter.Thus,  let's write  $a_0^{(n)}$ as
\begin{equation}
    a_0^{(n)} = F_n + \sum_{i=2}^{n}u^iF_{n-i+1} + \big\{\text{terms of higher orders in } u^2, \dots, u^n\big\}.
\end{equation}
Let $P$ denote the point $(u^2, \dots, u^n) = (0, \dots, 0)$. Notice that for each $k\in\{2,\dots,n\}$
\begin{equation}
    \partial_1 \partial_k a_0^{(n)}\big|_P = F_{n-k+1}'(u^1)
\end{equation}
and, by successive use of Proposition \ref{Prop2_statement},
\begin{equation}
	\partial_1 \partial_k a_0^{(n)}\big|_P = \partial_2 \partial_{k-1} a_0^{(n)}\big|_P = \cdots = \partial_j \partial_{k-j+1} a_0^{(n)}\big|_P,\qquad j\in\{1,\dots,k\}.
\end{equation}
This implies that
\begin{equation}
    [u^i u^j]\,a_0^{(n)} = F_{n-i-j+2}'(u^1),
\end{equation}
after relabelling $k-j \equiv i-1$ as $n-k+1 = n-(k-j)-j+1 = n-i-j+2$, where $[u^i u_j]\,a_0^{(n)}$ denotes the coefficient of $u^i u^j$ in $a_0^{(n)}$.
MIn more generality, we have that 
\begin{equation}
   F_{n-k+1}^{(s-1)}(u^1) = \partial_1^{(s-1)}\partial_k a_0^{(n)}\big|_P = \partial_{i_1}\cdots \partial_{i_{s-1}}\partial_{k-\sum_{j=1}^{s-1}i_j +s-1} a_0^{(n)}\big|_P. 
\end{equation}
Hence, by taking the Taylor expansion near the point $P$,
\begin{equation}
    a_0^{(n)} = F_n(u^1) + \sum_{s>0}\dfrac{1}{s!} \sum_{k_1=2}^n \cdots \sum_{k_s=2}^{n}(\partial_{k_1} \cdots \partial_{k_s} a_0^{(n)})\big|_P\,u^{k_1}\dots u^{k_s},
\end{equation}
we obtain \eqref{eq:gensol}, after relabelling $k= \sum_{j=1}^{s}k_j-s+1$. 

Finally, by means of Proposition \ref{Prop1_statement}, we have that $\partial_1 \partial_{\sum_{j=1}^{s}{k_j}-s+1} a_0^{(n)} = 0$ whenever ${n-1-\sum_{j=1}^{s}k_j+s-1+2 = n-\sum_{j=1}^sk_j +s \leq 0}$.

\end{proof}

\begin{corollary}
    For a function $a_0^{(n)}$ of the form \eqref{eq:gensol}, we have 
    \begin{equation}
            \partial_{n+1-k}a_0^{(n)}(F_1, \dots, F_n) = \begin{dcases}
            a_0^{(n)}(F_1', \dots, F_n'), & \text{ for } k = n,\\
                a_0^{(k)}(F_1', \dots, F_{k-1}', F_k), & \text{ otherwise.}
            \end{dcases}  
    \end{equation}
    \label{corollary:dimbridgeFs}
\end{corollary}
\begin{proof}
    The case of $k=n$ follows immediately from \eqref{eq:gensol}:
    \begin{equation}
    \begin{aligned}
        \partial_1 a_0^{(n)}(F_1, \dots, F_n) & = F_n'(u^1) + \sum_{s>0}\dfrac{1}{s!}\sum_{k_1=2}^{n} \cdots \sum_{k_s=2}^{n} \partial_1^{(s-1)} \left(F_{n+s-\sum_{j=1}^{s}k_j}'(u^1)\right) \\ & \equiv a_0^{(n)}(F_1', \dots, F_n').
         \end{aligned}
    \end{equation}
   We now consider the case $k\leq n-1$. Equivalently, we prove
   \begin{equation}
   	\partial_{k}a_0^{(n)}(F_1, \dots, F_n)=a_0^{(n-k+1)}(F_1', \dots, F_{n-k}', F_{n-k+1})
   	\label{corollary:dimbridgeFs_proofhelp}
   \end{equation}
	for $k\neq 1$. 
	Let us consider $k\geq2$. By \eqref{eq:gensol}, the left-hand side of \eqref{corollary:dimbridgeFs_proofhelp} reads
	\begin{align}
		\partial_{k}a_0^{(n)}&=\partial_{k}a_0^{(n)}(F_1, \dots, F_n)\notag\\
		&=\partial_k\bigg(F_n(u^1) + \sum_{s>0} \dfrac{1}{s!} \sum_{i_1=2}^n \cdots \sum_{i_s=2}^{n} u^{i_1}\cdots u^{i_s}\left(\tfrac{d}{d u^1}\right)^{(s-1)}F_{n+s-\sum_{j=1}^s i_j}(u^1)\bigg).
		\notag
	\end{align}
	In each term $[u^{i_1}\cdots u^{i_s}]\,a_0^{(n)}$ for ${i_1,\dots,i_s\in\{2,\dots,n\}}$, $u^k$ may be present. Let's denote by $m$ the number of times that $u^k$ appears in $u^{i_1}\cdots u^{i_s}$, namely let us set ${m=\overset{s}{\underset{j=1}{\sum}}\delta^{k}_{i_j}\in\{0,\dots,s\}}$. Given the symmetry of  \eqref{eq:gensol} in the indices $i_1,\dots,i_s$, for each $s>0$ we can rearrange the corresponding sums as
	\begin{align}
		& \quad \sum_{i_1=2}^n \cdots \sum_{i_s=2}^{n} u^{i_1}\cdots u^{i_s}\left(\tfrac{d}{d u^1}\right)^{(s-1)}F_{n+s-\sum_{j=1}^s i_j}(u^1)\notag\\
		&=\overset{s}{\underset{m=0}{\sum}}\binom{s}{m}\sum_{i_1=2}^n \cdots \sum_{i_{s-m}=2}^{n} u^{i_1}\cdots u^{i_{s-m}}(u^k)^m\left(\tfrac{d}{d u^1}\right)^{(s-1)}F_{n+s-\sum_{j=1}^{s-m} i_j-mk}(u^1)\notag.
	\end{align}
	Then\small
	\begin{align}
		\partial_{k}a_0^{(n)}&=\partial_k\bigg(\sum_{s>0} \dfrac{1}{s!} \overset{s}{\underset{m=0}{\sum}}\binom{s}{m}\sum_{i_1=2}^n \cdots \sum_{i_{s-m}=2}^{n} u^{i_1}\cdots u^{i_{s-m}}(u^k)^m\left(\tfrac{d}{d u^1}\right)^{(s-1)}F_{n+s-\sum_{j=1}^{s-m} i_j-mk}(u^1)\bigg)\notag\\
		&=\sum_{s>0}\overset{s}{\underset{m=1}{\sum}}\tfrac{1}{(m-1)!(s-m)!}\sum_{i_1=2}^n \cdots \sum_{i_{s-m}=2}^{n} u^{i_1}\cdots u^{i_{s-m}}(u^k)^{m-1}\left(\tfrac{d}{d u^1}\right)^{(s-1)}F_{n+s-\sum_{j=1}^{s-m} i_j-mk}(u^1).
		\notag
	\end{align}\normalsize
	By decreasing both of the indices $m$ and $s$ by $1$, we get\small
	\begin{align}
		\partial_{k}a_0^{(n)}&=\sum_{s>0}\overset{s-1}{\underset{m=0}{\sum}}\tfrac{1}{m!(s-m-1)!}\sum_{i_1=2}^n \cdots \sum_{i_{s-m-1}=2}^{n} u^{i_1}\cdots u^{i_{s-m-1}}(u^k)^{m}\left(\tfrac{d}{d u^1}\right)^{(s-1)}F_{n+s-\overset{s-m-1}{\underset{j=1}{\sum}} i_j-mk-k}(u^1)\notag\\
		&=\sum_{s>-1}\overset{s}{\underset{m=0}{\sum}}\tfrac{1}{m!(s-m)!}\sum_{i_1=2}^n \cdots \sum_{i_{s-m}=2}^{n} u^{i_1}\cdots u^{i_{s-m}}(u^k)^{m}\left(\tfrac{d}{d u^1}\right)^{(s)}F_{n+s+1-\overset{s-m}{\underset{j=1}{\sum}} i_j-mk-k}(u^1)\notag\\
		&=F_{n-k+1}(u^1)\notag\\
		& \quad +\sum_{s>0}\overset{s}{\underset{m=0}{\sum}}\tfrac{1}{m!(s-m)!}\sum_{i_1=2}^n \cdots \sum_{i_{s-m}=2}^{n} u^{i_1}\cdots u^{i_{s-m}}(u^k)^{m}\left(\tfrac{d}{d u^1}\right)^{(s-1)}F_{n+s+1-\overset{s-m}{\underset{j=1}{\sum}} i_j-mk-k}'(u^1),
		\notag
	\end{align}\normalsize
	where the sums over $i_1,\dots,i_{s-m}\in\{2,\dots,n\}$ can be truncated at $n-k+1$, as if (at least) one of the indices $i_1,\dots,i_{s-m}\in\{2,\dots,n\}$ is greater or equal than $n-k+2$ we get
	\begin{align}
		n+s+1-\sum_{j=1}^{s-m} i_j-mk-k&=n+s+1-\sum_{j=1}^{s} i_j-k\notag\\
		&\leq n+s+1-(n-k+2)-2(s-1)-k=1-s\leq0.
		\notag
	\end{align}
	By \eqref{eq:gensol}, the right-hand side of \eqref{corollary:dimbridgeFs_proofhelp} reads
	\begin{align}
		&a_0^{(n-k+1)}(F_1', \cdots, F_{n-k}', F_{n-k+1})=F_{n-k+1}(u^1)\notag\\
		&+\sum_{s>0} \dfrac{1}{s!} \sum_{i_1=2}^{n-k+1} \cdots \sum_{i_s=2}^{{n-k+1}} u^{i_1}\cdots u^{i_s}\left(\tfrac{d}{d u^1}\right)^{(s-1)}F_{{n-k+1}+s-\sum_{j=1}^s i_j}'(u^1),
		\notag
	\end{align}
	which becomes
	\begin{align}
		&a_0^{(n-k+1)}(F_1', \dots, F_{n-k}', F_{n-k+1})=F_{n-k+1}(u^1)\notag\\
		&+\sum_{s>0}\overset{s}{\underset{m=0}{\sum}}\tfrac{1}{m!(s-m)!} \sum_{i_1=2}^{n-k+1} \cdots \sum_{i_{s-m}=2}^{{n-k+1}} u^{i_1}\cdots u^{i_{s-m}}(u^k)^m\left(\tfrac{d}{d u^1}\right)^{(s-1)}F_{{n-k+1}+s-\overset{s-m}{\underset{j=1}{\sum}} i_j-mk}'(u^1).
		\notag
	\end{align}
   This proves \eqref{corollary:dimbridgeFs_proofhelp} for $k\geq2$.   
    \end{proof}
\begin{theorem}
    The general solution of $d \cdot d_L a_0^{(n)}=0$, for $L$ having a single Jordan block, is of the form \eqref{eq:gensol}.
    \label{thm:gensol}
\end{theorem}
\begin{proof}
    From Proposition \ref{prop:gensol}, we have that Equation \eqref{eq:main} implies the form \eqref{eq:gensol}, so we are left with proving the converse statement. That is, that the function $a_0^{(n)}=a_0^{(n)}(F_1, \dots, F_n)$ defined by \eqref{eq:gensol} does indeed satisfy \eqref{eq:main}. As pointed out in Remark \ref{rmk_equivsys}, this requirement is equivalent to the conditions of \eqref{Props12}.
    Hence it is sufficient to show that $a_0^{(n)}$, as defined by \eqref{eq:gensol}, satisfies  \begin{subequations}
	\begin{equation*}
\partial_i\partial_ja_0^{(n)}=\partial_{\overline{i}}\partial_{\overline{j}}a_0^{(n)},\qquad \text{when }i+j=\overline{i}+\overline{j},
	\end{equation*}
 where \begin{equation*}
     	\partial_i\partial_ja_0^{(n)}=0,\qquad \text{for }n-i-j\leq-2.
 \end{equation*}
 \end{subequations}

    Let us fix $i,j\in\{1,\dots,n\}$. As the statement  trivially holds for $n=1$, let us assume $n>1$.  Firstly let's consider the case where $i+j\geq n+2$, which implies $i,j\geq2$. By successive use of Corollary \ref{corollary:dimbridgeFs}, we have
    \begin{align}
    	\partial_i\partial_ja_0^{(n)}&=\partial_i\big(a_0^{(n-j+1)}(F_1', \dots, F_{n-j}', F_{n-j+1})\big)\notag\\
    	&=a_0^{(n-i-j+2)}(F_1'', \dots, F_{n-i-j+1}'', F_{n-i-j+2}'),
    	\notag
    \end{align}
	which vanishes, as $n-i-j+2\leq0$. Let us now consider the case where $i+j\leq n+1$. We have three possible subcases:
	\begin{itemize}
		\item[$(a)$] $i=j=1$;
		\item[$(b)$] $i=1$ and $j\geq2$, or $i\geq2$ and $j=1$;
		\item[$(c)$] $i,j\geq2$.
	\end{itemize}
	\textbf{Subcase (a).} If $i=j=1$ we must have that $\overline{i}=\overline{j}=1$, which trivially guarantees the fulfillment of \eqref{Props12}.
	\\\textbf{Subcase (b).} Let's consider $i=1$, $j\geq2$. If $\overline{i}=1$ (or, equivalently, $\overline{j}=1$), then \eqref{Props12} is trivially satisfied. Let us then assume  $\overline{i},\overline{j}\geq2$. By successive use of Corollary \ref{corollary:dimbridgeFs}, as above, we have
	\begin{align}
		\partial_{\overline{i}}\partial_{\overline{j}}a_0^{(n)}&=a_0^{(n-{\overline{i}}-{\overline{j}}+2)}(F_1'', \dots, F_{n-{\overline{i}}-{\overline{j}}+1}'', F_{n-{\overline{i}}-{\overline{j}}+2}')\notag\\
		&=a_0^{(n-i-j+2)}(F_1'', \dots, F_{n-i-j+1}'', F_{n-i-j+2}')\notag\\
		&\overset{i=1}{=}a_0^{(n-j+1)}(F_1'', \dots, F_{n-j}'', F_{n-j+1}').
		\notag
	\end{align}
	On the other hand, again by Corollary \ref{corollary:dimbridgeFs}, we get
	\begin{align}
		\partial_i\partial_ja_0^{(n)}&=\partial_1\big(a_0^{(n-j+1)}(F_1', \dots, F_{n-j}', F_{n-j+1})\big)\notag\\
		&=a_0^{(n-j+1)}(F_1'', \dots, F_{n-j}'', F_{n-j+1}')=\partial_{\overline{i}}\partial_{\overline{j}}a_0^{(n)}.
		\notag
	\end{align}
	\textbf{Subcase (c).} Let us consider $i,j\geq2$. Without loss of generality, we may impose that $\overline{i},\overline{j}\geq2$ (as if one is equal to $1$ we fall into subcase (b) by inverting the roles of $(i,j)$ and $(\overline{i},\overline{j})$). By Corollary \ref{corollary:dimbridgeFs}, we immediately have
	\begin{align}
		\partial_{\overline{i}}\partial_{\overline{j}}a_0^{(n)}&=a_0^{(n-{\overline{i}}-{\overline{j}}+2)}(F_1'', \dots, F_{n-{\overline{i}}-{\overline{j}}+1}'', F_{n-{\overline{i}}-{\overline{j}}+2}')\notag\\
		&=a_0^{(n-i-j+2)}(F_1'', \dots, F_{n-i-j+1}'', F_{n-i-j+2}')=\partial_i\partial_ja_0^{(n)}.
		\notag
	\end{align}     
\end{proof}
    
\begin{example}
In this example we reconstruct $a_0^{(4)}$ using \eqref{Props12}, and show that it satisfies Theorem \ref{thm:gensol} and Corollary \ref{corollary:dimbridgeFs}.

\vspace{0.3em}

Let $P$ denote the point $(u^2, \dots, u^n) = (0, \dots, 0)$. 
By \eqref{Props12b}, $\partial_i\partial_j a_0^{(4)} = 0$ whenever $i+j-6 \geq 0$, that is $\partial_i\partial_4 a_0^{(4)} = 0$ for $i\in\{2,3,4\}$ and $\partial_3\partial_3 a_0^{(4)} = 0$. This implies
\begin{equation}
	a_0^{(4)}=u^4\,F_1(u^1)+u^3\,A(u^1;u^2)+B(u^1;u^2)
\end{equation}
for some function $F_1$ of $u^1$ and some polynomials $A$, $B$ of $u^2$ with coefficients being functions of $u^2$. By \eqref{Props12a}, $\partial_i\partial_ja_0^{(4)}=\partial_{\overline{i}}\partial_{\overline{j}}a_0^{(4)}$ when $i+j=\overline{i}+\overline{j}$. Thus, the nontrivial information is provided by the two following cases.

\vspace{0.3em}

\begin{itemize}
	\item[$(a)$] $i=1$, $j=4$, $\overline{i}=2$, $\overline{j}=3$; \[
F_1'(u^1)=\partial_1\partial_4a_0^{(4)}=\partial_2\partial_3a_0^{(4)}=\partial_2 A(u^1;u^2),
\]
implying $A(u^1;u^2)=u^2\,F_1'(u^1)+F_2(u^1)$ for some function $F_2$ of $u^1$. Hence,
\begin{equation*}	a_0^{(4)}=u^4\,F_1(u^1)+u^3\,F_2(u^1)+u^2\,u^3\,F_1'(u^1)+B(u^1;u^2).
\end{equation*}

\vspace{0.5em}

	\item[$(b)$] $i=1$, $j=3$, $\overline{i}=2$, $\overline{j}=2$; \[
F_2'(u^1)+u^2\,F_1''(u^1)=\partial_1\partial_3a_0^{(4)}=\partial_2\partial_2a_0^{(4)}=\partial_2^2 B(u^1;u^2),
\]
implying $B(u^1;u^2)=\frac{(u^2)^2}{2}\,F_2'(u^1)+\frac{(u^2)^3}{6}\,F_1''(u^1)+u^2\,F_3(u^1)+F_4(u^1)$ for some functions $F_3$, $F_4$ of $u^1$. Hence,
\begin{align*}	a_0^{(4)}= \, & \, u^4\,F_1(u^1) + u^3\,F_2(u^1) + u^2\,F_3(u^1) + F_4(u^1) + u^2\,u^3\,F_1'(u^1) + \tfrac{(u^2)^2}{2}\,F_2'(u^1) \\ \, & \, + \tfrac{(u^2)^3}{6}\,F_1''(u^1),
\end{align*}
in accordance with Theorem \ref{thm:gensol}.
\end{itemize}

\vspace{0.4em}

Let us now show that $a_0^{(4)}$ does indeed satisfy Corollary \ref{corollary:dimbridgeFs} as well. By analogous reconstruction procedures, one gets
\begin{gather*}
      a_0^{(1)}=F_1(u^1), \qquad   a_0^{(2)} = F_1(u^1)u^2 + F_2(u^1), \\
       a_0^{(3)} = u^3 F_1(u^1) + \dfrac{(u^2)^2}{2}F_1'(u^1) +u^2F_2(u^1) + F_3(u^1). 
\end{gather*}
Thus, considering $\partial_k a_0^{(4)}$ for $k=1, 2, 3, 4$, we get
\begin{equation*}
      \partial_4 a_0^{(4)} = F_1(u^1) = a_0^{(1)}(F_1), \quad   \partial_3 a_0^{(4)} = F_2(u^1) + u^2 F_1'(u^1) = a_0^{(2)}(F_1', F_2), 
\end{equation*}
\begin{equation*}
     \partial_2 a_0^{(4)} = F_3(u^1)+ u^3 F_1'(u^1) + u^2 F_2'(u^1) + \tfrac{(u^2)^2}{2} F_1''(u^1) = a_0^{(3)}(F_1', F_2', F_3), 
\end{equation*}
\begin{align*}
     \partial_1 a_0^{(4)} = \, & \,  u^4 F_1'(u^1) + u^3 F_2'(u^1) + u^2 F_3'(u^1) + F_4'(u^1) + u^2 u^3 F_1''(u^1)  \\ \, & \, + (u^2)^2 \dfrac{F_2''(u^1)}{2}  + (u^2)^3 \dfrac{F_1'''(u^1)}{6} \,   = a_0^{(4)}(F_1', F_2', F_3', F_4'),
\end{align*}
    respectively.
\end{example}

\vspace{0.3em}

\subsection{The general case}
\subsubsection{A splitting lemma}
This lemma has been proved in \cite{BKM3}. Here we provide an alternative proof fit for our purposes.

\begin{lemma}
    The general solution $a_0$ of \eqref{eq:ddainL} corresponding to an operator $L$ of regular type takes the form 
\begin{equation}
    a_0 = \sum_{\alpha = 1}^r a_{0,\alpha},
\end{equation}
where $a_{0, \alpha}$ is the general  solution of $d \cdot d_{L^{m_\alpha}}a_{0,\alpha} = 0$.
\label{lemma:splitting}
\end{lemma}

\begin{proof}
Let $\alpha \neq \beta$. We shall prove Lemma \ref{lemma:splitting} in three steps as follows. 

  \begin{enumerate}
       \item  Letting $i = m_\alpha(\alpha)$, and $j= m_{\beta}(\beta)$ in  \eqref{eq:ddasimpl}, gives
    \begin{equation*}
      (u^{1(\alpha)}-u^{1(\beta)}) \partial_{m_\alpha(\alpha)} \partial_{m_{\beta}(\beta)} a_0 = 0.
    \end{equation*}
    Thus, $a_0$ cannot have any terms containing both $u^{m_\alpha(\alpha)}$ and $u^{m_{\beta}(\beta)}$.
    \item Let $i=m_{\alpha}(\alpha)$ and assume that the statement holds for $j \geq p + 1$, for some fixed $p$ such that $p-\overset{\beta-1}{\underset{\sigma=1}{\sum}}\,m_\sigma\in\{1,\dots,m_\beta-1\}$. Let's show that the statement holds for $j = p$. By setting $i = m_\alpha(\alpha)$, $j=p$ in \eqref{eq:ddasimpl}, we get
    \begin{equation*}
    \begin{aligned}
     \, & \,  ( u^{1(\alpha)}- u^{1(\beta)}) \partial_{m_{\alpha}(\alpha)} \partial_p a_0 -  \sum_{t = p+1}^{m_{\beta(\beta)}} u^{(t-p+1)(\beta)} \partial_t \partial_{m_{\alpha}(\alpha)} a_0 = 0 \\
     \implies \, & \, ( u^{1(\alpha)}- u^{1(\beta)}) \partial_{m_{\alpha}(\alpha)} \partial_p a_0 = 0,
       \end{aligned}
    \end{equation*}
    by the induction hypothesis, implying $\partial_{m_{\alpha}(\alpha)} \partial_p a_0 = 0$. Hence, as ${u^{1(\alpha)} \neq u^{1(\beta)}}$, $a_0$ does not have any terms containing both $u^{m_\alpha(\alpha)}$ and $u^{j(\beta)}$, $\forall \, j(\beta)$. 
    
    \item Now, let's vary $i \in \{1(\alpha), \dots, m_{\alpha}(\alpha)\}$, with base case $(i,j) = (m_\alpha(\alpha), j(\beta))$, as proved in step 2.  Thus, we assume that $\partial_{i(\alpha)} \partial_{j(\beta)} a_0 = 0$, $\forall \, i > q$, for some fixed $q\in\{1,\dots,m_\alpha-1\}$. Setting $(i,j) = (q, j)$ in \eqref{eq:ddasimpl} gives
    \begin{equation*}
        \begin{aligned}
         \, & \,    \sum_{s = q}^{m_{\alpha}} u^{(s-q+1)(\alpha)} \partial_{s(\alpha)} \partial_{j(\beta)} a_0 - \sum_{t = j}^{m_{\beta}} u^{(t-j+1)(\beta)} \partial_{t(\beta)} \partial_{q(\alpha)} a_0 = 0\\
            \implies \, & \, (u^{1(\alpha)}-u^{1(\beta)})\partial_{q(\alpha)} \partial_{j(\beta)} a_0 + \underbrace{\sum_{s = q+1}^{m_{\alpha}}u^{(s-q+1)(\alpha)}\partial_{s(\alpha)} \partial_{j(\beta)} a_0}_{ = \, 0 \text{ by hypothesis.}} \notag\\&- \sum_{t = j+1}^{m_{\beta}} u^{(t-j+1)(\beta)} \partial_{t(\beta)} \partial_{q(\alpha)} a_0 = 0.
        \end{aligned}
    \end{equation*}
	By an inductive argument over $j$, analogous to step $(2)$, one can easily show that also $\partial_{q(\alpha)}\partial_{j(\beta)}a_0=0$ for each $j\in\{1,\dots,m_\beta\}$, concluding the proof.
     \end{enumerate}
\end{proof}
Lemma \ref{lemma:splitting} implies that in order to characterise solutions of \eqref{eq:ddainL} it is sufficient to consider operators $L$ with a single Jordan block. 
\subsubsection{Linearity and flatness}
We prove that the linearity of $a_0$ amounts to the flatness of the connection $\nabla$ which, by Theorem \ref{mainTh} and Proposition \ref{LinIn}, is uniquely determined by the following three requirements:
\begin{subequations}
	\begin{equation}
		\nabla e=0,
		\label{eq:eflat}
	\end{equation}
	\begin{equation}
		\nabla \text{ is torsionless, }
		\label{eq:nablanotorsion}
	\end{equation}
	\begin{equation}
		d_\nabla(L-a_0\,I)=0.
		\label{eq:dnablaLminus}
	\end{equation}
\end{subequations}
\begin{remark}
	As shown in \cite{ABLR}, and recalled in Remark \ref{RemarkFlatness}, once \eqref{rc-intri} is satisfied then the flatness of the connection reduces to the condition
	\begin{equation}
		R(e,\cdot)=0,
		\notag
	\end{equation}
	which in canonical coordinates reads
	\begin{equation}
		e\big(\Gamma^i_{jk}\big)=0,\qquad i,j,k\in\{1,\dots,n\}.
		\label{eGammas0}
	\end{equation}
	\label{rmk:newflatcond}
\end{remark}
\begin{proposition}
	If $a_0$ is linear then $\nabla$ is flat.
\end{proposition}
\begin{proof}
	Above, we proved existence and uniqueness of a solution $\big\{\Gamma^{i}_{jk}\big\}_{i,j,k\in\{1,\dots,n\}}$ of the system
	\begin{subequations}\label{sysGammas}
		\begin{equation}\label{sysGammasI}
			T^i_{jk}=\Gamma^i_{jk}-\Gamma^i_{kj}=0,\qquad i,j,k\in\{1,\dots,n\},
		\end{equation}
		\begin{equation}\label{sysGammasII}
			(d_{\nabla}V)^i_{jk}=\partial_jV^i_k-\partial_kV^i_j+\Gamma^i_{js}V^s_k-\Gamma^i_{ks}V^s_k=0,\qquad i,j,k\in\{1,\dots,n\},
		\end{equation}
		\begin{equation}\label{sysGammasIII}
			\nabla_je^i=\Gamma^i_{js}e^s=0,\qquad i,j\in\{1,\dots,n\},
		\end{equation}
	\end{subequations}
	where $V=X\circ$ for some vector field realising
	\begin{itemize}
		\item $X^{1(\alpha)}\neq X^{1(\beta)}$ for $\alpha\neq\beta$, $\alpha,\beta\in\{1,\dots,r\}$,
		\item $X^{2(\alpha)}\neq0$, $\alpha\in\{1,\dots,r\}$.
	\end{itemize}
	In particular, we consider $X=E-a_0\,e$. Since $a_0$ is assumed to be linear, so are the components of $X$, and of $V$ as well. By taking the Lie derivatives of the torsion $T$, of $d_{\nabla}V$ and of $\nabla e$ along the unit vector field $e$, in canonical coordinates we get the following homogeneous linear system for the quantities $e\big(\Gamma^i_{jk}\big)$:
	\begin{eqnarray}\label{eGamma1}
		(\mathcal{L}_e(T))^i_{jk}&=&e\big(\Gamma^i_{jk}\big)-e\big(\Gamma^i_{kj}\big)=0,\qquad i,j,k\in\{1,\dots,n\},\\\label{eGamma2}
		(\mathcal{L}_e(\nabla  e))^i_{j}&=&e\big(\Gamma^i_{js}\big)e^s=0,\qquad i,j\in\{1,\dots,n\},\\\label{eGamma3}
	(\mathcal{L}_e(d_{\nabla} V))^i_{jk}&=&e\big(\Gamma^i_{js}\big)V^s_k-e\big(\Gamma^i_{ks}\big)V^s_k=0,\qquad i,j,k\in\{1,\dots,n\},
		\label{eq:ednablaV}	
	\end{eqnarray}
	where, to get the last part of the system, we have used the identity
	\[e\big(\partial_jV^i_k\big)=e\big(\partial_kV^i_j\big)=0\]
	that follows from the linearity of the components of $V$ and  the identity
	\begin{equation}
		\Gamma^i_{js}e\big(V^s_k\big)-\Gamma^i_{ks}e\big(V^s_j\big)=0.\label{AuxGammaVanishSymm}
	\end{equation}
In particular, the identity \eqref{AuxGammaVanishSymm} follows from observing that
	\begin{align}
		&\Gamma^{i(\alpha)}_{j(\beta)s(\sigma)}e\big(V^{s(\sigma)}_{k(\gamma)}\big)-\Gamma^{i(\alpha)}_{k(\gamma)s(\sigma)}e\big(V^{s(\sigma)}_{j(\beta)}\big)=\Gamma^{i(\alpha)}_{j(\beta)s(\gamma)}e\big(V^{s(\gamma)}_{k(\gamma)}\big)-\Gamma^{i(\alpha)}_{k(\gamma)s(\beta)}e\big(V^{s(\beta)}_{j(\beta)}\big)\notag\\
		&=\Gamma^{i(\alpha)}_{j(\beta)k(\gamma)}e\big(u^{1(\gamma)}-a_0-u^{1(\beta)}+a_0\big)+\overset{m_\gamma}{\underset{s=k+1}{\sum}}\Gamma^{i(\alpha)}_{j(\beta)s(\gamma)}e\big(V^{s(\gamma)}_{k(\gamma)}\big)-\overset{m_\beta}{\underset{s=j+1}{\sum}}\Gamma^{i(\alpha)}_{k(\gamma)s(\beta)}e\big(V^{s(\beta)}_{j(\beta)}\big)\notag\\
		&=\Gamma^{i(\alpha)}_{j(\beta)k(\gamma)}e\big(u^{1(\gamma)}-u^{1(\beta)}\big)+\overset{m_\gamma}{\underset{s=k+1}{\sum}}\Gamma^{i(\alpha)}_{j(\beta)s(\gamma)}e\big(V^{s(\gamma)}_{k(\gamma)}\big)-\overset{m_\beta}{\underset{s=j+1}{\sum}}\Gamma^{i(\alpha)}_{k(\gamma)s(\beta)}e\big(V^{s(\beta)}_{j(\beta)}\big)=0,
		\notag
	\end{align}
	as $e=\overset{r}{\underset{\sigma=1}{\sum}}\partial_{1(\sigma)}$ and $V^{s(\gamma)}_{k(\gamma)}$ does not depend on $\{u^{1(\sigma)}\}_{\sigma\in\{1,\dots,r\}}$ when $s\geq k+1$.
	
	Clearly, the system $(\ref{eGamma1},\ref{eGamma2},\ref{eGamma3})$ admits the trivial solution
	 \[e\big(\Gamma^i_{jk}\big)=0,\qquad i,j,k\in\{1,\dots,n\}.\]
	  Uniqueness follows the observation that the linear system $(\ref{eGamma1},\ref{eGamma2},\ref{eGamma3})$ for the quantities  $\big\{e(\Gamma^i_{jk})\big\}_{i,j,k\in\{1,\dots,n\}}$ coincides with the homogeneous part of the linear system for  the quantities $\big\{\Gamma^i_{jk}\big\}_{i,j,k\in\{1,\dots,n\}}$ (just replacing $\Gamma^i_{jk}$ with $e\big(\Gamma^i_{jk}\big)$) whose solution is unique by Theorem \ref{mainTh} and Proposition \ref{LinIn}.
\end{proof}
\begin{proposition}
	If $\nabla$ is flat then $a_0$ is linear.
\end{proposition}

\begin{proof}
  As in the converse direction, we consider the equation \eqref{eGamma3} obtained by taking the Lie derivative in the direction of the unit vector field, $e = \sum_{\alpha = 1}^r \partial_{1(\alpha)}$, where $r$ denotes the number of Jordan blocks, of the equation $d_{\nabla}V = 0$, \eqref{sysGammasII}.
    Now, however, we assume that $e(\Gamma^k_{ij}) = 0$ for any $\{i,j,k\}$. Furthermore, as $u^{1(\alpha)}$, for any $\alpha$, only appears on the diagonal of $V$, we get that $\Gamma^k_{il}e(V^j_l) \delta^l_{j} - \Gamma^k_{jl} e(V^i_l)\delta^l_{i} = 0 $, which leaves us with
    \begin{equation}
        e(\partial_i V_j^k) - e(\partial_j V^k_i) = 0.
    \end{equation}
    Note that we can only obtain non-trivial equations by setting $k = j$ or $k = i$, while  keeping $i \neq j$. Without loss of generality we consider the case of $k = j \neq i$ giving
    \begin{equation}
        \sum_{\alpha} \partial_{1(\alpha)} \partial_i V^j_j  = \sum_{\alpha} \partial_i \partial_{1(\alpha)}a_0 = 0, 
    \end{equation}
    where  we have used that $V^j_j = (L - a_0 I)^j_j = u^{1(\beta)}-a_0$, for $j$ referring to the $\beta$-th block. By the splitting lemma, Lemma \ref{lemma:splitting}, we get that
    \begin{equation}
       \partial_{1(\beta)}  \partial_{s(\beta)} a_{0, \beta} = 0,
    \end{equation}
    where $\beta$ is the block associated to the label $i \equiv s(\beta)$. This, however, is analogous to consider the one-block case of dimension $m_{\beta}$. Therefore, we may make use of \eqref{Props12a}, which implies that  
    \begin{equation}
        \partial_{p(\beta)} \partial_{q(\beta)} a_0^{(m_{\beta})} = 0 \Leftrightarrow \partial_{1(\beta)} \partial_{(p + q - 1)(\beta)} a_0^{(m_\beta)} = 0.
    \end{equation}
    Thus, we have that that $a_{0, \alpha}$ must be linear, by setting $s(\beta) = (p + q - 1)(\beta)$. Note that while this is only valid for $0 < s = p + q - 1 < m_{\beta} + 2$, the remaining cases follow by \eqref{Props12b}.     Finally, this must hold for any block $\beta$, as $i$ is arbitrary, which implies that $a_0$ is linear by the splitting lemma, Lemma \ref{lemma:splitting}.
\end{proof}

Thus, we have proved the following.
\begin{theorem}\label{lin=flat}
	Let $a_0$ be a solution to \eqref{eq:main}. Then the F-manifold with compatible connection and flat unit determined by $d_\nabla(L-a_0\,I)=0$, is flat if and only if $a_0$ is linear.
\end{theorem}

\subsection{Examples}
Below we list, for all the regular cases up to\footnote{We omit $n = 1$ as this case is necessarily semisimple.} (and including) dimension $4$, the explicit expressions of $a_0$ and the non-vanishing Christoffel symbols, obtained via the three constraints \eqref{eq:eflat} - \eqref{eq:dnablaLminus}. Moreover, case by case, we show the following equivalences.

\vspace{1em}

 \emph{Equivalence 1.} This equivalence is between condition \eqref{eq:3RC} and condition \eqref{eq:main}, starting from the first nontrivial case of dimension $3$  (the two-dimensional case is degenerate since condition \eqref{eq:3RC} follows from the identity $R^i_{jkl}=-R^i_{jlk}$).
    
In agreement with the general theory the solutions to \eqref{eq:ddainL} in dimensions $n \leq 4$ do indeed give \eqref{eq:3RC}. Thus, in the below, we show the converse statement. That is, if we impose \eqref{eq:3RC} then $a_0$ must have the form obtained from  \eqref{eq:ddainL}.

Any element of the $c$--tensor is either 0 or 1. See Table \ref{tab:exC} for the non-zero elements for dimensions $ 1 < n \leq 4$.

\begin{table}[h!]
\begin{center}

{
\begin{tabular}{ |c|c|c|  }

\hline
$n$   &  Jordan type  & Non-zero elements of $c$ \\
\hline
\hline
$2$ &  $2 \times 2$  & $c^1_{11}, c^2_{12}, c^2_{21}$ \\
\hline
$3$ &  $3 \times 3$  &  $ c^1_{11},  c^2_{12},  c^2_{21},  c^3_{13}, c^3_{22}, c^3_{31}$ \\
\hline
 &  $2 \times 2 + 1 \times 1$  &  $   c^1_{11},  c^2_{12},  c^2_{21},  c^3_{33}$ \\
\hline
$4$ &  $4 \times 4$  &  $ c^1_{11}, c^2_{12}, c^2_{21}, c^3_{13}, c^3_{22},  c^3_{31}, c^4_{14}, c^4_{23},  c^4_{32},  c^4_{41}$ \\
\hline
 &  $3 \times 3 + 1 \times 1$  &  $c^1_{11}, c^2_{12}, c^2_{21}, c^3_{13}, c^3_{22}, c^3_{31}, c^4_{44}$ \\
\hline
 &  $2 \times 2 + 2 \times 2$  &  $c^1_{11}, c^2_{12}, c^2_{21}, c^3_{33}, c^4_{34}, c^4_{43}$ \\
\hline
 &  $2 \times 2 + 1 \times 1 + 1 \times 1$  &  $ c^1_{11}, c^2_{12}, c^2_{21}, c^3_{33}, c^4_{44}$ \\
\hline
\end{tabular}
}
\end{center}
\caption{Nonzero elements of the $c$-tensor in dimensions  2, 3, and  4.}
\label{tab:exC}
\end{table}

 \emph{Equivalence 2.} This equivalence states a correspondence between linearity of solutions $a_0$ to \eqref{eq:main} and  bi-flatness of the associated F-manifold structure. By Theorem \ref{lin=flat} we have proved the correspondence between flatness of $\nabla$ and linearity of $a_0$, but we have not touched upon the dual structure. Thus, in the below, we give the Christoffel symbols associated to the second flat structure, which we will denote by $\widetilde{\Gamma}^i_{jk}$, establishing the existence and uniqueness of such a structure. Additionally, using these, the flatness of the dual connection can be easily verified.  

\vspace{1em}
 
 \emph{Equivalence 3.} This equivalence is between linear solutions, $a_0$, of the special form appearing in \cite{LP23},
 \begin{equation}
	a_0=\overset{r}{\underset{\alpha=1}{\sum}}\,\epsilon_{1(\alpha)}\,u^{1(\alpha)},
	\label{a0_linear_maincoords}
\end{equation} and the associated F-manifold to be Riemannian with Killing vector field (as studied in \cite{ABLR} and recalled in Subsection \ref{subsectionRiemannian}).  Note that, \eqref{a0_linear_maincoords} is a linear function of the first coordinate of each block $\{u^{1(\alpha)}\}_{\alpha\in\{1,\dots,r\}}$, with coefficients $\{\epsilon_{1(\alpha)}\}_{\alpha\in\{1,\dots,r\}}$. This is due to the fact that when a generic linear function $a_0$ is taken as a starting point, the only solution to (\eqref{gcompatc}, \eqref{Riemannian_flatFmnf_bridge}) and satisfying $\mathcal{L}_eg=0$ (obtained with Maple) is a degenerate metric, while non-degenerate solutions are permitted when starting from functions of the form \eqref{a0_linear_maincoords}. In the forthcoming examples, we list the matrices representing such metrics in canonical coordinates.

In the remainder of this section we will, for the purpose of readability, write coordinates $u^i$ with lower indices, i.e. as $u_i$.

\subsubsection{Dimension $2$: $2\times 2$ Jordan block}
The operator of multiplication by the Euler vector field is represented by the matrix
\beq
L=\begin{bmatrix}
u_{1} & 0\cr
u_{2} & u_{1} \cr
\end{bmatrix}.
\eeq

The general solution  of the  equation \eqref{eq:ddainL} is
\begin{equation}
    a_0=F_1(u_1) \, u_2+F_2(u_1).\label{ddasol2}
\end{equation}

The non-vanishing Christoffel symbols are given by
\begin{equation}
    \Gamma^i_{jk} =\frac{(-1)^{i+1} \delta_{j,2} \, \delta_{k,2} }{u_2} \,  \partial_{3-i}a_0.
    \label{Chr:2}
\end{equation}

\underline{The second equivalence}\\
Let $a_0 = \epsilon_1\,u_1 + \epsilon_2\,u_2$. Then, the Christoffel symbols for the dual structure are given by
\begin{gather*}
    \widetilde{\Gamma}^1_{11} = \dfrac{\epsilon_2 u_2 - u_1}{u_1^2}, \quad  \widetilde{\Gamma}^1_{12}  = -\dfrac{\epsilon_2}{u_1}, \quad  \widetilde{\Gamma}^1_{22} = \dfrac{\epsilon_2}{u_2}, \quad \widetilde{\Gamma}^2_{11} = \dfrac{u_2(1-\epsilon_1)}{u_1^2}, \\ \widetilde{\Gamma}^2_{12} = \dfrac{\epsilon_1-1}{u_1}, \quad  \widetilde{\Gamma}^2_{22} = - \dfrac{\epsilon_1}{u_2}.   
\end{gather*}

\underline{The third equivalence}\\
The metric corresponding to the linear function $a_0=\epsilon_1\,u_1$ is represented by the matrix
\begin{equation}
	g=\begin{bmatrix}
		F_1(u_2)&C_1\,u_2^{-\epsilon_1}\\C_1\,u_2^{-\epsilon_1}&0
	\end{bmatrix}
\end{equation}
where $C_1$ is an arbitrary constant and $F_1$ is an arbitrary function of a single variable (this example appears  in \cite{ABLR}).

\subsubsection{Dimension $3$: $3\times 3$ Jordan block}
The operator of multiplication by the Euler vector field is represented by the matrix
\beq
L=\begin{bmatrix}
u_{1} & 0 & 0\cr
u_{2} & u_{1} &  0\cr
u_{3} & u_{2} & u_{1}
\end{bmatrix}.
\eeq
The general solution  of the  equation \eqref{eq:ddainL} is
\begin{equation}
    a_0=F_1( u_1)u_3+\frac{1}{2}F'_1(u_1)u_2^{2}+F_2(u_1) u_{{2}}+F_3(u_1).
\label{ddasol3}
\end{equation}

The non-vanishing Christoffel symbols are
\begin{equation}
\begin{gathered}
    \Gamma^{1}_{22}   = \frac{1}{u_2}\left(\partial_2 a_0 - \frac{u_3}{u_2} \, \partial_3 a_0 \right), \qquad \Gamma^{1}_{23} = - \frac{u_2}{u_3} \,  \Gamma^{2}_{23} = \Gamma^{2}_{33}  = \frac{\partial_3 a_0}{u_2},\\
   \Gamma^{2}_{22}   = \frac{1}{u_2}\left(\frac{u_3^2}{u_2^2} \,  \partial_3 a_0 - \partial_1 a_0 \right), \qquad  \Gamma^{3}_{22} = -\frac{u_3}{u_2} \, \Gamma^3_{23} = \frac{u_3}{u_2^2}\left(\partial_1 a_0 - \frac{u_3}{u_2} \, \partial_2 a_0 \right), \\
   \Gamma^{3}_{33}  = -\frac{\partial_2 a_0}{u_2}.
\end{gathered}
\label{Chr:3}
\end{equation}

\underline{The first equivalence.}\\
It is sufficient to consider \eqref{eq:3RC} for \begin{equation}
    (i,j,k,l,m) = \overset{a)}{\overbrace{(2,1,1,2,3)}},  \, \overset{b)}{\overbrace{(2,2,1,3,2)}}, \, \overset{c)}{\overbrace{(2,3,1,2,3)}}.
\end{equation} 

Inputting \eqref{Chr:3} into \eqref{eq:3RC}, together with $c$ as in Table \ref{tab:exC}, we get the following. 

\begin{gather*}
    a):  \dfrac{u_3 \partial_3^2 a_0}{u_2^2} = 0, \quad b): \dfrac{u_3(u_3\partial_3^2 a_0 + u_2 \partial_2 \partial_3 a_0)}{u_2^3} = 0,\\
    c):  \dfrac{u_3(u_3 \partial_2 \partial_3 a_0 + u_2(\partial_2^2 a_0 - \partial_1 \partial_3 a_0))}{u_2^3} = 0.
\end{gather*}

\begin{itemize}
    \item a):  $a_0$ must be linear in $u_3$. 
\item b): After using the fact that $a_0$ must be linear in $u_3$ as obtained in a), we must have no terms depending on both $u_2$ and $u_3$.
\item c): Using the result of b), we get that $ \partial^2_2 a_0 = \partial_1 \partial_3 a_0$, which is the only remaining coefficient to be specified. 
\end{itemize}

This exactly specifies the solution \eqref{ddasol3}.

\vspace{1em}

\underline{The second equivalence}\\
Let $a_0 = \epsilon_1\,u_1 + \epsilon_2\,u_2 + \epsilon_3\,u_3$. Then, the non-zero Christoffel symbols for the dual structure are then given by
\begin{gather*}
    \widetilde{\Gamma}^1_{11} = \dfrac{\epsilon_2 u_1 u_2 + \epsilon_3(u_1 u_3 - u_2^2) - u_1^2}{u_1^3}, \qquad  \widetilde{\Gamma}^1_{12}  = -\dfrac{\epsilon_2 u_1 - \epsilon_3 u_2}{u_1^2}, \qquad  \widetilde{\Gamma}^1_{13} = -\dfrac{\epsilon_3}{u_1}, \\ \widetilde{\Gamma}^1_{22} = \dfrac{\epsilon_2 u_1 u_2 - \epsilon_3(u_1 u_3 +  u_2^2) }{u_1 u_2^2}, \qquad \widetilde{\Gamma}^1_{23} = \dfrac{\epsilon_3}{u_2}, \qquad  \widetilde{\Gamma}^2_{11} = \widetilde{\Gamma}^3_{12} =  \dfrac{(1- \epsilon_1)u_2}{u_1^2}, \\ \widetilde{\Gamma}^2_{12} = \widetilde{\Gamma}^3_{13}  =  \dfrac{\epsilon_1 - 1}{u_1}, \qquad \widetilde{\Gamma}^2_{22} =  -\dfrac{\epsilon_1 u_2^2 - \epsilon_3 u_3^2}{u_2^3}, \qquad \widetilde{\Gamma}^2_{23} = - \dfrac{\epsilon_3 u_3}{u_2^2}, \qquad \widetilde{\Gamma}^2_{33} =  \dfrac{\epsilon_3}{u_2}, \\ \widetilde{\Gamma}^3_{11} =  \dfrac{(1- \epsilon_1)(u_1 u_3 - u_2^2)}{u_1^3}, \qquad \widetilde{\Gamma}^3_{22} = \dfrac{\epsilon_1u_2(u_1 u_3 + u_2^2) - \epsilon_2 u_1 u_3^2 - u_2^3}{u_1 u_2^3}, \\ \widetilde{\Gamma}^3_{23} = -\dfrac{\epsilon_1 u_2 - \epsilon_2 u_3}{u_2^2}, \qquad \widetilde{\Gamma}^3_{33} = -\dfrac{\epsilon_2}{u_2}.
\end{gather*}

\underline{The third equivalence}\\
The metric corresponding to the linear function $a_0=\epsilon_1\,u^1$ is represented by the matrix
\begin{equation}
	g=\begin{bmatrix}
		g_{11}&g_{12}&g_{13}\\g_{12}&g_{13}&0\\g_{13}&0&0
	\end{bmatrix},
\end{equation}
where
\begin{align}
	g_{11}= \, & \, \frac{2}{9}C_1\epsilon_1\bigg(\epsilon_1-\frac{3}{2}\bigg)u_2^{-2-\frac{4\epsilon_1}{3}}u_3^2+2F_1'(u_2)u_3+2\epsilon_1F_1(u_2)\frac{u_3}{u_2}+F_2(u_2),
	\notag\\
	g_{12}= \, & \, - \frac{2}{3}C_1\epsilon_1 u_2 ^{-1-\frac{4\epsilon_1}{3}}u_3+F_1(u_2),
	\notag\\
	g_{13}= \, & \, C_1 u_2^{-\frac{4\epsilon_1}{3}},
	\notag
\end{align}
for some constant $C_1$ and some functions $F_1$, $F_2$ of a single variable.

\subsubsection{Dimension $3$: $2\times 2+1\times 1$ Jordan blocks}
The operator of multiplication by the Euler vector field is represented by the matrix
\beq\label{Lblock_2+1}
L=\begin{bmatrix}
u_{1} & 0 & 0\cr
u_{2} & u_{1} &  0\cr
0 & 0 & u_{3}
\end{bmatrix}.
\eeq
The general solution  of the  equation \eqref{eq:ddainL} is
\begin{equation}
    a_0=F_1(u_1)u_2+F_2(u_1)+F_3(u_3).\label{ddasol21}
\end{equation}

The non-vanishing Christoffel symbols are
\begin{equation}
    \begin{gathered}
       \Gamma^{1}_{11} = -\Gamma^{1}_{13} = -\Gamma^{1}_{33} = \frac{u_1-u_3}{u_2} \, \Gamma^2_{11} =  - \Gamma^{2}_{12} = -\frac{u_1-u_3}{u_2} \, \Gamma^2_{13} =  \Gamma^{2}_{23}   =\frac{u_1-u_3}{u_2} \, \Gamma^2_{33}  \\
    = -\frac{\partial_3 a_0}{(u_1-u_3)},   \qquad \Gamma^{1}_{22}   = \frac{u_1-u_3}{u_2} \Gamma^3_{12} = -\frac{u_1-u_3}{u_2} \Gamma^3_{23} = \frac{\partial_2 a_0}{u_2}, \qquad \qquad \, \,  \\   \Gamma^{2}_{22} = - \frac{\partial_1 a_0}{u_2}, \qquad
      \Gamma^3_{11}  = -\Gamma^3_{13} = \Gamma^3_{33} = \frac{1}{(u_1-u_3)}\left(\partial_1 a_0 - \frac{u_2}{(u_1-u_3)}\partial_2 a_0 \right).
    \end{gathered}
    \label{Chr:21}
\end{equation}

\underline{The first equivalence}\\

In this case, it is sufficient to consider \eqref{eq:3RC}
 for \begin{equation*}
    (i,j,k,l,m) = \overset{a)}{\overbrace{(2,1,1,3,2)}}, \,  \overset{b)}{\overbrace{(2,2,1,2,3)}}, \, \overset{c)}{\overbrace{(2,3,1,3,2)}}.
\end{equation*}  

Inputting \eqref{Chr:21} into \eqref{eq:3RC} together with $c$ as in Table \ref{tab:exC},  we get the following.
\begin{gather*}
    a):  \dfrac{\partial_2 \partial_3 a_0}{u_2} = 0, \quad b):  \dfrac{\partial_2 \partial_3 a_0}{u_1-u_3} + \dfrac{\partial_1 \partial_3 a_0}{u_2} = 0, \quad c):  \dfrac{\partial_2^2 a_0}{u_1-u_3} = 0.
\end{gather*}
\begin{itemize}
    \item a):  This equation implies that $a_0$ contains no terms depending on both $u_2$ and $u_3$. 
\item b): After using the result of a), this equation implies that we must have no terms depending on both $u_1$ and $u_3$.
\item c): This equation implies that $a_0$ must be linear in $u_2$.
\end{itemize}
This precisely specifies the solution \eqref{ddasol21}.

\vspace{1em}

\underline{The second equivalence}\\
Let $a_0 = \epsilon_1\,u_1 + \epsilon_2\,u_2 + \epsilon_3\,u_3$. Then, the non-zero Christoffel symbols for the dual structure are given by
\begin{gather*}
    \widetilde{\Gamma}^1_{11} = \dfrac{\epsilon_2 u_2(u_1 - u_3) - \epsilon_3 u_1 u_3- u_1^2 + u_1 u_3}{u_1^2 (u_1 - u_3)}, \qquad  \widetilde{\Gamma}^1_{12}  = - \dfrac{\epsilon_2}{u_1}, \qquad  \widetilde{\Gamma}^1_{13} = \dfrac{\epsilon_3}{u_1 - u_3}, \\ \widetilde{\Gamma}^1_{22} = \dfrac{\epsilon_2}{u_2}, \qquad  \widetilde{\Gamma}^2_{11} =  -\dfrac{((\epsilon_1 - 1)(u_1-u_3)^2 - \epsilon_3 u_3(2 u_1 - u_3))u_2}{u_1^2}, \\ \widetilde{\Gamma}^2_{12} =  \dfrac{\epsilon_1(u_1 - u_3) - \epsilon_3 u_3 - u_1 + u_3}{u_1(u_1 - u_3)}, \qquad \widetilde{\Gamma}^2_{22} =  -\dfrac{\epsilon_1}{u_2}, \qquad \widetilde{\Gamma}^2_{23} = - \dfrac{\epsilon_3}{u_1 - u_3}, \\ \widetilde{\Gamma}^2_{33} =  \dfrac{\epsilon_3 u_2}{(u_1 - u_3)^2}, \qquad \widetilde{\Gamma}^3_{11} =  \dfrac{(\epsilon_1u_1(u_1 - u_3) -  \epsilon_2 u_2 (2 u_1- u_3) )u_3}{u_1^2 (u_1 - u_3)^2}, \\ \widetilde{\Gamma}^3_{12} = -\dfrac{\epsilon_2 u_3}{u_1(u_1 - u_3)}, \qquad \widetilde{\Gamma}^3_{13} = -\dfrac{\epsilon_1(u_1 - u_3) - \epsilon_2 u_2}{(u_1 - u_3)^2}, \qquad \widetilde{\Gamma}^3_{23} = -\dfrac{\epsilon_2}{u_1 - u_3}, \\ \widetilde{\Gamma}^3_{33} = \dfrac{\epsilon_1 u_1(u_1 -  u_3) - \epsilon_2 u_2 u_3 - (u_1-u_3)^2}{u_3(u_1 - u_3)^2}.
\end{gather*}

\underline{The third equivalence}\\
The metric corresponding to the linear function $a_0=\epsilon_1\,u_1+\epsilon_3\,u_3$ is represented by the matrix
\begin{equation}
	g=\begin{bmatrix}
		g_{11}&g_{12}&0\\g_{12}&0&0\\0&0&g_{33}
	\end{bmatrix},
\end{equation}
where
\begin{align}
	g_{11}= \, & \, F_1(u_2)(u_3-u_1)^{-2\epsilon_3}+C_1 u_2^{1-\epsilon_1}(u_3-u_1)^{-2\epsilon_3-1},
	\notag\\
	g_{12}= \, & \, \frac{C_1}{2\epsilon_3} u_2^{-\epsilon_1}(u_3-u_1)^{-2\epsilon_3},
	\notag\\
	g_{33}= \, & \, C_2(u_1-u_3)^{-2\epsilon_1},
	\notag
\end{align}
for some constants $C_1$, $C_2$ and some function $F_1$ of a single variable.

\subsubsection{Dimension $4$: $4\times 4$ Jordan block}
The operator of multiplication by the Euler vector field is represented by the matrix
\beq
L=\begin{bmatrix}
u_{1} & 0 & 0 & 0\cr
u_{2} & u_{1} & 0 & 0\cr
u_3 & u_2 & u_1 & 0\cr
u_{4} & u_3 & u_{2} & u_{1}
\end{bmatrix}.
\eeq
The general solution  of the  equation \eqref{eq:ddainL} is
\begin{equation}
    a_0=\left( F'_1( u_1) u_2+F_2( u_1)\right)u_3+ F_1( u_1) u_4+\frac{1}{6}F''_1(u_1))u_2^{3}+\f{1}{2}F_2'(u_1)u_2^{2}+F_3(u_1)u_2+F_4(u_1).
    \label{ddasol4}
\end{equation} 
The non-vanishing Christoffel symbols are
        \begin{equation}
        \begin{gathered}
             \Gamma^1_{22}  =   \frac{1}{u_2} \left(  \partial_2 a_0 - \frac{u_3}{u_2} \, \partial_3 a_0 + \frac{1}{u_2}\left(\frac{u_3^2}{u_2} - u_4\right)\partial_4 a_0 \right), \qquad   \Gamma^1_{23} = \frac{1}{u_2}\left(\partial_3 a_0 - \frac{u_3}{u_2} \,  \partial_4 a_0 \right),  \\
                \Gamma^1_{24}   = \Gamma^1_{33} = -\frac{u_2}{u_3} \, \Gamma^2_{24} = \Gamma^2_{34}    = \frac{u_2^2}{(u_3^2-u_2u_4)} \, \Gamma^3_{24} =  -\frac{u_2}{u_3} \, \Gamma^3_{34} = \Gamma^3_{44}   = \frac{\partial_4 a_0}{u_2}, \\
               \Gamma^2_{22}  = \frac{1}{u_2}\left(- \partial_1 a_0 + \frac{u_3^2}{u_2^2} \, \partial_3 a_0 + 2\frac{u_3}{u_2^2}\left(u_4 - \frac{u_3^2}{u_2} \right)\partial_4 a_0  \right),  \\
                \Gamma^2_{23}  = \frac{1}{u_2^2}\left(- u_3 \, \partial_3 a_0 + \left(\frac{2u_3^2}{u_2} - u_4\right)\partial_4 a_0\right), \qquad 
                \Gamma^2_{33} = \frac{1}{u_2}\left(\partial_3 a_0 - 2\frac{u_3}{u_2} \, \partial_4 a_0\right), \\
                \Gamma^3_{22}  = \frac{1}{u_2^2}\left(u_3 \,  \partial_1 a_0 - \frac{u_3^2}{u_2} \, \partial_2 a_0 + \frac{1}{u_2}\left( 2\frac{u_3^4}{u_2^2} - 3 \frac{u_3^2u_4}{u_2} + u_4^2\right)\partial_4 a_0 \right), \\
                \Gamma^3_{23}    = \frac{1}{u_2}\left(-\partial_1 a_0 + \frac{u_3}{u_2} \, \partial_2 a_0 + \frac{u_3}{u_2^2} \left(u_4 - \frac{u_3^2}{u_2} \right) \partial_4 a_0 \right), \\
                \Gamma^3_{33}   = \frac{1}{u_2}\left(-\partial_2 a_0 + \frac{1}{u_2}\left(\frac{2u_3^2}{u_2}- u_4\right)\partial_4 a_0 \right), \\ 
                \Gamma^4_{22}   = \frac{u_2 u_4-u_3^2}{u_2^3}\left(\partial_1 a_0 - \frac{2u_3}{u_2} \,  \partial_2 a_0 + \frac{1}{u_2}\left(\frac{2 u_3^2}{u_2}-u_4\right)\partial_3 a_0 \right), \\
                 \Gamma^4_{23}   = \frac{1}{u_2^2}\left(u_3 \,  \partial_1 a_0 + \left(u_4-\frac{2u_3^2}{u_2} \right) \partial_2 a_0 + \frac{2}{u_2} \left(\frac{u_3^3}{u_2}-u_3 u_4 \right) \partial_3 a_0\right),\\
                   \Gamma^4_{24}    = \frac{1}{u_2}\left(-\partial_1 a_0+ \frac{u_3}{u_2} \, \partial_2 a_0 + \frac{1}{u_2}\left(u_4-\frac{u_3^2}{u_2}\right)\partial_3 a_0\right),   \\          
               \Gamma^4_{33}   = \frac{1}{u_2} \left(-\partial_1 a_0 + 2\frac{u_3}{u_2} \, \partial_2 a_0 + \frac{1}{u_2}\left(u_4-\frac{2u_3^2}{u_2} \right)\partial_3 a_0 \right), \\
               \Gamma^4_{34}  = \frac{1}{u_2}\left(-\partial_2 a_0 + \frac{u_3}{u_2} \, \partial_3 a_0\right), \qquad \Gamma^4_{44} = -\frac{\partial_3 a_0}{u_2}.
                    \end{gathered}
                     \label{Chr:4}
                        \end{equation}
\underline{The first equivalence}\\
In this case, it is sufficient to consider \eqref{eq:3RC}
 for \begin{equation*}
 \begin{aligned}
     (i,j,k,l,m) = \overset{a)}{\overbrace{(2,1,1,3,4)}}, \, \overset{b)}{\overbrace{(4,1,1,3,2)}}, \,  \overset{c)}{\overbrace{(4,3,1,3,2)}}, \\\underset{d)}{\underbrace{(2,1,1,2,3)}}, \,  \underset{e)}{\underbrace{(4,4,3,2,1)}}, \,  \underset{f)}{\underbrace{(4,4,2,4,1)}}.
     \end{aligned}
 \end{equation*} 

Inputting \eqref{Chr:4} into \eqref{eq:3RC} together with $c$ as in Table \ref{tab:exC}, we get the following.

\begin{gather*}
    a): \dfrac{u_3 \partial_4^2 a_0}{u_2^2} = 0, \quad b): \dfrac{\partial_3 \partial_4 a_0}{u_2} = 0, \quad 
    c): \dfrac{(u_3^2-u_2u_4)\partial_3 \partial_4 a_0 + u_2 u_3 \partial_2 \partial_4 a_0}{u_2^3} = 0, \\ d): \dfrac{(u_2 u_4 - u_3^2)\partial_3\partial_4 a_0 + u_2 u_3(\partial_3^2 a_0-\partial_2 \partial_4 a_0)}{u_2^3}, \end{gather*}
    \begin{gather*} e): \dfrac{u_2^2(\partial_2^2 a_0- \partial_1 \partial_3 a_0) + (u_2 u_4 - u_3^2)\partial_3^2 a_0}{u_2^3}, \\ f): \dfrac{(u_3^2-u_2 u_4)\partial_3 \partial_4 a_0- u_2(u_3 \partial_2 \partial_4 a_0 + u_2(\partial_2 \partial_3 a_0 - \partial_1 \partial_4 a_0))}{u_2^3}.
    \end{gather*}

\begin{itemize}
    \item a): $a_0$ must be linear in $u_4$.
\item b): $a_0$ does not have any terms depending on both $u_3$ and $u_4$.
\item c): After using b), $a_0$ does not have any terms depending on both $u_2$ and $u_4$.
\item d): After using b) and c), $a_0$ must be linear in $u_3$.
\item e): After using d), $\partial_2^2 a_0 = \partial_1 \partial_3 a_0$. This fixes the maximum degree in $u_2$ to 3 and specifies $[u_2^2]$, and $[u_2^3]$ from $[u_2 u_3]$ and $[u_3]$, respectively. Here $[u_i^j]$ denotes the coefficient of $u_i^j$ in $a_0$. 
\item f): After using b), and c), $\partial_2 \partial_3 a_0 = \partial_1 \partial_4 a_0$, which specifies the remaining term $[u_2 u_3]$ from the linear term in $u_4$.
\end{itemize}
This precisely determines $a_0$ to coincide with the solution \eqref{ddasol4}.

\vspace{1em}

\underline{The second equivalence}\\
 Let $a_0 = \epsilon_1\,u_1 + \epsilon_2\,u_2 + \epsilon_3\,u_3 + \epsilon_4\,u_4$. Then, the non-zero Christoffel symbols for the dual connection are given by 

\begin{gather*}
    \widetilde{\Gamma}^1_{11} = \dfrac{\epsilon_2 u_1^2 u_2 + \epsilon_3 u_1(u_1 u_3 - u_2^2) + \epsilon_4 u_2(u_1 u_2 + u_2^2 - 2 u_3) - u_1^3}{u_1^4}, \\
     \widetilde{\Gamma}^1_{12}  = - \dfrac{\epsilon_2 u_1^2 - \epsilon_3 u_1 u_2 - \epsilon_4(u_1 u_3 - u_2^2)}{u_1^3}, \qquad  \widetilde{\Gamma}^1_{13} = - \dfrac{\epsilon_3 u_1 - \epsilon_4 u_2}{u_1^2}, \qquad \widetilde{\Gamma}^1_{14} = - \dfrac{\epsilon_4}{u_1}, \\  \widetilde{\Gamma}^1_{22} =  \dfrac{\epsilon_2 u_1^2 u_2^2 - \epsilon_3 u_1 u_2(u_1 u_3 + u_2^2)- \epsilon_4(u_1^2(u_2 u_4 - u_3^2) - u_2^4)}{u_1^2 u_2^3}, \\
     \widetilde{\Gamma}^1_{23} =  \dfrac{\epsilon_3 u_1 u_2 - \epsilon_4 (u_1 u_3 + u_2^2)}{u_1 u_2^2}, \qquad \widetilde{\Gamma}^1_{24} = \widetilde{\Gamma}^1_{33} = \widetilde{\Gamma}^2_{34} = \widetilde{\Gamma}^3_{44}  = \dfrac{\epsilon_4}{u_2},  \\ \widetilde{\Gamma}^2_{11} =  \widetilde{\Gamma}^3_{12} = \widetilde{\Gamma}^4_{13} =  \dfrac{(1 - \epsilon_1)u_2}{u_1^2}, \qquad \widetilde{\Gamma}^2_{12}  = \widetilde{\Gamma}^3_{13} = \widetilde{\Gamma}^4_{14} =  \dfrac{\epsilon_1 - 1}{u_1}, 
   \\\widetilde{\Gamma}^2_{22} = \widetilde{\Gamma}^3_{23} = - \dfrac{\epsilon_1 u_2^3 - \epsilon_3 u_2 u_3^2- 2 \epsilon_4 u_3(u_2 u_4 - u_3^2)}{u_2^4}, \qquad \widetilde{\Gamma}^3_{12} = -\dfrac{\epsilon_2 u_3}{u_1(u_1 - u_3)}, \\
    \widetilde{\Gamma}^2_{23} = -\dfrac{\epsilon_3 u_2 u_3 + \epsilon_4(u_2 u_4 - 2 u_3^2)}{u_2^3}, \qquad \widetilde{\Gamma}^2_{24} = \widetilde{\Gamma}^3_{34} = -\dfrac{\epsilon_4 u_3}{u_2^2}, \qquad  \widetilde{\Gamma}^2_{33} = \dfrac{\epsilon_3 u_2 - 2 \epsilon_4 u_3}{u_2^2}, \\
    \widetilde{\Gamma}^3_{11} = \widetilde{\Gamma}^4_{12} = \dfrac{(1 - \epsilon_1)(u_1 u_3 - u_2^2)}{u_1^3},\\
      \widetilde{\Gamma}^3_{22} = \dfrac{\epsilon_1 u_2^3(u_1 u_3 + u_2^2) - \epsilon_2 u_1 u_2^2 u_3^2 + \epsilon_4 u_1(u_2 u_4 - u_3^2)(u_2 u_4- 2 u_3^2) - u_2^5 }{u_1 u_2^5},\\
   \widetilde{\Gamma}^3_{24} = -\dfrac{\epsilon_4(u_2 u_4 - u_3^2)}{u_2^3}, \qquad \widetilde{\Gamma}^3_{33} = - \dfrac{\epsilon_2 u_2^2 + \epsilon_4(u_2 u_4 - 2 u_3^2)}{u_2^3}, \\ \widetilde{\Gamma}^4_{11} = \dfrac{(1 - \epsilon_1)(u_1^2 u_4 - 2 u_1 u_2 u_3 + u_2^3)}{u_1^4}, 
\end{gather*} 
\vspace{-1.5em}
\begin{align*}
    \widetilde{\Gamma}^4_{22} = \dfrac{1}{u_1^2 u_2^5}(\, & \, \epsilon_1 u_1^2 u_2^3 u_4 - \epsilon_1 u_2^2(u_1^2 u_3^2 + u_2^4) -  2 \epsilon_2 u_1^2 u_2 u_3( u_2 u_4 -   u_3^2)  \\ \, & \,  - \epsilon_3 u_1^2(u_2 u_4 - u_3^2)(u_2 u_4 - 2 u_3^2) + u_2^6), \end{align*} 
 \begin{gather*}
\widetilde{\Gamma}^4_{23} = \dfrac{\epsilon_1 u_2^2(u_1 u_3 + u_2^2) + \epsilon_2 u_1 u_2(u_2 u_4 - 2 u_3^2) - 2 \epsilon_3 u_1 u_3(u_2 u_4 -  u_3^2) - u_2^4     }{u_1 u_2^4}, \\ \widetilde{\Gamma}^4_{24} = - \dfrac{\epsilon_1 u_2^2 - \epsilon_2 u_2 u_3 - \epsilon_3(u_2 u_4 - u_3^2)}{u_2^3}, \\ \widetilde{\Gamma}^4_{33} = - \dfrac{\epsilon_1 u_2^2 - 2 \epsilon_2 u_2 u_3 - \epsilon_3(u_2 u_4 - 2  u_3^2) }{u_2^3},  \qquad \widetilde{\Gamma}^4_{34} = - \dfrac{\epsilon_2 u_2 - \epsilon_3 u_3}{u_2^2}, 
 \qquad \widetilde{\Gamma}^4_{44} = -\dfrac{\epsilon_3}{u_2}.
\end{gather*}

\underline{The third equivalence}\\
The metric corresponding to the linear function $a_0=\epsilon_1\,u_1$ is represented by the matrix
\begin{equation}
	g=\begin{bmatrix}
		g_{11}&g_{12}&g_{13}&g_{14}\\g_{12}&g_{13}&g_{14}&0\\g_{13}&g_{14}&0&0\\g_{14}&0&0&0
	\end{bmatrix}
\end{equation}
where 
\begin{align}
	g_{11}= \, & \, -\frac{1}{6}C_1\epsilon_1(\epsilon_1-2)(\epsilon_1+2)u_2^{-\frac{3\epsilon_1}{2}-3}u_3^{3}+\frac{1}{2}C_1\epsilon_1(\epsilon_1-2) u_2^{-\frac{3\epsilon_1}{2}-2}u_3 u_4\notag\\
	\, & \, + 2F_1''(u_2) u_3^2+\bigg(4\epsilon_1\frac{u_3^2}{u_2}+3u_4\bigg)F_1'(u_2)+2F_2'(u_2)u_3\notag\\
	\, & \, + \epsilon_1 u_2^{-2}\bigg(2\epsilon_1 u_3^2+4 u_2 u_4-3 u_3^{2}\bigg)F_1(u_2)+2\epsilon_1\frac{u_3}{u_2}F_2(u_2)+F_3(u_2),\notag
 \end{align}
 \begin{align}
	g_{12}= \, & \, \frac{1}{2}C_1\epsilon_1 u_2^{-\frac{3\epsilon_1}{2}-1}\bigg(\epsilon_1\frac{u_3^2}{u_2}-u_4\bigg)+2\big(F_1'(u_2)+\epsilon_1 u_2^{-1}F_1(u^2)\big)u_3+F_2(u_2),
	\notag\\
	g_{13}= \, & \, -C_1\epsilon_1 u_2^{-\frac{3\epsilon_1}{2}-1} u_3+F_1(u_2),\notag\\
	g_{14}= \, & \, C_1 u_2^{-\frac{3\epsilon_1}{2}},
	\notag
\end{align}
for some constant $C_1$ and some functions $F_1$, $F_2$, $F_3$ of a single variable.

\subsubsection{Dimension $4$: $3\times 3+1\times 1$ Jordan blocks}
The operator of multiplication by the Euler vector field is represented by the matrix
\beq
L=\begin{bmatrix}
u_{1} & 0 & 0 & 0\cr
u_{2} & u_{1} & 0 & 0\cr
u_3 & u_2 & u_1 & 0\cr
0 & 0 & 0 & u_{4}
\end{bmatrix}.
\eeq
The general solution  of the  equation \eqref{eq:ddainL} is
\begin{equation}
    a_0=F_1( u_1)u_{{3}}+\frac{1}{2}F'_1(u_1)u_2^{2}+F_2(u_1)u_2+F_3( u_1)+F_4(u_4).\label{ddasol31}
\end{equation}

The non-vanishing Christoffel symbols are 
\begin{equation}
            \begin{gathered}
                \Gamma^1_{11}  = - \Gamma^1_{14} = \Gamma^1_{44}  = - \frac{(u_1-u_4)}{u_2} \, \Gamma^2_{11} = \Gamma^2_{12} =   \frac{(u_1-u_4)}{u_2} \, \Gamma^2_{14} = -\Gamma^2_{24}   = - \frac{(u_1-u_4)}{u_2} \, \Gamma^2_{44}   \\
                =  -\frac{(u_1-u_4)^2}{u_3(u_1-u_4)-u_2^2} \,  \Gamma^3_{11}= - \frac{(u_1-u_4)}{u_2} \, \Gamma^3_{12} =  \Gamma^3_{13} =\frac{(u_1-u_4)^2}{u_3(u_1-u_4)-u_2^2} \, \Gamma^3_{14} \quad \, \, \, \,      \\
                = \frac{(u_1-u_4)}{u_2} \, \Gamma^3_{24}  = -\Gamma^3_{34}   = - \frac{(u_1-u_4)^2}{u_3(u_1-u_4)-u_2^2} \, \Gamma^3_{44}= -\frac{\partial_4 a_0}{u_1-u_4}, \qquad \qquad \qquad \,  \\
               \Gamma^1_{22}   = \frac{1}{u_2} \left(\partial_2 a_0 - \frac{u_3}{u_2} \, \partial_3 a_0 \right), \\ \Gamma^1_{23}    = -\frac{u_2}{u_3} \, \Gamma^2_{23} = \Gamma^2_{33}  = \frac{(u_1-u_4)}{u_2} \, \Gamma^4_{13}= \frac{(u_1-u_4)}{u_2} \,  \Gamma^4_{22}=  -\frac{(u_1-u_4)}{u_2} \, \Gamma^4_{34}= \frac{\partial_3 a_0}{u_2},  \\
                \Gamma^2_{22}  = \frac{1}{u_2}\left(-\partial_1 a_0 + \frac{u_3^2}{u_2^2} \, \partial_3 a_0 \right), \qquad \Gamma^3_{22} = \frac{u_3}{u_2^2}\left(\partial_1 a_0 - \frac{u_3}{u_2} \, \partial_2 a_0\right)-\frac{\partial_4 a_0}{u_1-u_4},\\
                \Gamma^3_{23}   = \frac{1}{u_2}\left(-\partial_1 a_0 + \frac{u_3}{u_2} \, \partial_2 a_0 \right), \qquad \Gamma^3_{33} = -\frac{\partial_2 a_0}{u_2}, \\
                \Gamma^4_{11}  = -\Gamma^4_{14} = \Gamma^4_{44} =\frac{1}{u_1-u_4}\left(\partial_1 a_0 - \frac{u_2}{(u_1-u_4)}\,  \partial_2 a_0 + \frac{u_2^2-u_3(u_1-u_4)}{(u_1-u_4)^2} \, \partial_3 a_0 \right),\\
                \Gamma^4_{12} = -\Gamma^4_{24}= \frac{1}{u_1-u_4} \left(\partial_2 a_0 - \frac{u_2}{u_1-u_4} \, \partial_3 a_0 \right).
            \end{gathered}
            \label{Chr:31}
        \end{equation}

\underline{The first equivalence}\\
In this case, it is sufficient to consider \eqref{eq:3RC}
 for \begin{equation*}
 \begin{aligned}
     (i,j,k,l,m) = \overset{a)}{\overbrace{(4,1,1,4,2)}}, \, \overset{b)}{\overbrace{(2,1,1,4,3)}}, \,  \overset{c)}{\overbrace{(2,1,1,2,3)}}, \\\underset{d)}{\underbrace{(3,4,1,4,2)}}, \,  \underset{e)}{\underbrace{(4,2,1,4,2)}}, \,  \underset{f)}{\underbrace{(3,3,1,2,3)}}.
     \end{aligned}
 \end{equation*}

Inputting \eqref{Chr:31} into \eqref{eq:3RC} together with $c$ as in Table \ref{tab:exC},  we get:
\begin{gather*}
    a): \dfrac{\partial_2 \partial_4 a_0}{u_1-u_4} = 0, \quad b): \dfrac{\partial_3 \partial_4 a_0}{u_2} = 0, \quad 
    c): \dfrac{u_3 \partial_3^2 a_0}{u_2^2} = 0, \quad   d): \dfrac{\partial_2 \partial_3 a_0}{u_1-u_4} = 0, \\ e): \dfrac{(u_4-u_1)\partial_1 \partial_4 a_0 - u_2 \partial_2 \partial_4 a_0}{(u_1-u_4)^2} = 0, \quad  f): \dfrac{u_2(\partial_1 \partial_3 a_0 - \partial_2^2 a_0)-u_3 \partial_2 \partial_3 a_0}{u_2^2} = 0.
\end{gather*}

\begin{itemize}
    \item a): $a_0$ does not have any terms depending on both $u_2$ and $u_4$.
    \end{itemize}
    \begin{itemize}
\item b): $a_0$ does not have any terms depending on both $u_3$ and $u_4$.
\item c): $a_0$ must be linear in $u_3$.
\item d): $a_0$ does not have any terms depending on both $u_2$ and $u_3$.
\item e): After using a), $a_0$ does not have any terms depending on both $u_1$, $u_4$.
\item f): After using d), $\partial_1 \partial_3 a_0 = \partial_2^2 a_0$.
\end{itemize}
Thus, $u_4$ is completely isolated but otherwise arbitrary, the $u_3$ dependence is linear and only coupled with $u_1$, making also $u_2$ coupled only with $u_1$. Finally $a_0$ must be quadratic in $u_2$ with coefficients specified by $[u_1 u_3]$.   This gives the solution \eqref{ddasol31}.

\vspace{1em}

\underline{The second equivalence}\\
Let $a_0 = \epsilon_1\,u_1 + \epsilon_2\,u_2 + \epsilon_3\,u_3 + \epsilon_4\,u_4$. Then, the non-zero Christoffel symbols for the dual structure are given by

\begin{gather*}
    \widetilde{\Gamma}^1_{11} = \dfrac{\epsilon_2 u_1 u_2(u_1 - u_4) + \epsilon_3(u_1 u_3 - u_2^2)(u_1 - u_4) - \epsilon_4 u_1^2 u_4 - u_1^3 + u_1^2 u_4}{u_1^3(u_1-u_4)}, \\
    \widetilde{\Gamma}^1_{12}  = - \dfrac{\epsilon_2 u_1- \epsilon_3 u_2}{u_1^2}, \qquad  \widetilde{\Gamma}^1_{13} = - \dfrac{\epsilon_3}{u_1}, \qquad \widetilde{\Gamma}^1_{14} =  \dfrac{\epsilon_4}{u_1 - u_4}, \\ \widetilde{\Gamma}^1_{22} =  \dfrac{\epsilon_2 u_1 u_2 - \epsilon_3(u_1 u_3 + u_2^2)}{u_1 u_2^2}, \qquad \widetilde{\Gamma}^1_{23} =  \dfrac{\epsilon_3}{u_2}, \qquad \widetilde{\Gamma}^1_{44} =  -\dfrac{\epsilon_4 u_1}{u_4(u_1 - u_4)}, \\
    \widetilde{\Gamma}^2_{11} = \widetilde{\Gamma}^3_{12} =  - \dfrac{u_2(\epsilon_1 (u_1 - u_4)^2 -  \epsilon_4 u_4(2u_1 - u_4) -(u_1 - u_4)^2)}{u_1^2(u_1 - u_4)^2}, \\
    \widetilde{\Gamma}^2_{12}  =   \dfrac{\epsilon_1(u_1^2 -  u_4) - \epsilon_4 u_4 - (u_1 - u_4)}{u_1(u_1 - u_4)}, \qquad \widetilde{\Gamma}^2_{14} = - \dfrac{\epsilon_4 u_2}{(u_1 - u_4)^2}  \qquad\widetilde{\Gamma}^2_{22} =  - \dfrac{\epsilon_1 u_2^2 - \epsilon_3 u_3^2}{u_2^3}, \\
    \widetilde{\Gamma}^2_{23} = -\dfrac{\epsilon_3 u_3}{u_2^2}, \qquad \widetilde{\Gamma}^2_{24} = \widetilde{\Gamma}^3_{34} = \dfrac{\epsilon_4}{u_1 - u_4}, \qquad \widetilde{\Gamma}^2_{33} = \dfrac{\epsilon_3}{u_2}, \qquad  \widetilde{\Gamma}^2_{44} = \dfrac{\epsilon_4 u_2}{(u_1 - u_4)^2},
\end{gather*}\vspace{-1.5em}
\begin{align*}
\widetilde{\Gamma}^3_{11}  = -\dfrac{1}{u_1^3(u_1 - u_4)^3}(& \,\epsilon_1(u_1-u_4)^3(u_1 u_3 - u_2^2) \\ \, & \,  - \epsilon_4 u_4\left((u_1 u_3 - u_2^2)(u_1-u_4)^2 - u_1(u_1 u_3 - u_2^2)(u_1 - u_4) + u_1^2 u_2^2 u_4\right)  \\ \, & \, - (u_1 - u_4)^3(u_1 u_3 - u_2^2)),\\
    \widetilde{\Gamma}^3_{22} = \dfrac{1}{u_1 u_2^3(u_1 - u_4)}(& \, \epsilon_1 u_2(u_1 - u_4)(u_1 u_3 + u_2^2) - \epsilon_2 u_1 u_3^2 (u_1  - u_4) - \epsilon_4 u_2^3 u_4 -u_2^3( u_1 -  u_4)),
\end{align*}\vspace{-1.5em}
\begin{gather*}
  \widetilde{\Gamma}^3_{23} = -\dfrac{\epsilon_1 u_2 - \epsilon_2 u_3}{u_2^2}, \qquad \widetilde{\Gamma}^3_{33} = - \dfrac{\epsilon_2}{u_2}, \qquad  \widetilde{\Gamma}^3_{44} = \dfrac{\epsilon_4(u_3(u_1 - u_4) - u_2^2)}{(u_1 - u_4)^3},
\end{gather*}\vspace{-1.5em}
  \begin{align*}
       \widetilde{\Gamma}^4_{11} = \dfrac{\epsilon_4}{u_1^3 (u_1 - u_4)^3}(& \, \epsilon_1 u_1^2(u_1 - u_4)^2 - \epsilon_2 u_1 u_2(u_1 - u_4)(2u_1 - u_4)  \\ \, & \,  - \epsilon_3 (u_4(u_1 - u_4)(u_1 u_3 - u_2^2) - 2u_1^2(u_1 u_3 - u_2^2) + u_1^2(2u_3 u_4 + u_2^2)) ),
  \end{align*}\vspace{-1.5em}
  \begin{gather*}
      \widetilde{\Gamma}^4_{12} = \dfrac{u_4(\epsilon_2 u_1(u_1 - u_4) - \epsilon_3 u_2(2u_1 - u_4))}{u_1^2 (u_1 - u_4)^2}, \qquad \widetilde{\Gamma}^4_{13} = \widetilde{\Gamma}^4_{22}  = \dfrac{\epsilon_3 u_4}{u_1(u_1 - u_4)}, \\ \widetilde{\Gamma}^4_{14} = -\dfrac{\epsilon_1(u_1- u_4)^2 - \epsilon_2 u_2 (u_1 - u_4) - \epsilon_3(u_1 u_3 - u_2^2 - u_3 u_4)}{(u_1 - u_4)^3}, \\ \widetilde{\Gamma}^4_{24} = -\dfrac{\epsilon_2(u_1 - u_4) - \epsilon_3 u_2}{(u_1 - u_4)^2}, \qquad \widetilde{\Gamma}^4_{34} = - \dfrac{\epsilon_3}{u_1 - u_4}, \\ 
          \widetilde{\Gamma}^4_{44}   = \dfrac{\epsilon_1 u_1(u_1 - u_4)^2 - \epsilon_2 u_2 u_4(u_1 - u_4)  - \epsilon_3 u_4(u_3(u_1 - u_4) - u_2^2)  - (u_1 - u_4)^3}{u_4(u_1 - u_4)^3}.
     \end{gather*}

\underline{The third equivalence}\\
The metric corresponding to the linear function $a_0=\epsilon_1\,u_1+\epsilon_4\,u_4$ is represented by the matrix
\begin{equation}
	g=\begin{bmatrix}
		g_{11}&g_{12}&g_{13}&0\\g_{12}&g_{13}&0&0\\g_{13}&0&0&0\\0&0&0&g_{44}
	\end{bmatrix},
\end{equation}
where 
\begin{align}
	g_{11} = \, & \, (u_1-u_4)^{-2\epsilon_4}\Big(\tfrac{2}{9}C_1\epsilon_1\big(\epsilon_1-\tfrac{3}{2}\big) u_2^{-\frac{4}{3}\epsilon_1-2} u_3^2 + 2F_1'(u_2)u_3+2\epsilon_1\frac{u_3}{u_2}F_1(u_2)+F_2(u_2)\Big)\notag\\
	\, & \, + (u_1-u_4)^{-2\epsilon_4-1}\Big(\tfrac{4}{3}C_1\epsilon_4\big(\epsilon_1-\tfrac{3}{2}\big) u_2^{-\frac{4}{3}\epsilon_1}u_3-2\epsilon_4 u_2 F_1(u_2)\Big)\notag\\
	\, & \, + (u_1-u_4)^{-2\epsilon_4-2}\Big(C_1\epsilon_4(2\epsilon_4+1)(u_2)^{-\frac{4}{3}\epsilon_1+2}\Big),
	\notag\\
	g_{12} = \, & \, -2C_1\epsilon_4 u_2^{-\frac{4}{3}\epsilon_1+1}(u_1-u_4)^{-2\epsilon_4-1}\notag\\ \, & \, - \frac{2}{3}C_1\epsilon_1 u_2^{-\frac{4}{3}\epsilon_1-1} u_3 (u_1-u_4)^{-2\epsilon_4}+(u_1-u_4)^{-2\epsilon_4}F_1(u_2),
	\notag\\
	g_{13}= \, & \, C_1 u_2^{-\frac{4\epsilon_1}{3}}(u_1-u_4)^{-2\epsilon_4},\notag\\
	g_{44}= \, & \, C_2(u_1-u_4)^{-2\epsilon_1},
	\notag
\end{align}
for some constants $C_1$, $C_2$ and some functions $F_1$, $F_2$ of a single variable.

\subsubsection{Dimension $4$: $2\times 2+2\times 2$ Jordan blocks}
The operator of multiplication by the Euler vector field is represented by the matrix
\beq
L=\begin{bmatrix}
u_{1} & 0 & 0 & 0\cr
u_{2} & u_{1} & 0 & 0\cr
0 & 0 & u^3 & 0\cr
0 & 0 & u_{4} & u_{3}
\end{bmatrix}.
\eeq
The general solution  of the  equation \eqref{eq:ddainL} is
\begin{equation}
    a_0=F_1(u_1)u_2+F_2(u_1)+F_3(u_3)u_4+F_4(u_3).\label{ddasol22}
\end{equation}

The non-vanishing Christoffel symbols are given by 
  \begin{equation}
            \begin{gathered}
                \Gamma^1_{11}  = -\Gamma^1_{13} = \Gamma^1_{33} = \Gamma^2_{12} = - \Gamma^2_{23} = - \frac{1}{u_1-u_3}\left(\partial_3 a_0+\frac{u_4}{u_1-u_3}\partial_4 a_0 \right),\\
             \Gamma^1_{14}  = - \Gamma^1_{34} = -\frac{(u_1-u_3)}{u_2} \Gamma^2_{14} = \Gamma^2_{24} = \frac{(u_1-u_3)}{u_2}\Gamma^2_{34}  = \frac{u_4}{u_1-u_4}\Gamma^3_{44}=\frac{\partial_4 a_0}{u_1-u_3}, \\
                 \Gamma^1_{22}   = \frac{u_1-u_3}{u_2} \Gamma^3_{12} =- \frac{u_1-u_3}{u_2} \Gamma^3_{23}  = \frac{(u_1-u_3)^2}{u_2 u_4} \, \Gamma^4_{12} =  - \frac{(u_1-u_3)^2}{u_2 u_4} \, \Gamma^4_{23} = - \frac{u_1-u_3}{u_2} \,  \Gamma^4_{24} \\ 
                = \frac{\partial_2 a_0}{u_2}, \qquad
               \Gamma^2_{11}  = -\Gamma^2_{13} = \Gamma^2_{33}= \frac{u_2}{(u_1-u_3)^2}\left(\partial_3 a_0 + \frac{2 u_4}{u_1-u_3} \, \partial_4 a_0\right), \qquad \qquad  \, \\ \Gamma^2_{22} = -\frac{\partial_1 a_0}{u_2}, \qquad
                \Gamma^3_{11}   = -\Gamma^3_{13} = \Gamma^3_{33} = - \Gamma^4_{14} =  \Gamma^4_{34} = \frac{1}{u_1-u_3}\left(\partial_1 a_0 - \frac{u_2}{u_1-u_3} \, \partial_2 a_0\right), \\
                \Gamma^4_{11}   = -\Gamma^4_{13}= \Gamma^4_{33} = \frac{u_4}{(u_1-u_3)^2}\left(\partial_1 a_0 - \frac{2 u_2}{u_1-u_3} \,  \partial_2 a_0\right), \qquad \Gamma^4_{44} = - \frac{\partial_3 a_0}{u_4}.
            \end{gathered}
            \label{Chr:22}
        \end{equation}
        
\underline{The first equivalence}\\
In this case, it is sufficient to consider \eqref{eq:3RC} for
 \begin{equation*}
 \begin{aligned}
     (i,j,k,l,m) = \overset{a)}{\overbrace{(2,1,1,4,2)}}, \, \overset{b)}{\overbrace{(2,1,1,3,2)}}, \,  \overset{c)}{\overbrace{(2,2,1,4,3)}}, \\\underset{d)}{\underbrace{(2,2,1,2,4)}}, \,  \underset{e)}{\underbrace{(2,2,1,2,3)}}, \,  \underset{f)}{\underbrace{(4,4,3,2,1)}}.
     \end{aligned}
 \end{equation*}

Inputting \eqref{Chr:22} into \eqref{eq:3RC} together with $c$ as in Table \ref{tab:exC}, we get the following.
\begin{gather*}
    a): \dfrac{\partial_2 \partial_4 a_0}{u_2} = 0, \quad b): \dfrac{\partial_2 \partial_3 a_0}{u_2} = 0, \quad 
    c): \dfrac{u_4 \partial_4^2 a_0}{(u_1-u_3)^2} = 0, \quad   d): \dfrac{\partial_2 \partial_4 a_0}{u_1-u_3} + \dfrac{\partial_1 \partial_4 a_0}{u_2} = 0, \\ e): \dfrac{u_2 u_4 \partial_2 \partial_4 a_0 + (u_1-u_3)(u_2 \partial_2 \partial_3 a_0 + (u_1-u_3)\partial_1 \partial_3 a_0)}{u_2(u_1-u_3)^2} = 0, \quad  f): \dfrac{u_2 \partial_2^2 a_0}{(u_1-u_3)^2} = 0.
\end{gather*}
\begin{itemize}
    \item a): $a_0$ does not have any terms depending on both $u_2$ and $u_4$.
\item b): $a_0$ does not have any terms depending on both $u_2$ and $u_3$.
\item c): $a_0$ must be linear in $u_4$.
\item d): After using a) $a_0$ does not have any terms depending on both $u_1$ and $u_4$.
\item e): After using a), and b), $a_0$ does not have any terms depending on both $u_1$, $u_3$.
\item f): $a_0$ must be linear in $u_2$.
\end{itemize}
From this we have that we have mixing between $u_3$ and $u_4$, and $u_1$ and $u_2$ separately. Moreover $a_0$ must be linear in both $u_2$ and $u_4$, giving precisely  the solution \eqref{ddasol22}. 

\vspace{0.5em}

\underline{The second equivalence}\\
Let $a_0 = \epsilon_1\,u_1 + \epsilon_2\,u_2 + \epsilon_3\,u_3 + \epsilon_4\,u_4$. Then, the non-zero Christoffel symbols for the dual structure are given by
\begin{gather*}
     \widetilde{\Gamma}^1_{11} = \dfrac{(\epsilon_2 u_2 - 1)(u_1 - u_3)^2  - \epsilon_3 u_1 u_3(u_1 - u_3) - \epsilon_4 u_1^2 u_4}{u_1^2(u_1 - u_3)^2},  \qquad  \widetilde{\Gamma}^1_{12} = -\dfrac{\epsilon_2}{u_1}, \\    \widetilde{\Gamma}^1_{13} = \widetilde{\Gamma}^2_{23} = \dfrac{\epsilon_3(u_1 - u_3) + \epsilon_4 u_4}{(u_1 - u_3)^2}, \qquad    \widetilde{\Gamma}^1_{14} = \widetilde{\Gamma}^2_{24} = \dfrac{\epsilon_4}{u_1 - u_3}, \qquad \widetilde{\Gamma}^1_{22} = \dfrac{\epsilon_2}{u_2}, \\   \widetilde{\Gamma}^1_{33} = -\dfrac{u_1(\epsilon_3 u_3(u_1 - u_3) - \epsilon_4 u_4 (u_1 - 2 u_3))}{u_3^2(u_1 - u_3)^2}, \qquad    \widetilde{\Gamma}^1_{34} = - \dfrac{\epsilon_4 u_1}{u_3(u_1 - u_3)}, \\
      \widetilde{\Gamma}^2_{11} = -  \dfrac{u_2}{u_1^2(u_1-u_3)^3}((\epsilon_1 - 1)(u_1 - u_3)^3  - \epsilon_3 u_3((u_1 - u_3)^2 + u_1(u_1 - u_3))  - 2 \epsilon_4 u_1^2 u_4),\\
     \widetilde{\Gamma}^2_{12} = \dfrac{(\epsilon_1 - 1)(u_1 - u_3)^2 - \epsilon_3u_3(u_1 - u_3) - \epsilon_4 u_1 u_4}{u_1(u_1 - u_3)^2}, \\ \widetilde{\Gamma}^2_{13}  = - \widetilde{\Gamma}^2_{33} = -\dfrac{u_2(\epsilon_3(u_1 - u_3) + 2 \epsilon_4 u_4)}{(u_1 - u_3)^3}, \qquad \widetilde{\Gamma}^2_{14} = - \widetilde{\Gamma}^2_{34} = - \dfrac{\epsilon_4 u_2}{(u_1 - u_3)^2},  \\
     \widetilde{\Gamma}^2_{22} = - \dfrac{\epsilon_1}{u_2}, \qquad    \widetilde{\Gamma}^3_{11} = \dfrac{u_3(\epsilon_1 u_1(u_1 - u_3) - \epsilon_2 u_2 (2u_1 - u_3))}{u_1^2(u_1 - u_3)^2}, \qquad \widetilde{\Gamma}^3_{12} = \dfrac{\epsilon_2 u_3}{u_1(u_1 - u_3)}, \\ \widetilde{\Gamma}^3_{13} = -\dfrac{\epsilon_1(u_1 - u_3) - \epsilon_2 u_2}{(u_1 - u_3)^2}, \qquad \widetilde{\Gamma}^3_{23} =  \widetilde{\Gamma}^4_{24} = - \dfrac{\epsilon_2}{u_1-u_3}, \\
     \widetilde{\Gamma}^3_{33} =  \dfrac{\epsilon_1 u_1 u_3(u_1 - u_3) - \epsilon_2 u_2 u_3^2 + (\epsilon_4 u_4 - u_3)(u_1 - u_3)^2}{u_3^2(u_1-u_3)^2},\\ \widetilde{\Gamma}^3_{34} = - \dfrac{\epsilon_4}{u_3},  \qquad \widetilde{\Gamma}^3_{44} =  \dfrac{\epsilon_4}{u_4},  \qquad \widetilde{\Gamma}^4_{11} = - \widetilde{\Gamma}^4_{13} = - \dfrac{\epsilon_4(\epsilon_1(u_1 - u_3) - 2 \epsilon_2 u_2)}{(u_1-u_3)^3},\\
     \widetilde{\Gamma}^4_{12}  = - \widetilde{\Gamma}^4_{23} =  \dfrac{\epsilon_2 u_4}{(u_1-u_3)^2}, \qquad \widetilde{\Gamma}^4_{14} = - \dfrac{\epsilon_1(u_1 -  u_3) - \epsilon_2 u_2}{(u_1-u_3)^2},
\end{gather*}
\begin{gather*}
     \widetilde{\Gamma}^4_{33} = - \dfrac{u_4 \left(\epsilon_1 u_1 ((u_1 - u_3)^2 - u_3(u_1 - u_3)) + \epsilon_2 u_2 u_3^2 + (\epsilon_3 - 1)(u_1 - u_3)^3   \right)}{u_3^2(u_1 - u_3)^3}, \\
         \widetilde{\Gamma}^4_{44} = - \dfrac{\epsilon_3}{u_4}.
\end{gather*}

\underline{The third equivalence}\\
The metric corresponding to the linear function $a_0=\epsilon_1\,u_1+\epsilon_3\,u_3$ is represented by the matrix
\begin{equation}
	g=\begin{bmatrix}
		g_{11}&g_{12}&0&0\\g_{12}&0&0&0\\0&0&g_{33}&g_{34}\\0&0&g_{34}&0
	\end{bmatrix},
\end{equation}
where
\begin{align}
	g_{11}=\, & \, (u_3-u_1)^{-2\epsilon_3}\bigg(F_1(u_2)+C_1\frac{u_2^{1-\epsilon_1}}{u_3-u_1}\bigg),
	\notag\\
	g_{12}= \, & \, \frac{C_1}{2\epsilon_3} u_2^{-\epsilon_1}(u_3-u_1)^{-2\epsilon_3},
	\notag\\
	g_{33}= \, & \, (u_3-u_1)^{-2\epsilon_1}\bigg(F_2(u_4)+C_2\frac{u_4^{1-\epsilon_3}}{u_3-u_1}\bigg),\notag\\
	g_{34}=\, & \, -\frac{C_2}{2\epsilon_1} u_4^{-\epsilon_3}(u_3-u_1)^{-2\epsilon_1},
	\notag
\end{align}
for some constants $C_1$, $C_2$ and some functions $F_1$, $F_2$ of a single variable.

\subsubsection{Dimension $4$: $2\times 2+1\times 1+1\times 1$ Jordan blocks}
The operator of multiplication by the Euler vector field is represented by the matrix
\beq
L=\begin{bmatrix}
u_{1} & 0 & 0 & 0\cr
u_{2} & u_{1} & 0 & 0\cr
0 & 0 & u_3 & 0\cr
0 & 0 & 0 & u_{4}
\end{bmatrix}.
\eeq
The general solution  of the  equation \eqref{eq:ddainL} is
\begin{equation}
    a_0=F_1(u_1)u_2+F_2(u_1)+F_3(u_3)+F_4(u_4).\label{ddasol211}
\end{equation}

The non-vanishing Christoffel symbols are 
 \begin{equation}
            \begin{gathered}
             \Gamma^1_{11}  = \Gamma^2_{12} =-\left(\frac{\partial_3 a_0}{u_1-u_3} + \frac{\partial_4 a_0}{u_1-u_4} \right), \\
                 \Gamma^1_{13}  = -\Gamma^1_{33} = \Gamma^2_{23} =- \frac{u_1-u_3}{u_2} \, \Gamma^2_{13} =\frac{u_1-u_3}{u_2} \, \Gamma^2_{33} =  \frac{u_3-u_4}{u_1-u_3} \, \Gamma^4_{33} = -\frac{u_3-u_4}{u_1-u_3} \, \Gamma^4_{34} \\
                   = \frac{\partial_3 a_0}{u_1-u_3}, \hspace{26em} \, \,  \\
                    \Gamma^1_{14}  = - \Gamma^1_{44} = \Gamma^2_{24} = -\frac{u_1-u_4}{u_2} \, \Gamma^2_{14} = \frac{u_1-u_4}{u_2} \, \Gamma^2_{44}=\frac{u_3-u_4}{u_1-u_4} \, \Gamma^3_{34}=-\frac{u_3-u_4}{u_1-u_4} \, \Gamma^3_{44}\\
                 = \frac{\partial_4 a_0}{u_1-u_4}, \hspace{26em} \, \, 
                \end{gathered}
                \label{Chr:211}
 \end{equation}
\begin{equation*}
	\begin{gathered}
		\Gamma^1_{22}   = \frac{u_1-u_3}{u_2} \, \Gamma^3_{12} =  -\frac{u_1-u_3}{u_2} \, \Gamma^3_{23} = \frac{u_1-u_4}{u_2} \, \Gamma^4_{12} = - \frac{u_1-u_4}{u_2} \, \Gamma^4_{24} = \frac{\partial_2 a_0}{u_2},\\
		\Gamma^2_{11}   = u_2\left(\frac{\partial_3 a_0}{(u_1-u_3)^2} + \frac{\partial_4 a_0}{(u_1-u_4)^2}\right), \qquad \Gamma^2_{22} = - \frac{\partial_1 a_0}{u_2},
	\end{gathered}
\end{equation*}
 \begin{equation*}
                 \begin{gathered}
                \Gamma^3_{11}  = -\Gamma^3_{13} = \frac{1}{u_1-u_3}\left(\partial_1 a_0 - \frac{u_2}{u_1-u_3} \, \partial_2 a_0\right),\\
                \Gamma^3_{33}  = \left(\frac{1}{u_1-u_3}-\frac{1}{u_3-u_4} \right)\partial_1 a_0 - \frac{u_2}{(u_1-u_3)^2} \,  \partial_2 a_0, \\
                \Gamma^4_{11}  = - \Gamma^4_{14} = \frac{1}{u_1-u_4}\left(\partial_1 a_0 - \frac{u_2}{u_1-u_4} \, \partial_2 a_0\right),\\
                \Gamma^4_{44}   = \frac{1}{u_1-u_4}\left(\partial_1 a_0 - \frac{u_2}{u_1-u_4} \, \partial_2 a_0\right) + \frac{1}{u_3-u_4} \, \partial_3 a_0.
            \end{gathered}
        \end{equation*}

\underline{The first equivalence}\\
In this case, it is sufficient to consider \eqref{eq:3RC}
 for \begin{equation*}
 \begin{aligned}
     (i,j,k,l,m) = \overset{a)}{\overbrace{(2,1,1,3,2)}}, \, \overset{b)}{\overbrace{(2,1,1,4,2)}}, \,  \overset{c)}{\overbrace{(1,1,1,4,3)}}, \\\underset{d)}{\underbrace{(3,4,1,3,4)}}, \,  \underset{e)}{\underbrace{(4,3,4,1,3)}}, \,  \underset{f)}{\underbrace{(4,4,4,2,1)}}.
     \end{aligned}
 \end{equation*} 

Inputting \eqref{Chr:211} into \eqref{eq:3RC} together with $c$ as in Table \ref{tab:exC},  we get the following.

\begin{gather*}
    a): \dfrac{\partial_2 \partial_3 a_0}{u_2} = 0, \quad b): \dfrac{\partial_2 \partial_4 a_0}{u_2} = 0, \quad 
    c): \dfrac{(u_3-u_4) \partial_3 \partial_4 a_0}{(u_1-u_3)(u_1-u_4)} = 0, \quad   d): \dfrac{\partial_1 \partial_3 a_0}{u_3-u_4} = 0, \\ e): \dfrac{\partial_1 \partial_4 a_0}{u_3-u_4} = 0, \quad  f): \dfrac{u_2 \partial_2^2 a_0}{(u_1-u_4)^2} = 0.
\end{gather*}

\begin{itemize}
    \item a) - e): $a_0$ does not have any terms depending on both $\{u_2, u_3\}$, $\{u_2, u_4\}$, $\{u_3, u_4\}$, $\{u_1, u_3\}$, $\{u_1, u_4\}$ (i.e. we only have mixing between $u_1$ and $u_2$) .
\item f): $a_0$ must be linear in $u_2$.
\end{itemize}
This gives precisely  the solution \eqref{ddasol211}.

\vspace{1em}

\underline{The second equivalence}\\
Let $a_0 = \epsilon_1\,u_1 + \epsilon_2\,u_2 + \epsilon_3\,u_3 + \epsilon_4\,u_4$. Then, the non-zero Christoffel symbols for the dual structure are given by
\begin{align*}
    \widetilde{\Gamma}^1_{11} = \dfrac{1}{u_1^2(u_1 - u_4)(u_1 - u_3)}(& \, \epsilon_2 u_2(u_1 - u_3)(u_1 - u_4) - \epsilon_3 u_1 u_3(u_1 - u_4) \\ \, & \, - \epsilon_4 u_1 u_4 (u_1 - u_3) - u_1(u_1 - u_3)(u_1 - u_4),
    \end{align*}
\begin{gather*}
    \widetilde{\Gamma}^1_{12} = - \dfrac{\epsilon_2}{u_1}, \qquad \widetilde{\Gamma}^1_{13} = \widetilde{\Gamma}^2_{23}  = \dfrac{\epsilon_3}{u_1 - u_3}, \qquad \widetilde{\Gamma}^1_{14} = \widetilde{\Gamma}^2_{24} = \dfrac{\epsilon_4}{u_1 - u_4}, \qquad  \widetilde{\Gamma}^1_{22} = \dfrac{\epsilon_2}{u_2}, \end{gather*}
    \begin{gather*}
    \widetilde{\Gamma}^1_{33} = - \dfrac{\epsilon_3 u_1}{u_3(u_1 - u_3)}, \qquad \widetilde{\Gamma}^1_{44} = \dfrac{\epsilon_4 u_1}{u_4(u_1 - u_4)},
\end{gather*}
\begin{align*}
    \widetilde{\Gamma}^2_{11} = - \dfrac{u_2}{u_1^2(u_1 - u_3)^2(u_1 - u_4)^2}(& \, (\epsilon_1 - 1) (u_1 - u_3)^2(u_1 - u_4)^2\\ \, & \,  - \epsilon_3 u_3 (2u_1(u_1 - u_4)^2 - u_3(u_1 - u_4)) \\ \, & \,  + \epsilon_4 u_4 (2u_1(u_1 - u_3)^2 - u_4(u_1 - u_3)) ),
\end{align*}
\begin{align*}
    \widetilde{\Gamma}^2_{12} = \dfrac{1}{u_1(u_1 - u_3)(u_1 - u_4)}(& \, \epsilon_1 u_1^2 - \epsilon_1 u_1(u_3 + u_4) + \epsilon_1 u_3 u_4 - \epsilon_3 u_3 (u_1 - u_4) \\ \, & \, - \epsilon_4 u_4 (u_1 - u_3) - u_1^2 + u_1(u_3 + u_4) - u_3 u_4) ,
\end{align*}
\begin{gather*}
    \widetilde{\Gamma}^2_{22} = -\dfrac{\epsilon_1}{u_2}, \qquad \widetilde{\Gamma}^2_{33} = \dfrac{\epsilon_3 u_2}{(u_1 - u_3)^2}, \qquad  \widetilde{\Gamma}^3_{11} = \dfrac{u_3(\epsilon_1 u_1(u_1 - u_3) - \epsilon_2u_2(2 u_1 - u_3))}{u_1^2(u_1 - u_3)^2}, \\ \widetilde{\Gamma}^3_{12} = \dfrac{\epsilon_2 u_3}{u_1(u_1 - u_3)}, \qquad \widetilde{\Gamma}^3_{13} = -\dfrac{\epsilon_1(u_1 - u_3) - \epsilon_2 u_2}{(u_1 -u_3)^2}, \qquad \widetilde{\Gamma}^3_{23} = - \dfrac{\epsilon_2}{u_1 - u_3},
\end{gather*}
\begin{align*}
    \widetilde{\Gamma}^3_{33} = \dfrac{1}{u_3(u_1 - u_3)^2(u_1 - u_4)}(&\, \epsilon_1 u_1 (u_1 - u_3)(u_3 - u_4) - \epsilon_2 u_2 u_3(u_3 - u_4) \\ \, & \,  - \epsilon_4 u_4(u_1 - u_3)^2 - (u_1 - u_3)^2(u_3 - u_4),
\end{align*}
\begin{gather*}
    \widetilde{\Gamma}^3_{34} = \dfrac{\epsilon_4}{u_3 - u_4}, \qquad \widetilde{\Gamma}^3_{44} = - \dfrac{u_3 \epsilon_4}{u_4(u_3 - u_4)}, \\ \widetilde{\Gamma}^4_{11} = - \dfrac{u_4(\epsilon_1 u_1(u_1 - u_4) - \epsilon_2u_2 (2u_1 - u_4))}{u_1^2(u_1 - u_4)^2}, \qquad \widetilde{\Gamma}^4_{12} = \dfrac{\epsilon_2 u_4}{u_1(u_1-u_4)}, \\ \widetilde{\Gamma}^4_{14} = - \dfrac{\epsilon_1(u_1 - u_4) - \epsilon_2 u_2}{(u_1 - u_4)^2}, \qquad \widetilde{\Gamma}^4_{24} = - \dfrac{\epsilon_2}{u_1 - u_4}, \qquad \widetilde{\Gamma}^4_{33} = \dfrac{\epsilon_3 u_4}{u_3(u_3- u_4)}, \\ \widetilde{\Gamma}^4_{34} = - \dfrac{\epsilon_3}{u_3 - u_4},
\end{gather*}
\begin{align*}
    \widetilde{\Gamma}^4_{44} = \dfrac{1}{u_4(u_1 - u_4)^2(u_3 - u_4)}(& \, \epsilon_1 u_1(u_3 - u_4)(u_1 - u_4)   - \epsilon_2 u_2 u_4(u_3 - u_4) \,  + \\ \, & \,   \epsilon_3 u_3(u_1 - u_4)^2 - (u_1^2 + u_4^2)(u_3 - u_4)  + 2 u_1 u_4(u_2 - u_4) ).
\end{align*}

\underline{The third equivalence}\\
The metric corresponding to the linear function $a_0=\epsilon_1\,u_1+\epsilon_3\,u_3+\epsilon_4\,u_4$ is represented by the matrix
\begin{equation}
	g=\begin{bmatrix}
		g_{11}&g_{12}&0&0\\g_{12}&0&0&0\\0&0&g_{33}&0\\0&0&0&g_{44}
	\end{bmatrix},
\end{equation}
where
\begin{align}
	g_{11}= \, & \, (u_1-u_3)^{-2\epsilon_3}(u_4-u_1)^{-2\epsilon_4}\bigg(F_1(u_2)-C_1\frac{\epsilon_3}{\epsilon_4}\frac{u_2^{1-\epsilon_1}}{u_1-u_3}+C_1\frac{u_2^{1-\epsilon_1}}{u_4-u_1}\frac{}{}\bigg),
	\notag\\
	g_{12}= \, & \, \frac{C_1}{2\epsilon_4} u_2^{-\epsilon_1}(u_1-u_3)^{-2\epsilon_3}(u_4-u_1)^{-2\epsilon_4},
	\notag\\
	g_{33}=\, & \, C_2(u_3-u_4)^{-2\epsilon_4}(u_1-u_3)^{-2\epsilon_1},\notag\\
	g_{44}= \, & \, C_3(u_3-u_4)^{-2\epsilon_3}(u_1-u_4)^{-2\epsilon_1},
	\notag
\end{align}
for some constants $C_1$, $C_2$, $C_3$ and some function $F_1$ of a single variable.

\section{Some conjectures}
Let us summarise the results of  the previous section. We started from regular F-manifolds with Euler vector field $(M,\circ,e,E)$ in dimensions $2,3,$ and $4$ and we constructed the torsionless connection $\nabla$ uniquely determined by imposing the conditions
\begin{itemize}
\item $\nabla e=0$;
\item $d_{\nabla}[(E-a_0e)\,\circ]=0$, where $a_0$ is an arbitrary function.
\end{itemize}
We found that
\begin{enumerate}
\item The data $(\nabla,\circ,e)$ define on $M$ the structure of an F-manifold with compatible connection and flat unit vector field if and only if  $a_0$ is a solution  of the equation \eqref{eq:main}. In one direction, i.e. from solutions of \eqref{eq:main} to F-manifolds with compatible connection, the result follows from Theorem \ref{mainTh} and Proposition \ref{LinIn}.
\item The data $(\nabla,\circ,e,\nabla^*,*,E)$, where $*$ is the dual product and the $\nabla^*$ is given by \eqref{dualfromnatural}, define a bi-flat structure if and only if $a_0$ is a linear function (notice that any linear function satisfies \eqref{eq:main}). In general, due to  Theorem \ref{lin=flat} we already know that the data $(\nabla,\circ,e)$ define a flat structure if and only if $a_0$ is a linear function. Moreover, due to the results of \cite{LP23}, the full statement is true in the case of $r$ Jordan blocks of  sizes $m_1,...,m_r$ for special linear solutions containing only the main variables of each block
\begin{equation*}
		a_0=\overset{r}{\underset{\alpha=1}{\sum}}\,m_\alpha\varepsilon_{\alpha}u^{1(\alpha)}=\overset{r}{\underset{\alpha=1}{\sum}}\,m_\alpha\varepsilon_{\alpha}u^{m_0+m_1+\dots+m_{\alpha-1}+1}.
	\end{equation*}
\item Computations with Maple suggest that, in the case of flat structures obtained from linear function $a_0$, the system 
\begin{equation*}
		\big(\nabla_X g\big)(Y,Z)=\frac{1}{2}\,d\theta(X\circ Y,Z)+\frac{1}{2}\,d\theta(X\circ Z,Y),
	\end{equation*}
for the metric $g$ defining the associated Riemannian F-manifolds	with Killing unit vector field admits non-degenerate solutions if and only if $a_0$ contains only the main variables of the blocks. This result, if confirmed in arbitrary dimension, would allow to ``distinguish'' the special bi-flat structures studied in \cite{LP23}
 (\emph{Lauricella bi-flat structures}) from those obtained from generic linear functions $a_0$. In the semisimple setting all the variables are on the same footing.  In this case the existence of non-degenerate metrics satisfying the above system follows from the results of  \cite{ABLR}.
\end{enumerate} 
Taking into account the comments above, we conjecture that the same results hold true in arbitrary dimension. This is summarized in the table below.

\begin{table}[h!]
\begin{center}

{
\begin{tabular}{ |c|c|c|  }
\hline
 $d\cdot d_La_0=0$   &  Geometric structures on F-manifolds  & Conjectures \\
\hline
\hline
Generic solutions &  Compatible  connection  &  $\eqref{shc-intro}\Rightarrow d\cdot d_La_0=0$? \\
\hline
Linear functions &  Compatible flat connection  & Bi-flat? \\
\hline
$a_0$ given by \eqref{sls-intro} & Bi-flat structure \cite{LP23} & Riemannian? \\
\hline
\end{tabular}
}

\end{center}
\caption{Summary of results and conjectures.}
\label{tab:Summary}
\end{table}

\end{document}